\documentclass[a4paper,11pt]{article}
\usepackage{jheppub}
\usepackage{lineno}
\usepackage{amsthm}
\usepackage[minimal,nowidering]{yhmath}
\usepackage{tikz}
\usetikzlibrary{arrows.meta}
\usetikzlibrary{patterns}
\usepackage{mathrsfs}
\usepackage{booktabs}
\usepackage{array}
\usepackage{placeins}
\makeatletter
\gdef\@fpheader{Prepared for submission to JHEP}
\makeatother

\newtheorem{proposition}{Proposition}[section]
\newtheorem{conjecture}[proposition]{Conjecture}
\newtheorem{lemma}[proposition]{Lemma}
\newtheorem{theorem}[proposition]{Theorem}
\newtheorem{corollary}[proposition]{Corollary}
\theoremstyle{definition}
\newtheorem{definition}[proposition]{Definition}
\theoremstyle{remark}
\newtheorem{remark}[proposition]{Remark}
\theoremstyle{plain}

\newcommand{\CS}[2]{\text{CS}_{#1}[#2]}
\newcommand{\DCS}[2]{\CS{#1}{#2}\times\CS{-#1}{#2}}
\newcommand{\Ad}{\text{Ad}}
\newcommand{\Irr}{\mathrm{Irr}}

\newcommand{\be}{\begin{equation}}
\newcommand{\ee}{\end{equation}}
\newcommand{\ba}{\begin{aligned}}
\newcommand{\ea}{\end{aligned}}
\newcommand{\bea}{\begin{eqnarray}}
\newcommand{\eea}{\end{eqnarray}}

\def\bp{\begin{pmatrix}}
\def\ep{\end{pmatrix}}

\title{ Large-Level Limit of Doubled Chern-Simons Theory Part I: Regular Sector, Boundaries and Interfaces}

\author[b]{Shengyi Guan}
\author[c]{Qiang Jia}
\author[a]{Jiahua Tian}
\author[b,d]{Yi-Nan Wang}
\affiliation[a]{School of Physics, East China Normal University, \\
Shanghai, 200241, China}
\affiliation[b]{School of Physics, Peking University, \\
Beijing 100871, China}
\affiliation[c]{Department of Physics, Korea Advanced Institute of Science Technology, \\
Daejeon 34141, Korea}
\affiliation[d]{Center for High Energy Physics, Peking University, \\
Beijing, 100871, China}
\emailAdd{guanshengyi12345@gmail.com}\emailAdd{qjia1993@kaist.ac.kr}\emailAdd{jtian1905@gmail.com}\emailAdd{ynwang@pku.edu.cn}

\abstract{We study three-dimensional $BF$ theory as the large-level limit of the doubled Chern-Simons theory $G_k\times G_{-k}$ for simply connected, compact, simple $G$, a limit that connects the symmetry topological field theories of finite and continuous symmetries. Lines whose two highest weights grow linearly with $k$ at fixed difference become $BF$ lines, with magnetic flux set by their average and electric charge equal to their difference. Representing level-$k$ torus states as normalized flux-charge wave packets, we prove that every fixed modular transformation of any fixed normalizable state with regular flux converges in the strong-$*$ topology to its $BF$ counterpart. The Dehn twist gives each line exactly its Aharonov-Bohm phase at every level, while $S$ converges to a Fourier transform exchanging flux and charge, without escaping to singular holonomies or large charges. States that approach singular fluxes at large $k$ are not covered. The natural boundary 2d theory is speculated to be a decompactified boson in the abelian case $G=U(1)$, and $\dim G$ non-compact free scalars for general $G$. For $SU(2)$, we discuss the gapped interfaces between adjacent levels, which must shift some flux or charge by a finite amount. In the $k\rightarrow\infty$ limit the interface is conjectured to be the $c=1$ Runkel-Watts theory.}

\begin{document}
\maketitle
\raggedbottom
\emergencystretch=1em

% section_01.tex
\section{Introduction}
\label{sec:introduction}

The relation between 3d $BF$ theory and doubled Chern--Simons theory  shall be dated to the TQFT formulation of 3d gravity~\cite{Achucarro:1986uwr,Witten:1988hc}. At the classical level, after assigning
\begin{equation}
 A_\pm=A\pm\frac{B}{2k}\,,
 \label{eq:introduction-classical-scaling-revised}
\end{equation}
on a closed three-manifold the doubled Chern--Simons action of $\DCS{k}{G}$ can be rewritten as
\begin{equation}
 \frac{k}{4\pi}
 \bigl(\operatorname{CS}(A_+)-\operatorname{CS}(A_-)\bigr)
 =
 -\frac{1}{2\pi}\int_M\operatorname{tr}(B\wedge F_A)
 -\frac{1}{48\pi k^2}
   \int_M\operatorname{tr}(B\wedge[B,B])\,,
 \label{eq:intro-action-contraction}
\end{equation}
where $A$ and $B$ take values in the compact Lie algebra $\mathfrak g$ of $G$, $F_A=dA+\frac12[A,A]$, and
$\operatorname{CS}(A)=-\int_M\operatorname{tr}\bigl(A\wedge dA
 +\frac13 A\wedge[A,A]\bigr)$. 
As such, the Chern--Simons gravity on a space with $\Lambda>0$ and Euclidean signature is described by the $SU(2)_k\times SU(2)_{-k}$ Chern--Simons theory. The flat space limit shall be taken as $\Lambda\rightarrow 0$, i.e. the large-level limit $k\rightarrow \infty$, and (\ref{eq:intro-action-contraction}) becomes the non-abelian $BF$ theory in 3d.

When such limit is taken, the action
$-\int\operatorname{tr}(B\wedge F_A)$ imposes flatness of $A$, while the shift
symmetry of $B$ removes the local bulk degrees of freedom
\cite{Blau:1989bq}.  Here $BF[G]$ means that $A$ is a connection
for the compact group $G$ and $B$ is an unrestricted adjoint-valued
one-form.  In particular, compactness of $G$ does not make the $B$ field
periodic, hence once an invariant inner
product identifies $\mathfrak g$ with $\mathfrak g^*$, the $BF$-theory can also be interpreted as the 
Chern--Simons theory for $G\ltimes\mathfrak g^*$.

The goal of this paper is to provide a more precise statement of the $BF$-theory as the $k\rightarrow\infty$ limit of the $\DCS{k}{G}$ theory, at the quantum and categorical level.

An obvious difference between the doubled Chern--Simons theory and the $BF$-theory is the spectrum of topological line operators. For a single Chern-Simons theory CS$_k[G]$ with a compact, simple and simply connected group $G$ and positive
integer level $k$, the integrable representations form a modular tensor category
$\mathcal C_k(G)$~\cite{Verlinde:1988sn,Moore:1988qv}.  Pairing the theory with its orientation reverse then gives
\[
 Z(\mathcal C_k(G))\simeq
 \mathcal C_k(G)\boxtimes\mathcal C_k(G)^{\mathrm{rev}},
\]
and cancels the chiral framing anomaly
\cite{Mueger:2001crc}.  Such doubled theories also admit
state-sum and lattice descriptions, namely the Turaev--Viro state sum~\cite{Turaev:1992hq} and the string-net Hamiltonians~\cite{Levin:2004mi}, which generalize Kitaev's quantum-double models for finite groups~\cite{Kitaev:1997wr}. For both the modular tensor category $\mathcal C_k(G)$ and $Z(\mathcal C_k(G))$ with finite $k$, the spectrum of simple anyon lines is finite.

While for the $BF$-theory, from the description of the compact quantum double~\cite{Koornwinder:1996uq,Koornwinder:1998xg}
\[
  D(G)=C(G)\rtimes_{\operatorname{Ad}}G\,,
\]
its simple anyon lines are labeled by a conjugacy class $[g]\subset G$
and an irreducible representation of its centralizer $C_G(g)$.  The first
label measures the holonomy around the line, and the second is an electric
charge compatible with that holonomy. 

There are precedents for taking a Chern--Simons limit directly at the
level of amplitudes, e.g. the computation of the $BF$ partition functions on Seifert homology spheres and the comparison with the large-level limits of doubled Chern--Simons theory or complex-group Chern--Simons matrix integrals
\cite{Blau:2022krm}.  Their analysis exposes a difficulty
that is also relevant here: the formal integral over flat connections
can contain singular contributions, and the choice of limiting
prescription matters. 

The study of continuous symmetry categories and Symmetry Topological Field Theory (SymTFT) provides another motivation. The categorical framework of topological holography and SymTFT for finite symmetry categories (i.e. fusion categories) is well established~\cite{Kong:2015flk,He:2016xpi,Ji:2019jhk,Kong:2020cie,Gaiotto:2020iye,Moradi:2022lqp,Ji:2021esj,Freed:2022qnc,Schafer-Nameki:2023jdn,Bhardwaj:2023fca,Huang:2023pyk,Bhardwaj:2023idu,Bhardwaj:2023bbf,Choi:2024tri,Luo:2025phx,Schafer-Nameki:2025fiy,Qi:2025tal,Bergman:2026lnz}. While for continuous groups the state spaces and
line spectra cease to be finite. Earlier works on SymTFTs for Lie groups include the 
$BF$-type description of continuous symmetries
~\cite{Brennan:2024fgj,Antinucci:2024zjp,Bonetti:2024cjk,Apruzzi:2024htg,Antinucci:2024bcm,Gagliano:2024off,Argurio:2024ewp,Jia:2025jmn,Bonetti:2025dvm,Jia:2026bbo,Jia:2026tfh} for both abelian and non-abelian cases, as well as spacetime (super-)symmetries~\cite{Apruzzi:2025hvs,Borsten:2025pvq,Ambrosino:2026hhv}. 

While for a categorical formulation of the continuous SymTFT, the proper symmetry category and its Drinfeld center must also accommodate continuous families
of objects. The proposals include the categories of skyscraper sheaves and quasi-coherent sheaves~\cite{Freed:2009qp,Jia:2025vrj,Stockall:2025ngz}. While the most complete construction turns out to be operator-algebraic, which uses measurable Hilbert fields and describes the center
through the conjugation groupoid and its twisted $C^*$-algebra
\cite{Jia:2026vcr}.  For the nonanomalous compact case
this gives the same compact-double representation theory used here.
Recent work has also derived the associated Hilbert spaces and regular
modular kernels, and investigated candidate gaugings
\cite{Jia:2026bbo,Jia:2026tsl}.

In the approach of this paper, our comparison of the $BF$-theory and the large $k$ limit of $\DCS{k}{G}$ uses a different
identification of states.  It keeps the difference of the two chiral
weights finite while their average grows with the level, and realizes
the resulting doubled operators on one Hilbert space. Moreover, we test the comparison of the two theories on normalized
states and modular transformations rather than only comparing formal
label sets.

We first fix the classical scaling. The doubled Chern--Simons action is related to the $BF$-action as in (\ref{eq:intro-action-contraction}) 
with the relative lift and invariant form specified in
Section~\ref{sec:Limit_Regular_Simples}. The gauge transformations contract to the ordinary
and shift symmetries of $BF$ theory. On a manifold with boundary, we compare boundary states after specifying the polarization and boundary counterterms.

The quantum version of this scaling requires a dictionary between the
labels of the two theories.  In $\DCS{k}{G}$ a topological line is
labelled by a pair $a\boxtimes\bar b$ of level-$k$ integrable highest weights, one for each
chiral factor, where $\bar b$ is the highest weight of the representation dual to $b$.  All level-$k$ modular data depend on $k$ through the
shifted level $r_k=k+h^\vee$, where $h^\vee$ is the dual Coxeter number
of $G$.  On the $BF$ side a line is labelled
by a conjugacy class and a centralizer representation.  To compare the
two, fix a maximal torus $T\subset G$ (for $SU(N)$, the diagonal
matrices).  Every element of $G$ is conjugate to an element of $T$, and
two elements of $T$ are conjugate in $G$ exactly when they are related
by the Weyl group $W$ (for $SU(N)$, permutations of the diagonal
entries).  A fundamental domain for this action is the Weyl alcove, a
simplex in the Lie algebra of $T$, and we write $Y$ for its interior.
Points of $Y$ correspond to \emph{regular} conjugacy classes of $G$, whose
centralizer is exactly $T$, whereas points on the walls of the alcove
are \emph{singular} classes, with a larger centralizer; the identity
element, whose centralizer is all of $G$, is a vertex.  The irreducible
representations of $T$ are its characters, i.e. the continuous homomorphisms
$T\to U(1)$.  They form the character lattice
$X^*(T)=\operatorname{Hom}(T,U(1))$, which for simply connected $G$ is the 
weight lattice $P$ of $\mathfrak g$.  We therefore label the characters of
$T$ by weights $\eta\in P=X^*(T)$. Writing $\rho$ for the Weyl vector, half the sum of
the positive roots, we assign to a doubled label $a\boxtimes\bar b$ the
pair
\begin{equation}
 x_{a,b}=\frac{a+b+2\rho}{2r_k},
 \qquad
 \eta_{a,b}=a-b.
 \label{eq:introduction-flux-charge-revised}
\end{equation}
The first entry is the average of the two weights, shifted by $\rho$,
measured in units of the shifted level; it is a point of the alcove and
therefore a conjugacy class, which we interpret as the magnetic flux, or
holonomy, carried by the line.  The second entry is the difference of
the two weights; it lies in $P$ and is a character of $T$, which for a
regular flux is precisely an electric charge compatible with that
holonomy.  Thus $(x,\eta)\in Y\times P$ is exactly the label of a simple
line of $BF$ theory with regular flux.  The shift by $\rho$ is the same
one that appears in the Kac--Peterson formula for the modular
$S$-matrix, which evaluates characters of $G$ at the group element
$\exp\bigl(2\pi i(a+\rho)/r_k\bigr)$.  The scaling is a
correlated large-weight limit: both weights grow with $k$ while their
difference stays fixed.  Holding $a$ and $b$ fixed would instead send
$x$ to the identity class, a singular point where the centralizer is the
whole group, and would therefore only ever reach the purely electric
lines.

The two theories are then compared on one Hilbert space.  For $BF$
theory on a torus, the states with regular flux are square-integrable
functions $\psi(x,\eta)$ of a continuous flux $x\in Y$ and a discrete
charge $\eta\in P$. They form the Hilbert space
\be
\mathcal H_0=L^2(Y\times P,dx\otimes\#)\,,\qquad
\|\psi\|_{\mathcal H_0}^2=\sum_{\eta\in P}\int_Y|\psi(x,\eta)|^2\,dx\,.
\ee
Here $dx$ is the flat (Lebesgue) measure on the alcove, induced by the
invariant inner product on the Lie algebra of $T$ normalized so that long
roots have squared length two (section~\ref{sec:Limit_Regular_Simples}),
and $\#$ is the counting measure on $P$, which gives every charge unit
weight.  The product measure $dx\otimes\#$ thus integrates over the continuous flux and sums over the discrete charge.

Equivalently,
after a Fourier transformation in the charge, $\mathcal H_0$ is the space
of wavefunctions of two commuting holonomies around the two cycles of the
torus, invariant under simultaneous conjugation by $W$.  At finite level
the torus Hilbert space $\mathcal H_k$ is spanned by the finitely many
doubled labels, whose images $(x_{a,b},\eta_{a,b})$ form a lattice in
$Y\times P$ of spacing proportional to $1/k$.  We embed $\mathcal H_k$
into $\mathcal H_0$ by sending each normalized basis state to the
normalized characteristic function of a small cell around its lattice
point.  This sampling map is an isometry, so it fixes how a given
continuum wave packet is normalized at each level, and the cells fill
the alcove as $k\to\infty$.

Our main result, Theorem~\ref{thm:strong-modular}, concerns the modular
group $SL(2,\mathbb Z)$ of large diffeomorphisms of the torus, generated
by the exchange $S$ of the two cycles and the Dehn twist $T$ along one
of them.  Any three-dimensional topological field theory represents this
group unitarily on its torus Hilbert space. At finite level this doubled torus
representation $\rho_k$ is given by the modular $S$- and $T$-matrices of
$Z(\mathcal C_k(G))$. Using the sampling map, the
finite-level action can be transported to $\mathcal H_0$, acting
trivially on the part of $\mathcal H_0$ that the finite-level states do
not reach; we denote the result by $\rho_k^\#$.  The theorem states that
for every fixed modular transformation $\gamma$ and every fixed wave
packet $f\in\mathcal H_0$,
\begin{equation}
 \|\rho_k^\#(\gamma)f-\rho_0(\gamma)f\|_{\mathcal{H}_0}\longrightarrow0\,,
 \qquad
 \|\rho_k^\#(\gamma)^*f-\rho_0(\gamma)^*f\|_{\mathcal{H}_0}\longrightarrow0\,.
 \label{eq:intro-main-limit}
\end{equation}
Here $\rho_0$ is the regular $BF$ modular action.  Its Dehn twist
multiplies a flux--charge state by $e^{2\pi i(\eta,x)}$, the
Aharonov--Bohm pairing of its electric and magnetic labels; its
$S$ transformation exchanges the two torus cycles and becomes a
Weyl-symmetrized Fourier transform.  Equation~\eqref{eq:intro-main-limit}
means that transforming a fixed normalizable wave packet before or
after taking the limit gives the same state.  It also gives convergence
of the corresponding matrix elements between fixed wave packets.
The limit is unchanged for the admissible class of sampling
embeddings specified in Section~\ref{sec:Modular_Limit}.

The proof explains why a pointwise expansion of the modular matrix
would not suffice.  The doubled Kac--Peterson formula is a sum over two
Weyl elements.  Equal Weyl elements produce the continuum Fourier
kernel.  Unequal elements have a rapidly varying phase, and Poisson
summation controls their cancellation on smooth packets supported away
from the walls.  One must then exclude loss of norm into charges that
grow with $k$ or into shrinking neighborhoods of the alcove walls.
The no-escape estimate supplies this step and promotes the local
calculation to the Hilbert-space statement
\eqref{eq:intro-main-limit}.

The restriction to regular sector has important physical interpretations.  The alcove boundary
has zero flat measure, so changing a wavefunction only on that boundary
does not change its $L^2$ class.  A sequence of normalized states can
nevertheless concentrate near it.  In particular, the finite-level
vacuum approaches a singular label and need not converge to a
normalizable vector of $\mathcal H_0$.  The theorem concerns fixed
continuum packets, not every level-dependent sequence of states.  It
does not by itself imply convergence of vacuum partition functions,
fusion spaces, or the full line category.  Even for finite modular
categories, the genus-one modular data do not determine the category
\cite{Mignard:2017mms}.

The final part of the paper studies the limit in the presence of
boundaries and interfaces.  On a manifold with boundary, Chern--Simons
theory has edge modes described by a Wess--Zumino--Witten (WZW) model.
The doubled theory has one left-moving and one right-moving copy.  We ask
which boundary theory survives in the $BF$ limit.  For abelian $G$ the
boundary wavefunctions of the two theories can be compared exactly.
This requires a choice of polarization, meaning a choice of the boundary
variables on which the wavefunction depends, together with matching
boundary counterterms.  The comparison shows explicitly that the
direction in field space that becomes $B$ decompactifies.  For
non-abelian $G$ we check the central charge, the large-level expansion of
the WZW action, the global symmetry, and the contraction of the current
algebra.  All four point to a boundary theory of $\dim G$ non-compact
free scalars valued in $\mathfrak g$, with a careful normalization of the zero modes and the flat
gauging of this boundary theory. This gauging
needs a regulator at the special holonomies where extra zero modes
appear.  We regard these results as consistency checks on a candidate
boundary theory.

For neighboring $SU(2)$ levels we ask whether an interface can
preserve a fixed electric charge while moving the magnetic coordinate
by only one lattice scale.  Proposition~\ref{prop:adjacent-wall-nonrefinement}
obstructs this particular transport by a gapped interface obeying the
stated genus-one consistency conditions.  The diagonal cosets
$(SU(2)_k\times SU(2)_1)/SU(2)_{k+1}$ furnish gapless candidates:
their branching rules contain such channels at bounded conformal
weight.  Their neutral diagonal spectra have the
Runkel--Watts large-level scaling
\cite{Goddard:1984vk,Goddard:1986ee,Runkel:2001ng}.
The continuous-orbifold analysis~\cite{Gaberdiel:2011aa} and the construction of fermionic
Chern--Simons interfaces~\cite{Gaiotto:2026qai} also give useful comparisons.

The structure of this paper is as follows.
In section~\ref{sec:Limit_Regular_Simples}, we fix the normalizations of
the root data and of the flat measure on the Weyl alcove, derive the
contraction of the doubled Chern--Simons action and of its gauge
symmetries to $BF$ theory, label the regular simple lines of the compact
quantum double by $(x,\eta)\in Y\times P$, and construct the finite-level
flux--charge lattice that samples them.
In section~\ref{sec:Modular_Limit}, we prove
Theorem~\ref{thm:strong-modular}.  The two main steps of the proof are
the Poisson decay of the unequal-Weyl terms in
Proposition~\ref{prop:off-diagonal-decay} and the no-escape estimate in
Theorem~\ref{thm:no-escape}; we also show that the limit is the same for
every admissible family of sampling embeddings.
In section~\ref{sec:boundary-theories}, we turn to the boundary theories,
whose analysis is not used in the proofs of
section~\ref{sec:Modular_Limit}.  We compare the boundary wavefunctions of
the two theories exactly for $G=U(1)$ in section~\ref{sec:U1_WZW}, give
four checks of the free-scalar candidate for general $G$ in
section~\ref{sec:G_WZW},
propose a regular-sector ansatz for the non-abelian boundary
wavefunctions in section~\ref{sec:nonabelian-BF-wavefunctions}, and study
the flat gauging of the free-scalar candidate in
section~\ref{sec:flat-gauging}.
In section~\ref{sec:interfaces}, we formulate a refinement criterion for
level-changing interfaces, prove the $SU(2)$ obstruction to uniformly
refining gapped transport in
Proposition~\ref{prop:adjacent-wall-nonrefinement}, and
relate the conformal weights of the coset channels to the flux
displacement across the interface.  This leads us to
Conjecture~\ref{conj:runkel-watts} on the large-level limit of the
unit-jump interface.
In section~\ref{sec:discussion}, we summarize our results and discuss
open questions.  Appendix~\ref{app:endpoint-modular-action} gives the
Fourier construction of the endpoint modular representation and its
compatibility with the regular holonomy chart.

% section_02.tex
\section{Classical contraction and regular flux--charge scaling}
\label{sec:Limit_Regular_Simples}

The limit studied here keeps a magnetic holonomy finite while resolving an
electric charge in the small separation of the two chiral labels.  At the
level of fields this means holding $A=(A_++A_-)/2$ and
$B=k(A_+-A_-)$ fixed.  At the level of Wilson-line labels the two highest
weights grow with $k$, whereas their difference remains fixed.  These are
the variables in which a finite set of doubled Chern--Simons states can
approximate wave packets with continuous flux and discrete charge.

Throughout this section and Section~\ref{sec:Modular_Limit}, $G$ is compact,
connected, simple, and simply connected, and $k$ is a positive integer.
The abelian boundary calculation will be treated separately.  We first fix
the conventions needed to compare the classical action with the quantum
modular matrices, then describe the regular line labels and their sampling
lattice.  The distinction between regular wave packets and singular-flux
line insertions will be kept explicit.

\subsection{Holonomies, root data, and the Hilbert-space measure}

A flat connection around a loop is measured by its holonomy in $G$, up to
conjugation.  The alcove is a set of representatives for this magnetic
variable.  The weight lattice describes the electric charges that remain
after a regular holonomy reduces the gauge stabilizer to $T$.  Fixing their
normalizations together is necessary: the same pairing determines charge
quantization, topological spin, and the phase of the modular $S$ kernel.

Write $\mathfrak g=\operatorname{Lie}(G)$ for the compact real Lie algebra,
realized by anti-Hermitian matrices, and choose a maximal torus $T\subset G$.
We use the single real Cartan space
\begin{equation}
 \mathfrak t:=-i\,\operatorname{Lie}(T)\subset\mathfrak g_{\mathbb C},
 \qquad x\longmapsto e^{2\pi i x}\in T.
 \label{eq:normalized-torus-covering}
\end{equation}
Thus $x\in\mathfrak t$ is Hermitian and $2\pi i x\in\mathfrak g$.
The kernel of this map is the cocharacter lattice $X_*(T)$, which equals
the coroot lattice $Q^\vee$ because $G$ is simply connected.

For products of Lie-algebra elements, $\operatorname{tr}$ denotes
$1/(2h^\vee)$ times the trace in the adjoint representation; in particular,
\begin{equation}
 \operatorname{tr}(XY):=
 \frac{1}{2h^\vee}
 \operatorname{Tr}_{\mathfrak g_{\mathbb C}}
 \bigl(\operatorname{ad}_X\operatorname{ad}_Y\bigr),
 \qquad X,Y\in\mathfrak g_{\mathbb C},
 \label{eq:trace-normalization}
\end{equation}
where $h^\vee$ is the dual Coxeter number.  The numerator is the Killing
form, extended complex-bilinearly.  Hence $-\operatorname{tr}(XY)$ is
positive definite on the compact real algebra.  For $SU(N)$,
$\operatorname{tr}(XY)=\operatorname{Tr}_{\mathrm{fund}}(XY)$.
For Lie-algebra-valued differential forms, products inside the trace
include the displayed wedge products.
The level-one normalization is fixed by requiring that the Chern--Weil class
\begin{equation}
 \left[-\frac{1}{8\pi^2}
 \operatorname{tr}(F\wedge F)\right]
 \in H^4(BG;\mathbb R)
 \label{eq:integral-CS-form}
\end{equation}
is the image of the basic generator of $H^4(BG;\mathbb Z)$.  In particular,
for every closed oriented four-manifold $X$ and every $G$-bundle on $X$,
\begin{equation}
 -\frac{1}{8\pi^2}\int_X
 \operatorname{tr}(F\wedge F)\in\mathbb Z.
 \label{eq:integral-CS-period}
\end{equation}
The choice of the basic generator in \eqref{eq:integral-CS-form} fixes the
level-one normalization: the integrality condition
\eqref{eq:integral-CS-period} alone would also allow positive integer
multiples of this form.  The Euclidean inner product used for roots,
alcoves, and Lebesgue measure is simply
\begin{equation}
 (x,y):=\operatorname{tr}(xy),
 \qquad x,y\in\mathfrak t.
 \label{eq:geometric-Cartan-form}
\end{equation}
With this normalization every long root has squared length two.
We use $(\, ,\,)$ for roots, coroots, weights, alcove coordinates,
and covolumes, and extend it complex-bilinearly when necessary.
The conversion to the positive pairing of compact Lie-algebra elements is
\begin{equation}
 -\operatorname{tr}\bigl((2\pi i x)(2\pi i y)\bigr)
 =(2\pi)^2(x,y).
 \label{eq:level-geometric-form-conversion}
\end{equation}
Gauge-field pairings below are written directly as
$-\operatorname{tr}(XY)$; Cartan-coordinate pairings in the boundary and
interface calculations use the same $(\, ,\,)$ as above.
Angle brackets retain their usual meanings for weight--coroot duality,
the explicitly defined split pairing, and Hilbert-space or cotangent
duality pairings.

The roots are normalized real character roots: a root $\alpha$ satisfies
\begin{equation}
 \chi_\alpha(e^{2\pi i x})=e^{2\pi i\alpha(x)}.
 \label{eq:real-character-root-convention}
\end{equation}
Thus $\Delta\subset\mathfrak t^*$.  We use the geometric form in
Eq.~\eqref{eq:geometric-Cartan-form} to identify
$\mathfrak t^*\simeq\mathfrak t$ whenever roots occur in metric formulas.
Fix simple
roots $\alpha_1,\ldots,\alpha_\ell$ and the corresponding positive system
$\Delta_+$, where $\ell=\operatorname{rank}G$.  For a root $\alpha$,
put $\alpha^\vee=2\alpha/(\alpha,\alpha)$.  The root, coroot, and weight
lattices are
\[
 Q=\mathbb Z\Delta,\qquad
 Q^\vee=\mathbb Z\{\alpha^\vee:\alpha\in\Delta\},\qquad
 P=(Q^\vee)^*.
\]
In the common Euclidean realization fixed by the geometric form
\eqref{eq:geometric-Cartan-form},
\begin{equation}
 Q^\vee\subset Q\subset P.
 \label{eq:root-coroot-weight-inclusions}
\end{equation}
Indeed, a long coroot is the corresponding long root, while a short coroot
is the corresponding short root multiplied by the lacing number; the second
inclusion is the integrality of the Cartan matrix.  The first inclusion can
be strict in non-simply-laced types, so $Q$ and $Q^\vee$ will not be
interchanged below.  We also write
\begin{equation}
 P_{\mathbb R}:=P\otimes_{\mathbb Z}\mathbb R,\qquad
 Q_{\mathbb R}:=Q\otimes_{\mathbb Z}\mathbb R
 =Q^\vee\otimes_{\mathbb Z}\mathbb R\simeq\mathfrak t.
 \label{eq:real-weight-space}
\end{equation}
All these lattices have full rank, so
$P_{\mathbb R}=Q_{\mathbb R}\simeq\mathfrak t$.
The canonical weight--coroot pairing is
$\langle\lambda,\alpha^\vee\rangle
=2(\lambda,\alpha)/(\alpha,\alpha)$.  Let
\[
 P_+=\{\lambda\in P:
       \langle\lambda,\alpha_i^\vee\rangle\geq0
       \text{ for }1\leq i\leq\ell\}.
\]
We write $W$ for the Weyl group, $w_0$ for its longest element, and
$\epsilon(w)=\det(w)$.  For a weight $\lambda$, its dual highest weight is
\begin{equation}
 \lambda^*:=-w_0\lambda.
 \label{eq:dual-highest-weight-convention}
\end{equation}
Let $\theta$ be the highest root,
$\rho=\frac12\sum_{\alpha\in\Delta_+}\alpha$ the Weyl vector, with $h^\vee$
as in \eqref{eq:trace-normalization}.  The shifted level is
\begin{equation}
 r_k=k+h^\vee.
 \label{eq:preliminary-dimensions}
\end{equation}
For simply connected $G$, restriction to the center identifies
$P/Q\simeq\widehat{Z(G)}$.  Thus two highest weights define the same central
character precisely when their difference lies in $Q$.
Standard references for the root, coroot, weight, affine-Weyl, and center
conventions used here are \cite{Bourbaki:2002,Kac:1990gs}.

With our exponential convention, the maximal torus is identified as
\begin{equation}
 T\simeq\mathfrak t/Q^\vee,
 \qquad x\longmapsto g_x:=e^{2\pi i x}.
 \label{eq:torus-exponential-convention}
\end{equation}
Using the inner product to identify $\mathfrak t^*$ with $\mathfrak t$,
the weight lattice is the lattice dual to $Q^\vee$:
\begin{equation}
 P=\bigl\{\eta\in\mathfrak t:
       (\eta,q^\vee)\in\mathbb Z
       \text{ for every }q^\vee\in Q^\vee\bigr\}.
 \label{eq:weight-coroot-duality}
\end{equation}
Each $\eta\in P$ determines a character of $T$ by
\[
 \chi_\eta(g_y)=e^{2\pi i(\eta,y)}.
\]
The integrality condition makes this expression independent of the
representative $y$ modulo $Q^\vee$. Every continuous character of $T$
arises uniquely in this way.

Two elements of $T$ are conjugate in $G$ precisely when they lie in the
same Weyl orbit. On $\mathfrak t$, Weyl transformations together with
coroot translations form the affine Weyl group
$W_{\mathrm{aff}}=W\ltimes Q^\vee$, acting by
$x\mapsto wx+q^\vee$. A closed fundamental alcove for this action
and its interior are
\begin{equation}
 \overline Y=
 \bigl\{x\in\mathfrak t:
   (\alpha_i,x)\geq0\ (1\leq i\leq\ell),\quad
   (\theta,x)\leq1\bigr\},
 \qquad Y=\operatorname{int}\overline Y.
 \label{eq:fundamental-alcove}
\end{equation}
Every conjugacy class in $G$ has a unique representative $g_x$ with
$x\in\overline Y$. This representative is regular precisely when
$x\in Y$, in which case $C_G(g_x)=T$. On the boundary
$\partial\overline Y$, the centralizer strictly contains $T$.

Since the highest root $\theta$ is long and has squared length two,
$(\theta,x)=\langle x,\theta^\vee\rangle$. The lower-wall conditions
are likewise equivalent to
$\langle x,\alpha_i^\vee\rangle\geq0$, since
\[
 (\alpha_i,x)
 =\frac{(\alpha_i,\alpha_i)}{2}
   \langle x,\alpha_i^\vee\rangle.
\]
Thus the alcove inequalities use the same coroot pairings as the
dominance and level conditions for integrable weights below.

The geometric inner product \eqref{eq:geometric-Cartan-form} fixes the
flat (Lebesgue) measure $dx$ on $\mathfrak t$.\footnote{
In orthonormal coordinates, this is the usual Euclidean volume element
$dx=dx^1\cdots dx^\ell$. For smooth, compactly supported wave packets,
the norm integrals agree with ordinary Riemann integrals. Such packets
form a dense subspace of the $L^2$ state space used below and may be
kept in mind for physical intuition. The Lebesgue formulation also
includes their $L^2$ limits, which need not be continuous.}
For a full lattice $L\subset\mathfrak t$,
$\operatorname{covol}(L)$ denotes the $dx$-volume of a fundamental cell.
Since $P$ and $Q^\vee$ are dual,
\begin{equation}
 \operatorname{covol}(P)\operatorname{covol}(Q^\vee)=1.
 \label{eq:lattice-and-Haar-normalization}
\end{equation}
More precisely, if $F\subset\mathfrak t$ is any measurable fundamental cell
for $Q^\vee$, normalized Haar measure on $T$ is characterized by
\begin{equation}
 \int_T f(t)\,dm_T(t)
 =\operatorname{covol}(Q^\vee)^{-1}
   \int_F f(e^{2\pi i x})\,dx
 \label{eq:Haar-integral-normalization}
\end{equation}
for continuous $f$.  Unless another measure is displayed, an alcove integral uses the flat
measure $dx$.  This is the measure that appears in the sampling limit and,
after Fourier transformation in the second torus holonomy, in the
genus-one Hilbert space of Section~\ref{sec:Modular_Limit}.  It should be
distinguished from the distribution of conjugacy classes of a
Haar-distributed element of $G$.  For a continuous class function $F$, the
Weyl integration formula gives
\begin{equation}
 \int_G F(g)\,dm_G(g)
 =\frac{1}{\operatorname{covol}(Q^\vee)}
   \int_Y F(g_x)\,\mathcal J(x)\,dx,
 \qquad
 \mathcal J(x)=\prod_{\alpha\in\Delta_+}
                  4\sin^2\!\bigl(\pi(\alpha,x)\bigr).
 \label{eq:alcove-Haar-Jacobian}
\end{equation}
The factor $1/|W|$ in the torus form of the Weyl integration formula is
cancelled by the $|W|$ alcoves in a torus fundamental domain; see
\cite[pp.~123--128]{Bump:2013}.  The Jacobian $\mathcal J$
belongs to the Haar push-forward in Eq.~\eqref{eq:alcove-Haar-Jacobian}.
The flat measure used for genus one is fixed by the torus holonomy
coordinates and their Fourier transform.

The boundary of the alcove consists of holonomies with enhanced gauge
stabilizer.  Such holonomies have zero probability in a wave packet with
a square-integrable density relative to $dx$.  This gives the following
precise meaning to the use of an open alcove in the Hilbert space.

\begin{proposition}[Removing the alcove boundary in $L^2$]
\label{prop:alcove-null-boundary}
Restriction from $\overline Y$ to $Y$ and extension by zero from $Y$ to
$\overline Y$ are inverse unitary maps on $L^2$ equivalence classes:
\begin{equation}
 L^2(Y,dx)\simeq L^2(\overline Y,dx).
 \label{eq:open-closed-alcove-L2}
\end{equation}
The same statement holds after taking the product with $P$ equipped with
counting measure.
\end{proposition}

\begin{proof}
The boundary lies in the finite union of affine hyperplanes
\[
 \{x:(\alpha_i,x)=0\}\quad(1\leq i\leq\ell),
 \qquad \{x:(\theta,x)=1\}.
\]
Choose an orthonormal coordinate normal to one such hyperplane.  In these
coordinates its intersection with a bounded box is contained in a product
of an $(\ell-1)$-dimensional box and a singleton.  Fubini's theorem and the
zero one-dimensional Lebesgue measure of a singleton imply that this
intersection has $\ell$-dimensional measure zero.  The alcove is bounded,
so a finite union of these intersections covers its boundary and
$dx(\partial\overline Y)=0$.

For a measurable $f$ on $\overline Y$, deleting this null set leaves its
norm unchanged:
\[
 \int_{\overline Y}|f|^2\,dx=\int_Y|f|^2\,dx.
\]
Restriction consequently respects almost-everywhere equivalence and is
isometric.  Conversely, extend a representative of $g\in L^2(Y)$ by zero
on $\partial\overline Y$.  This defines an $L^2(\overline Y)$ class with
the same norm.  Restriction followed by extension changes a representative
only on the null boundary, so the two maps are inverse on the quotient
spaces.  Finally, $P$ is countable and
\[
 (dx\otimes\#)(\partial\overline Y\times P)
 =\sum_{\eta\in P}dx(\partial\overline Y)=0.
\]
The argument therefore applies in each charge fiber and to their Hilbert
direct sum.
\end{proof}

This identification does not impose a boundary condition on a wavefunction.
A general $L^2$ class has no boundary value to impose.  It also does not
remove singular-flux representations from the line spectrum: a state or
distribution supported exactly on a wall is outside this absolutely
continuous $L^2$ description.  Its stabilizer and its allowed electric
charges must be specified separately.  We return to this distinction below.
We write $K\Subset U$ when $K$ is compact and contained in the open set $U$.

\subsection{The doubled action and its contraction}

The opposite levels cancel the terms that would remain if the two
connections were equal.  Their first relative variation of the doubled Chern--Simons terms is proportional
to curvature, which is the coupling required for a $BF$ action.  Keeping
that variation finite fixes the factor of $k$ in the definition of $B$.

Let $M$ be an oriented three-manifold and let $A_+$ and $A_-$ be connections
on the same principal $G$-bundle.  Their average is a connection and their
difference is an adjoint-valued one-form.  In a common trivialization put
\begin{equation}
 F_A=dA+\frac12[A,A],
 \qquad
 \operatorname{CS}(A)
 =-\int_M\operatorname{tr}\left(A\wedge dA
 +\frac13 A\wedge[A,A]\right).
 \label{eq:CS-functional}
\end{equation}
The doubled action on a closed manifold is
\begin{equation}
 S_k[A_+,A_-]=\frac{k}{4\pi}
 \bigl(\operatorname{CS}(A_+)
      -\operatorname{CS}(A_-)\bigr)
 \quad\bmod 2\pi\mathbb Z.
 \label{eq:doubled-CS-action}
\end{equation}
Choose the real-valued lift determined by the relative transgression formula
below, and denote the resulting lifted action by $\widetilde S_k$.
Introduce
\begin{equation}
 A=\frac{A_++A_-}{2},\qquad B=k(A_+-A_-),\qquad
 A_\pm=A\pm\frac{B}{2k}.
 \label{eq:classical-field-recombination}
\end{equation}
For an adjoint-valued one-form $\Xi$, set $A_t=A+t\Xi$.  The variation
of the Chern--Simons functional in a fixed boundary trivialization is
\[
 \frac{d}{dt}\operatorname{CS}(A_t)
 =-2\int_M\operatorname{tr}(\Xi\wedge F_{A_t})
   +\int_{\partial M}\operatorname{tr}(A_t\wedge\Xi).
\]
Here
$F_{A_t}=F_A+tD_A\Xi+\tfrac12t^2[\Xi,\Xi]$.
Integrating from $t=-1$ to $t=1$ removes the odd term, while
$\int_{-1}^1t^2dt=2/3$.  Symmetry of the invariant form gives
$\operatorname{tr}(\Xi\wedge\Xi)=0$, so the boundary integral only
contains $A$.  This proves the exact transgression identity
\begin{align}
 \operatorname{CS}(A+\Xi)
 -\operatorname{CS}(A-\Xi)
 ={}&-4\int_M\operatorname{tr}(\Xi\wedge F_A)
 -\frac23\int_M\operatorname{tr}(\Xi\wedge[\Xi,\Xi])\notag\\
 &+2\int_{\partial M}\operatorname{tr}(A\wedge\Xi).
 \label{eq:relative-CS-transgression}
\end{align}
Taking $\Xi=B/(2k)$ yields
\begin{equation}
 \widetilde S_k[A_+,A_-]
 =-\frac1{2\pi}\int_M\operatorname{tr}(B\wedge F_A)
 -\frac1{48\pi k^2}\int_M
 \operatorname{tr}(B\wedge[B,B])
 \label{eq:classical-CS-to-BF}
\end{equation}
on a closed manifold.  With a fixed boundary trivialization the accompanying
boundary term is
\begin{equation}
 \frac1{4\pi}\int_{\partial M}
 \operatorname{tr}(A\wedge B).
 \label{eq:relative-CS-boundary-term}
\end{equation}
The first term in Eq.~\eqref{eq:classical-CS-to-BF} is the $BF[G]$ action
in our normalization.  Its variation with respect to $B$ imposes $F_A=0$;
variation with respect to $A$ imposes $D_AB=0$.  The finite-level equations
make the approach to these constraints equally explicit.  Adding and
subtracting $F_{A_+}=F_{A_-}=0$ gives the equations of motion in the variables $A$ and $B$
\begin{equation}
 F_A+\frac{1}{8k^2}[B,B]=0,\qquad D_AB=0.
 \label{eq:contracted-classical-equations}
\end{equation}
Thus the cubic interaction measures the finite-$k$ deformation of the
flatness constraint.  Freed's relative formulation supplies the global
bundle and boundary interpretation of the action and its transgression
\cite[Secs.~1--3]{Freed:1992vw}.

The infinitesimal gauge transformations contract in the same variables.
To distinguish their parameters from the Weyl sign $\epsilon(w)$, write
\begin{equation}
 \varepsilon_\pm=\varepsilon\pm\frac{\beta}{2k}.
 \label{eq:classical-gauge-recombination}
\end{equation}
Substitution in
$\delta A_\pm=d\varepsilon_\pm+[A_\pm,\varepsilon_\pm]$ yields
\begin{equation}
 \delta A=D_A\varepsilon+\frac1{4k^2}[B,\beta],\qquad
 \delta B=D_A\beta+[B,\varepsilon].
 \label{eq:classical-gauge-contraction}
\end{equation}
The $k\to\infty$ transformations are the ordinary gauge and shift symmetries
of $BF[G]$.
The shift $\delta B=D_A\beta$ preserves the action on a closed manifold:
its variation is proportional to $-\int_Md\operatorname{tr}(\beta F_A)$,
because the Bianchi identity gives $D_AF_A=0$.  This is the extra gauge
redundancy of the momentum field $B$ that survives the contraction.
This is an In\"on\"u--Wigner contraction in the standard sense
\cite{Inonu:1953sp,Saletan:1961iqo}.

The same contraction can be read from the metric Lie algebra.  For
$X\in\mathfrak g$, define $J_X=(X,X)$ and
$P_X=(X,-X)/(2k)$ in $\mathfrak g\oplus\mathfrak g$.  Then
\begin{equation}
 [J_X,J_Y]=J_{[X,Y]},\qquad
 [J_X,P_Y]=P_{[X,Y]},\qquad
 [P_X,P_Y]=\frac1{4k^2}J_{[X,Y]}.
 \label{eq:classical-Lie-contraction}
\end{equation}
For the split level form determined by the two Chern--Simons summands, set
\begin{equation}
 \bigl\langle X_+\oplus X_-,Y_+\oplus Y_-\bigr\rangle_k^{\mathrm{split}}
 :=-k\bigl(
 \operatorname{tr}(X_+Y_+)
 -\operatorname{tr}(X_-Y_-)\bigr).
 \label{eq:split-level-pairing}
\end{equation}
In the rescaled basis,
\begin{equation}
 \langle J_X,J_Y\rangle_k^{\mathrm{split}}
 =\langle P_X,P_Y\rangle_k^{\mathrm{split}}=0,
 \qquad
 \langle J_X,P_Y\rangle_k^{\mathrm{split}}
 =-\operatorname{tr}(XY).
 \label{eq:classical-contracted-pairing}
\end{equation}
Here $J_X$ and $P_X$ are the abstract generators in this $k$-dependent
basis; $P_X=(X,-X)/(2k)$ specifies their realization in
$\mathfrak g\oplus\mathfrak g$ before the contraction.

At $k=\infty$ these formulas give the metric Lie algebra
\[
 \mathfrak g\ltimes_{\operatorname{ad}}\mathfrak g_{\mathrm{ab}}
 \simeq\mathfrak g\ltimes_{\operatorname{ad}^*}\mathfrak g^*,
 \qquad
 \langle J_X,P_Y\rangle_\infty=-\operatorname{tr}(XY),
\]
with vanishing $J$--$J$ and $P$--$P$ pairings.  For the connection
$J_A+P_B$, its Chern--Simons functional is
\[
 -2\int_M\operatorname{tr}(B\wedge F_A)
 +\int_{\partial M}\operatorname{tr}(A\wedge B).
\]
Multiplication by $1/(4\pi)$ gives the leading action in
Eq.~\eqref{eq:classical-CS-to-BF} and the boundary term in
Eq.~\eqref{eq:relative-CS-boundary-term}.  The classical tangent group $G\ltimes\mathfrak g_{\mathrm{ab}}$ and the quantum double enter at different
steps.  The ordinary group $C^*$-algebra of
$G\ltimes\mathfrak g_{\mathrm{ab}}$, after Fourier transformation in the
abelian normal subgroup, has the form
\[
 C_0(\mathfrak g^*)\rtimes_{\operatorname{Ad}^*}G.
\]
Its momentum variable is Lie-algebra valued.  The line algebra used below
is instead $\mathfrak D(G)=C(G)\rtimes_{\operatorname{Ad}}G$, with a
group-valued magnetic holonomy.  The latter arises in the combinatorial
quantization of the Poisson structures for tangent-group Chern--Simons
theory, where $\mathfrak D(G)$ acts as a quantum symmetry
\cite{Meusburger:2003hc}; the $SU(2)$ particle interpretation is discussed
in \cite{Bais:1998yn}.  A contraction of the classical Lie brackets alone
would not identify these two algebras.  Our genus-one comparison uses the
compact double with its group-valued flux.

In the abelian boundary calculation in section~\ref{sec:U1_WZW}, we use lower-case
fields.  The specialization of \eqref{eq:classical-field-recombination} is
\begin{equation}
 a=\frac12(a_++a_-),\qquad b=k(a_+-a_-).
 \label{eq:U1-BF-fields}
\end{equation}
The finite-level periodicities are inherited from the two compact fields.
In normalized holonomy coordinates, integral shifts
$a_\pm\mapsto a_\pm+\lambda_\pm$ give
\begin{equation}
 a\longmapsto a+\frac{\lambda_++\lambda_-}{2},
 \qquad
 b\longmapsto b+k(\lambda_+-\lambda_-).
 \label{eq:U1-recombined-periodicities}
\end{equation}
In particular, transformations which keep the diagonal coordinate fixed
have $\lambda_-=-\lambda_+$ and shift $b$ by $2k\lambda_+$.
Hence the relative period grows linearly with $k$; this is the periodicity
used in the boundary wavefunction calculation.

Equations \eqref{eq:classical-field-recombination}--
\eqref{eq:classical-gauge-contraction} concern families for which $A$ and
$B$ remain fixed while
\begin{equation}
 A_+-A_-=\frac{B}{k}=O(k^{-1}).
 \label{eq:near-diagonal-contraction-rate}
\end{equation}
For such smooth families the cubic term in
\eqref{eq:classical-CS-to-BF} tends to zero.  This is a near-diagonal
classical contraction on a common bundle.  The classical change of variables
uses the integer level $k$, whereas the quantum sampling below naturally
involves the shifted level $r_k=k+h^\vee$; boundary comparisons use the
polarizations fixed in the boundary analysis.

\subsection{Regular compact-double labels}

A line insertion can carry magnetic flux and electric charge at the same
time.  Gauge transformations conjugate its holonomy $g$, so the magnetic
datum is $[g]$.  Once a representative $g$ is fixed, the transformations
that preserve it form $C_G(g)$.  The electric datum is an irreducible
representation of this stabilizer.  For $g=g_x$ with $x\in Y$, the
stabilizer is $T$ and all its irreducible representations are the
characters $\chi_\eta$, $\eta\in P$.  This is the origin of a continuous
flux variable paired with a lattice-valued charge
\cite{Koornwinder:1996uq,Bais:1998yn}.

The corresponding operator algebra contains functions of the holonomy and
unitary gauge transformations.  The multiplication law expressing how a
gauge transformation conjugates holonomy is encoded by a crossed product:
\begin{equation}
  D(G)=C(G)\rtimes_{\Ad}G,
 \qquad
  D_{\mathrm{reg}}(G)
 =C_0(G_{\mathrm{reg}})\rtimes_{\Ad}G.
 \label{eq:full-and-regular-compact-double}
\end{equation}
Here $G_{\mathrm{reg}}$ is the open set of regular elements and $\rtimes$
denotes the full crossed product.  Compactness of $G$ implies agreement
with the reduced crossed product \cite{Williams:2007}.
If $M(F)$ is a holonomy observable and $U_h$ implements a gauge
transformation, the defining covariance relation reads
\begin{equation}
 U_hM(F)U_h^{-1}=M(F\circ\operatorname{Ad}_{h^{-1}}).
 \label{eq:compact-double-covariance}
\end{equation}
This relation is sufficient for the representation model needed here;
no tensor-product decomposition is required for the genus-one limit.

The regular algebra is an ideal because $G_{\mathrm{reg}}$ is invariant
under conjugation.  Its nondegenerate representations extend uniquely to
$\mathfrak D(G)$.  The full algebra also includes singular-flux
representations, with charges in $\operatorname{Irr}C_G(g)$.
For identity flux these are precisely $\operatorname{Irr}G$, including
the tensor unit.  Consequently, the regular label set is a stratum of the
full line spectrum and does not by itself contain a tensor unit.

For $x\in Y$ and $\eta\in P$, Mackey--Green induction associates to the
orbit of $g_x$ and the stabilizer character $\chi_\eta$ the covariant
representation
\begin{equation}
 \pi_{x,\eta}
 =\operatorname{Ind}_{T}^{G}
   (\operatorname{ev}_{g_x},\chi_\eta)
 \quad\text{of}\quad  D_{\mathrm{reg}}(G).
 \label{eq:regular-Mackey-representation}
\end{equation}
Here $\operatorname{ev}_{g_x}$ is evaluation at $g_x$, and
$\mathbb C_\eta$ is the one-dimensional $T$-module with character
$\chi_\eta$.  To fix the induction convention, one may realize the carrier
as
\begin{equation}
 \mathcal H_{x,\eta}
 =\left\{
 \begin{aligned}
 f:G\to\mathbb C:\quad
 &f\text{ is measurable, modulo a.e. equality},\\
 &f(ut)=\chi_\eta(t)^{-1}f(u)
     \quad\text{for every }t\in T\text{ and for a.e. }u\in G,\\
 &\int_{G/T}|f(u)|^2\,d(uT)<\infty
 \end{aligned}
 \right\}.
 \label{eq:explicit-induced-carrier}
\end{equation}
Here $d(uT)$ is the normalized $G$-invariant probability measure on $G/T$.
For $F\in C_0(G_{\mathrm{reg}})$ and $h\in G$, the covariant pair acts by
\begin{equation}
 (M_x(F)f)(u)=F(ug_xu^{-1})f(u),
 \qquad (U_hf)(u)=f(h^{-1}u).
 \label{eq:explicit-compact-double-action}
\end{equation}
The inverse character in \eqref{eq:explicit-induced-carrier} fixes our
induction convention.  It also has a direct geometric meaning: changing
the choice of gauge frame $u$ by $ut$ must be compensated by the electric
character of $t$.  The function $f$ is therefore an $L^2$ section of the
associated line bundle over $G/T$, the conjugacy orbit of $g_x$.
The norm is well defined because $|\chi_\eta(t)|=1$.

Both actions preserve the equivariance condition.  For the first, $t$
commutes with $g_x$, so
$F(utg_xt^{-1}u^{-1})=F(ug_xu^{-1})$; for the second, left translation
commutes with the right $T$-action.  Invariance of the measure on $G/T$
makes $U_h$ unitary, and direct substitution verifies
Eq.~\eqref{eq:compact-double-covariance}.  The imprimitivity theorem
identifies these induced representations as all the irreducible
representations on the regular orbit stratum
\cite{Rieffel:1974,Koornwinder:1996uq,
Koornwinder:1998xg}.  Taking $F\in C(G)$ in the same formula
gives their extension to the full compact double.

Before an alcove representative is chosen, the labels obey
\begin{equation}
 (x,\eta)\sim(wx+q^\vee,w\eta),
 \qquad w\in W,\quad q^\vee\in Q^\vee.
 \label{eq:regular-label-affine-Weyl-identification}
\end{equation}
Indeed, the translation does not change $g_x$, while a representative of
$w$ conjugates $g_x$ and transports the stabilizer character from $\eta$ to
$w\eta$.  Every regular conjugacy class meets $Y$ exactly once.  Moreover,
an affine Weyl element fixing a point of $Y$ must be the identity, since an
interior point lies on no affine reflection hyperplane.  Hence choosing
$x\in Y$ removes the ambiguity in both entries, and the regular irreducible
classes are parametrized without repetition by
\begin{equation}
 (x,\eta)\in Y\times P.
 \label{eq:regular-BF-labels}
\end{equation}

The common comparison space used at genus one is
\begin{equation}
 \mathcal H_0=L^2(Y\times P,dx\otimes\#),\qquad
 \|f\|_{\mathcal H_0}^2
 =\sum_{\eta\in P}\int_Y|f(x,\eta)|^2\,dx,
 \label{eq:preliminary-comparison-Hilbert-space}
\end{equation}
where $P$ has the discrete topology and $\#$ is counting measure.
Proposition~\ref{prop:alcove-null-boundary} permits replacing $Y$ by
$\overline Y$ in this norm.  A vector $f(x,\eta)$ describes amplitudes for
a wave packet in the line labels.  The induced carrier
$\mathcal H_{x,\eta}$ describes the gauge orbit of one fixed line; these
are different uses of a Hilbert space.  Section~\ref{sec:Modular_Limit}
identifies $\mathcal H_0$ with
$L^2(T\times T,dm_T\otimes dm_T)^W$, the space of $L^2$-wavefunction on the moduli space of flat $G$-connections on a spatial torus, where $W$ acts diagonally
and $dm_T$ is normalized Haar measure on $T$.
This identification uses Fourier expansion in the second holonomy.

Smooth wave packets supported in a compact subset of $Y$ and in finitely
many charge fibers are dense in $\mathcal H_0$.  Indeed, the convergent
sum in Eq.~\eqref{eq:preliminary-comparison-Hilbert-space} can first be
truncated to finitely many charges.  For each retained fiber, removing
a shrinking neighborhood of $\partial Y$ changes the $L^2$ norm by a
quantity tending to zero, by absolute continuity of the integral.
Mollification inside the remaining open set then gives a smooth compactly
supported approximation.  These elementary operations explain why the
modular proof can begin with regular wave packets and later extend to
every vector of $\mathcal H_0$.  A separate estimate is still needed to
show that the finite-level operators do not send their norm into smaller
and smaller wall neighborhoods or unbounded charges.

\subsection{Finite labels and the regular scaling window}

The continuum limit requires a choice of which finite-level states are
held at fixed physical flux and charge.  Holding the highest weights
themselves fixed would drive the flux to the identity as $k$ increases.
To probe a regular holonomy, both weights must instead grow linearly with
the shifted level while their difference stays finite.

Let $\mathcal C_k=\mathcal C_k(G)$ be the level-$k$ WZW fusion category.  Its
simple objects are indexed by
\begin{equation}
 P_k^+=\{a\in P_+:\langle a,\theta^\vee\rangle\leq k\}.
 \label{eq:integrable-weight-alcove}
\end{equation}
The shifted chiral coordinate of $a\in P_k^+$ is
\begin{equation}
 z_a=\frac{a+\rho}{r_k}\in Y.
 \label{eq:shifted-chiral-coordinate}
\end{equation}
To check that $z_a$ lies in the open alcove, observe that
for every simple coroot,
\begin{equation}
 \langle z_a,\alpha_i^\vee\rangle
 =\frac{\langle a,\alpha_i^\vee\rangle+1}{r_k}>0,
 \label{eq:shifted-coordinate-lower-walls}
\end{equation}
whereas at the affine wall
\begin{equation}
 \langle z_a,\theta^\vee\rangle
 =\frac{\langle a,\theta^\vee\rangle+h^\vee-1}{r_k}
 \leq\frac{k+h^\vee-1}{k+h^\vee}<1.
 \label{eq:shifted-coordinate-affine-wall}
\end{equation}
The parametrization of integrable affine highest weights by
\eqref{eq:integrable-weight-alcove}, and the shifted level $k+h^\vee$, are
standard; see \cite[Chs.~12--13]{Kac:1990gs}.

We use a charge-oriented convention for doubled simples.  If
$\bar b=-w_0b$ is the highest weight of the dual representation, , cf.\
eq.~\eqref{eq:dual-highest-weight-convention}, the ordered
pair $(a,b)$ denotes the actual center object
\begin{equation}
 Z^{\mathrm{ch}}_{a,b}=a\boxtimes\bar b\in
 Z(\mathcal C_k)\simeq
 \mathcal C_k\boxtimes\mathcal C_k^{\mathrm{rev}},
 \qquad a,b\in P_k^+.
 \label{eq:charge-oriented-center-label}
\end{equation}
For the center equivalence and the reverse-braiding convention, see
\cite{Mueger:2001crc,Bakalov:2001}.
Its flux and charge are respectively the average and the unscaled difference
of the two shifted chiral labels:
\begin{equation}
 x_{a,b}=\frac{z_a+z_b}{2}
 =\frac{a+b+2\rho}{2r_k},
 \qquad
 \eta_{a,b}=r_k(z_a-z_b)=a-b.
 \label{eq:classical-regular-lattice}
\end{equation}
The inverse formulas are
\begin{equation}
 a=r_kx_{a,b}-\rho+\frac{\eta_{a,b}}2,
 \qquad
 b=r_kx_{a,b}-\rho-\frac{\eta_{a,b}}2.
 \label{eq:inverse-flux-charge-map}
\end{equation}
They show that the map $(a,b)\mapsto(x_{a,b},\eta_{a,b})$ is
injective.  Conversely, a pair $(x,\eta)\in\mathfrak{t}\times P$ 
corresponds to a doubled simple precisely when
\begin{equation}
 r_kx-\rho+\frac\eta2\in P_k^+,
 \qquad
 r_kx-\rho-\frac\eta2\in P_k^+.
 \label{eq:flux-charge-integrability-test}
\end{equation}
The lattice condition applies to the two combinations in
Eq.~\eqref{eq:flux-charge-integrability-test}; it allows $\eta/2$ itself to
lie outside $P$.
Equivalently,
\begin{equation}
 \frac{a+\rho}{r_k}-\frac{b+\rho}{r_k}
 =\frac{\eta_{a,b}}{r_k}.
 \label{eq:chiral-flux-separation}
\end{equation}
Since $r_k=k+h^\vee$, this agrees to leading order with the antidiagonal
separation in \eqref{eq:classical-field-recombination}.  The identities
\begin{equation}
 S^{(k)}_{\bar b,\bar d}=S^{(k)}_{bd},\qquad
 \theta^{(k)}_{\bar b}=\theta^{(k)}_b
 \label{eq:dual-label-modular-identities}
\end{equation}
show that the $S$-matrix entry between $a\boxtimes\bar b$ and $c\boxtimes\bar d$ in
$\mathcal C_k\boxtimes\mathcal C_k^{\rm rev}$ and the categorical twist of $a\boxtimes\bar b$ are
$S^{(k)}_{ac}\overline{S^{(k)}_{bd}}$ and
$\theta^{(k)}_a(\theta^{(k)}_b)^{-1}$, respectively.  Thus the convention in
Eq.~\eqref{eq:charge-oriented-center-label} agrees with the modular
conventions of Section~\ref{sec:Modular_Limit}, while the endpoint charge is
the simple difference $a-b$.
Here $\theta_a^{(k)}$ is the categorical twist scalar; it is unrelated to the
highest root $\theta$ used in the alcove inequalities.
The duality identities in \eqref{eq:dual-label-modular-identities} are the
standard unitary ribbon-category identities; see
\cite[Chs.~2--3]{Bakalov:2001}.

The finite set sampled at level $k$ is
\begin{equation}
 \Lambda_k
 :=\left\{\left(\frac{a+b+2\rho}{2r_k},a-b\right):
             a,b\in P_k^+\right\}\subset Y\times P.
 \label{eq:regular-lattice}
\end{equation}
For a fixed charge $\eta\in P$, its exact slice is
\begin{equation}
 \Lambda_{k,\eta}
 =\left\{\left(\frac{\mu+\rho+\eta/2}{r_k},\eta\right):
        \mu\in P_k^+,\ \mu+\eta\in P_k^+\right\}.
 \label{eq:fixed-charge-admitted-mesh}
\end{equation}
The ambient flux affine lattice in this formula is
\begin{equation}
 r_k^{-1}(P+\rho+\eta/2).
 \label{eq:fixed-charge-full-mesh-preview}
\end{equation}
The integrability conditions in
\eqref{eq:fixed-charge-admitted-mesh} are essential near an alcove wall.  On
a compact set $K\Subset Y$, however, they eventually impose no further
deletion.  To see this, set
\begin{equation}
 \delta_K:=\min_{x\in K}
 \min\bigl\{
   \langle x,\alpha_i^\vee\rangle\ (1\leq i\leq\ell),
   1-\langle x,\theta^\vee\rangle
 \bigr\}>0.
 \label{eq:compact-alcove-wall-distance}
\end{equation}
If $x=(\mu+\rho+\eta/2)/r_k$ belongs to the ambient mesh, the two recovered
weights are $b=\mu$ and $a=\mu+\eta$.  For each simple coroot,
\begin{align}
 \langle\mu,\alpha_i^\vee\rangle
 &=r_k\langle x,\alpha_i^\vee\rangle-1
   -\frac12\langle\eta,\alpha_i^\vee\rangle,\notag\\
 \langle\mu+\eta,\alpha_i^\vee\rangle
 &=r_k\langle x,\alpha_i^\vee\rangle-1
   +\frac12\langle\eta,\alpha_i^\vee\rangle.
 \label{eq:fixed-charge-lower-wall-bounds}
\end{align}
Both are bounded below by
$r_k\delta_K-1-\frac12|\langle\eta,\alpha_i^\vee\rangle|$.  At the affine
wall one similarly has
\begin{align}
 k-\langle\mu,\theta^\vee\rangle
 &=r_k\bigl(1-\langle x,\theta^\vee\rangle\bigr)-1
   +\frac12\langle\eta,\theta^\vee\rangle,\notag\\
 k-\langle\mu+\eta,\theta^\vee\rangle
 &=r_k\bigl(1-\langle x,\theta^\vee\rangle\bigr)-1
   -\frac12\langle\eta,\theta^\vee\rangle.
 \label{eq:fixed-charge-affine-wall-bounds}
\end{align}
These quantities are positive once $k$ is large enough, uniformly for
$x\in K$.  The reverse inclusion is immediate from the definition of
$\Lambda_{k,\eta}$, and therefore
\begin{equation}
 \{x:(x,\eta)\in\Lambda_{k,\eta}\}\cap K
 =r_k^{-1}(P+\rho+\eta/2)\cap K
 \label{eq:fixed-charge-local-full-mesh}
\end{equation}
for all sufficiently large $k$.  Thus compact regular flux windows are
eventually sampled by the full translated weight lattice, while the charge
coordinate remains discrete.

The cell volume of this fixed-charge mesh is
\begin{equation}
 \operatorname{covol}(r_k^{-1}P)
 =r_k^{-\ell}\operatorname{covol}(P).
 \label{eq:fixed-charge-mesh-volume}
\end{equation}
This determines the square-root normalization of sampled wavefunctions
in Section~\ref{sec:Modular_Limit}: summing their squared coefficients
must approximate the continuum integral over flux.  In particular, the
factor $1/2$ in $x_{a,b}$ does not halve the mesh spacing at fixed charge.
Increasing both weights by $p\in P$ changes $x$ by $p/r_k$.

For example, take $G=SU(2)$, write its fundamental weight as
$\varpi=\alpha/2$, and normalize $(\alpha,\alpha)=2$.
Then $P=\mathbb Z\varpi$, $Q^\vee=2\mathbb Z\varpi$,
$\rho=\varpi$, and $h^\vee=2$.  A point of $Y$ is $x=s\varpi$ with
$0<s<1$.  The finite labels $a=n\varpi$, $b=m\varpi$ have
$0\leq n,m\leq k$, and
\begin{equation}
 s=\frac{n+m+2}{2(k+2)},\qquad \eta=(n-m)\varpi.
 \label{eq:SU2-flux-charge-example}
\end{equation}
At charge $\eta=q\varpi$, this reduces to
$s=(m+1+q/2)/(k+2)$, with $m,m+q\in\{0,\ldots,k\}$.
Thus the spacing in $s$ is $1/(k+2)$ and its offset depends on $q$.
Since $|\varpi|=1/\sqrt2$, the flat measure and cell volume are
$dx=ds/\sqrt2$ and $1/[\sqrt2(k+2)]$, respectively.  This elementary
example exhibits both the continuous flux and the integer charge that
survive the limit.

We call a sequence of labels $(a_k,b_k)$ a regular large-$k$ sequence if
\begin{equation}
 x_{a_k,b_k}\longrightarrow x\in Y,
 \qquad \eta_{a_k,b_k}=\eta
 \quad\text{for all sufficiently large }k,
 \label{eq:regular-large-level-sequence}
\end{equation}
for some $(x,\eta)\in Y\times P$.
Since $P$ is discrete, the second condition is equivalent to
convergence of the charge to $\eta$.

To formulate convergence of the corresponding operators, we work
on a common Hilbert space.
A sequence of bounded operators $A_k$ converges strongly to $A$
if $\|(A_k-A)\psi\|\to0$ for every fixed vector $\psi$.
This is convergence in the strong operator topology (SOT).
Strong-$*$ convergence requires, in addition, that
$A_k^*\to A^*$ strongly.

In the genus-one convergence theorem,
Theorem~\ref{thm:strong-modular}, we apply these notions
to the modular operators after embedding the finite-level state
spaces into $\mathcal H_0$.
For each fixed normalizable wave packet, the finite-level modular
action then converges in Hilbert-space norm to the limiting action.

% section_03.tex
\section{Strong regular genus-one convergence}
\label{sec:Modular_Limit}

On a spatial torus, $\mathcal S$ exchanges the two cycles and $\mathcal T$
performs a Dehn twist.  In Chern--Simons quantization their matrix elements
are respectively related to Hopf-link amplitudes and line spins
\cite{Witten:1988hf,Elitzur:1989nr}.  These operations provide a quantum test
of the contraction in Section~\ref{sec:Limit_Regular_Simples}: a state with a
smooth distribution of regular magnetic flux and finitely many electric
charges should have the same modular transition amplitudes in the
large-level doubled theory and in $BF$ theory.

A comparison of individual matrix entries does not answer this question.
The number of finite labels grows with $k$, whereas the limiting flux is
continuous.  A normalizable limiting state must therefore be compared with
a superposition of finite-level states, with amplitudes scaled by the square
root of the flux-cell volume.  We construct this comparison explicitly.
Theorem~\ref{thm:strong-modular} states convergence in Hilbert-space norm under any fixed
modular transformation, for every fixed normalizable state in the regular
$BF$ Hilbert space.  Its adjoint statement treats the inverse transformation
at the same time.

The proof of Theorem~\ref{thm:strong-modular} proceeds from local
convergence to a global operator limit.  The endpoint unitarity and
modular relations used below are established independently in
Appendix~\ref{app:endpoint-Fourier-unitarity}.
In Subsection~\ref{subsec:kp-poisson}, we separate the doubled
Kac--Peterson formula into equal-Weyl terms, which reproduce the $BF$
Fourier kernel through Riemann sums, and unequal-Weyl terms, which vanish
after smearing by Poisson summation.  This gives local weighted
convergence on compact regular flux windows with finitely many charges
(Proposition~\ref{prop:local-S-convergence}).
In Subsection~\ref{subsec:no-escape}, global sampling estimates,
together with unitarity of both the finite and limiting transforms,
rule out a loss of probability to arbitrarily large charges or
concentration near the alcove walls.  This upgrades the local result
to global norm convergence for sampled smooth test states
(Theorem~\ref{thm:no-escape}).
In Subsection~\ref{subsec:common-hilbert-modular}, the cell approximation
of Lemma~\ref{lem:global-cell-approximation} transfers this convergence
to the common Hilbert space $\mathcal H_0$; density and uniform
boundedness then extend it to every fixed state in $\mathcal H_0$.
Together with the exact sampled twist formula, this proves strong
convergence of the modular generators.  Unitarity gives convergence
of their adjoints, and finite products yield strong-$*$ convergence
for every fixed element of $SL(2,\mathbb Z)$.
Finally, Subsection~\ref{subsec:embedding-independence} shows that the
same limit holds for every admissible family of sampling embeddings
(Lemma~\ref{lem:admissible-embedding-robustness}), completing the
embedding-independence assertion of the theorem.

Recall from Eq.~\eqref{eq:weight-coroot-duality} that the geometric inner
product of Eq.~\eqref{eq:geometric-Cartan-form}
identifies the weight lattice with the Euclidean dual lattice,
$P=(Q^\vee)^*$.  This geometric inner product determines Lebesgue measure on
$\mathfrak t$ and hence all covolumes below.  In the
long-root-length-two convention the coroot lattice is integral, so
$Q^\vee\subset P$.  Consequently,
\begin{equation}
 \operatorname{covol}(Q^\vee)
 =|P/Q^\vee|\operatorname{covol}(P).
 \label{eq:dual-lattice-covolume-identity}
\end{equation}
Writing $\ell=\operatorname{rank}G$, set
\begin{equation}
 v_k=\operatorname{covol}(P)r_k^{-\ell},\qquad
 \gamma_G=\operatorname{covol}(Q^\vee)^{-1}.
 \label{eq:lattice-cell-normalization}
\end{equation}
Thus $v_k$ is the volume of a fundamental cell of $r_k^{-1}P$.  Let
$F_{Q^\vee}\subset\mathfrak t$ be a measurable fundamental domain for
$Q^\vee$.  Normalized Haar measure on
$T=\mathfrak t/Q^\vee$ is characterized by
\begin{equation}
 \int_T H(t)\,dm_T(t)
 =\gamma_G\int_{F_{Q^\vee}}H([y])\,dy
 \label{eq:normalized-torus-Haar}
\end{equation}
for every integrable $H$.  We occasionally abbreviate this relation by
$dm_T=\gamma_G\,dy$ when working in a torus coordinate chart.

The finite and limiting Hilbert spaces are
\begin{equation}
 \mathcal H_k=\ell^2(\Lambda_k),\qquad
 \mathcal H_0=L^2(Y\times P,dx\otimes\#)
 \cong\bigoplus_{\eta\in P}L^2(Y,dx),
 \label{eq:common-Hilbert-spaces}
\end{equation}
where $\#$ is counting measure on the discrete lattice $P$.  In particular,
\[
 \|g\|_{\mathcal H_0}^2
 =\sum_{\eta\in P}\int_Y|g(x,\eta)|^2\,dx.
\]
The space $\mathcal H_k$ carries the canonical orthonormal basis
$\{\delta_{(x,\eta)}:(x,\eta)\in\Lambda_k\}$.  If
$(x,\eta)\leftrightarrow(a,b)$ and
$(y,\omega)\leftrightarrow(c,d)$ under Eq.~\eqref{eq:regular-lattice}, the
finite doubled modular $S$-operator has, in this basis, the matrix
\begin{equation}
 (\mathcal S_k)_{(x,\eta),(y,\omega)}
 :=S^{(k)}_{ac}\,\overline{S^{(k)}_{bd}}
 =:S^{(k,-k)}_{(x,\eta),(y,\omega)}.
 \label{eq:finite-doubled-S-operator}
\end{equation}
Here the complex conjugate in the second factor reflects the reversed second
chiral theory.

For a pointwise-defined function $g:Y\times P\to\mathbb C$, define the
restriction and weighted sampling maps
\begin{equation}
 Q_kg=g|_{\Lambda_k},\qquad J_kg=\sqrt{v_k}\,Q_kg.
 \label{eq:weighted-sampling-map}
\end{equation}
Both $Q_k g$ and $J_k g$ are vectors in $\mathcal H_k$. We use these maps on smooth test functions; they are not defined on
arbitrary $L^2$ equivalence classes.  The factor $\sqrt{v_k}$ ensures that
$\|J_kg\|_{\mathcal H_k}^2=v_k\sum_{z\in\Lambda_k}|g(z)|^2$ is the Riemann sum
corresponding to the $L^2$ norm displayed above.

Initially we sample the test-function space
\[
 C_c^\infty(Y\times P)
 :=C_c^\infty(Y)\otimes_{\mathrm{alg}}c_{00}(P).
\]
Here $c_{00}(P)$ denotes the finitely supported complex functions on $P$.
Thus a test function has finitely many nonzero charge fibers, each in
$C_c^\infty(Y)$.

Throughout this section, $|\mathord\cdot|$ is the norm induced by the
geometric inner product of Eq.~\eqref{eq:geometric-Cartan-form}, also used
on $\mathfrak t^*$ through the duality identification.  We write
$\mathcal A(Y\times P)$ for the charge-Schwartz class: every fiber
$g(\mathord\cdot,\eta)$ extends smoothly to a neighborhood of $\overline Y$,
and, for every multi-index $\alpha$ and every $N\geq0$,
\begin{equation}
 \sup_{\eta\in P}(1+|\eta|)^N
 \sup_{x\in\overline Y}|\partial_x^\alpha g(x,\eta)|<\infty\,,
 \label{eq:charge-Schwartz}
\end{equation}
thus its flux derivatives decay faster than any inverse power of the charge, uniformly over the closed alcove. The charge-Schwartz class is unchanged if this norm is replaced by any
equivalent norm.  Here $C^\infty(\overline Y)$ means restriction from a
neighborhood of $\overline Y$.  The class contains
$C^\infty(\overline Y)\otimes_{\mathrm{alg}}c_{00}(P)$ and is dense in
$\mathcal H_0$.  The same sampling maps are used on
$\mathcal A(Y\times P)$; Lemma~\ref{lem:global-sampling} provides uniform
control of their norms.

To compare finite-level states with wave packets in $\mathcal H_0$,
we replace each discrete basis state by a normalized, piecewise constant
wavefunction. Fix a bounded half-open fundamental domain $F_P$
of $P$ with $0\in\operatorname{int}F_P$.\footnote{
Here \emph{fundamental} means that every $y\in\mathfrak t$ admits
a unique decomposition $y=u+p$, with $u\in F_P$ and $p\in P$.
Equivalently, $\mathfrak t=\bigsqcup_{p\in P}(F_P+p)$, where
$\bigsqcup$ denotes a disjoint union. Thus $F_P$ contains exactly
one representative of each equivalence class modulo $P$.
The half-open convention avoids duplicate representatives on
cell boundaries.}
For example, given a lattice basis $p_1,\ldots,p_\ell$ of $P$,
we may take
\[
 F_P=\left\{\sum_{j=1}^{\ell}t_jp_j:
                 -\tfrac12\leq t_j<\tfrac12\right\}.
\]
For a sampling point $z=(x,\eta)\in\Lambda_k$, rescale this reference
cell by the factor $r_k^{-1}$, translate it to $x$, and retain its intersection
with the alcove:
\begin{equation}
 \begin{gathered}
 C_{k,z}=Y\cap(x+r_k^{-1}F_P),\qquad
 I_k\delta_z=|C_{k,z}|^{-1/2}1_{C_{k,z}}\delta_\eta\,,
 \end{gathered}
 \label{eq:lattice-cell-isometry}
\end{equation}
Here $|C_{k,z}|$ denotes the $dx$-volume of the cell, $1_{C_{k,z}}$
is its indicator function, and $\delta_\eta$ selects the charge fiber
$\eta$. Thus $I_k\delta_z$ is constant on that cell, vanishes
elsewhere, and has unit norm. Cells crossing an alcove wall are
truncated by the intersection with $Y$.

Since $x\in Y$ and $0\in\operatorname{int}F_P$, every $C_{k,z}$
has positive volume. At fixed charge, the sampling points lie in
a translate of $r_k^{-1}P$, so their cells are disjoint up to
null sets. Different charge fibers are orthogonal. The images of
the basis vectors are therefore orthonormal, and
$I_k:\mathcal H_k\to\mathcal H_0$ is an isometry.
Its Hilbert-space adjoint is given by
\begin{equation}
 (I_k^*f)(z)=|C_{k,z}|^{-1/2}
             \int_{C_{k,z}}f(y,\eta)\,dy,
 \qquad z=(x,\eta).
 \label{eq:cell-embedding-adjoint}
\end{equation}
We introduce $\mathsf P_k=I_kI_k^*$ as the orthogonal projection operator onto
$I_k\mathcal H_k$: it replaces $f$ by its average on each cell,
\[
 (\mathsf P_k f)(y,\eta)
 =\frac{1}{|C_{k,z}|}\int_{C_{k,z}}f(u,\eta)\,du,
 \qquad y\in C_{k,z},
\]
and is zero outside the union of the cells in that charge fiber.
These formulas are understood almost everywhere.

At finite $k$, the cells need not cover all of $Y$.
Lemma~\ref{lem:global-cell-approximation} controls the uncovered
region and the truncated cells, and proves
$\|\mathsf P_kf-f\|_{\mathcal H_0}\to0$ for every fixed
$f\in\mathcal H_0$. Thus the embedded finite-level states
approximate arbitrary normalizable wave packets in norm.

Before defining the endpoint operators, we realize $\mathcal H_0$
as a Hilbert space of square-integrable wavefunctions on the moduli
space of flat $G$-connections on the torus,
\begin{equation}
    \mathcal M_{\mathrm{flat}}(T^2)\simeq (T\times T)/W.
\end{equation}
We work on the subset where the first holonomy is regular. Let
\[
    q:T_{\mathrm{reg}}\times T
    \longrightarrow (T_{\mathrm{reg}}\times T)/W
\]
be the quotient map. The diagonal Weyl action is free because
the first component is regular. We equip the quotient with
\begin{equation}
    m_{\mathrm{quot}}
    :=q_*\bigl((dm_T|_{T_{\mathrm{reg}}})\otimes dm_T\bigr),
    \label{eq:regular-BF-quotient-Haar}
\end{equation}
the pushforward of normalized product Haar measure.

Here the pushforward assigns to each measurable set in the quotient
the measure of its full preimage. Indeed, the full preimage of a measurable set in the quotient
consists of $|W|$ Weyl-related lifts, each with the same product
Haar measure. Since $dm_T=\gamma_G\,dx$ in each alcove chart,
summing these contributions gives
$|W|\gamma_G\,dx\otimes dm_T$. Thus, under the identification
$(T_{\mathrm{reg}}\times T)/W\simeq Y\times T$, the measure
$m_{\mathrm{quot}}$ is represented by
$|W|\gamma_G\,dx\otimes dm_T$.
This differs by a factor of $|W|$ from simply restricting product
Haar measure to a single alcove, which gives
$\gamma_G\,dx\otimes dm_T$.

For later use, define the diagonal Weyl extension
\[
 \operatorname{Ext}:L^2(Y\times P)\longrightarrow L^2(T\times P)^W,
 \qquad
 (\operatorname{Ext}f)(wx,w\eta)=f(x,\eta).
\]
Here $x\in Y$, $\eta\in P$, and $w\in W$, and $wx$ is understood as a point
of $T=\mathfrak t/Q^\vee$.  Every regular
Weyl orbit meets $Y$ once, so the formula is unambiguous on
$T_{\rm reg}\times P$.  The union of the root hypertori,
\[
    T\setminus T_{\mathrm{reg}}
    =\bigcup_{\alpha\in\Delta_+}
      \{t\in T:\chi_\alpha(t)=1\},
\]
is Haar-null. We set $\mathrm{Ext}\,f=0$ on
$(T\setminus T_{\mathrm{reg}})\times P$; this choice does not
affect its $L^2$ equivalence class. All identities involving
$\mathrm{Ext}$ are understood almost everywhere.

The two descriptions of a torus state have different physical uses.  In a
holonomy representation one specifies the commuting $A$-holonomies around
both cycles.  Fourier transformation of the second holonomy replaces it by
an electric charge of the centralizer of the first.  For regular first
holonomy that centralizer is $T$, so its charges are precisely $P$.
The proposition identifies these representations and fixes their relative
normalization.  The quotient measure used here comes from Haar measure on
$T\times T$; it is not the pushforward of Haar measure on $G$ with its Weyl
Jacobian.  No measure on an ambient space of noncommuting pairs is being
restricted to its measure-zero commuting locus.
\begin{proposition}[Regular quotient--Haar model]
\label{prop:regular-BF-Hilbert-model}
There are natural unitary identifications
\begin{equation}
 L^2\!\left((T_{\rm reg}\times T)/W,m_{\rm quot}\right)
 \simeq L^2(T\times T)^W
 \simeq L^2(T\times P)^W
 \simeq L^2(Y\times P,dx\otimes\#)=\mathcal H_0.
 \label{eq:regular-BF-Hilbert-chain-early}
\end{equation}
The left-hand term is the quotient--Haar $L^2$-space on commuting holonomy pairs whose first holonomy is regular. Its measure is the pushforward
$m_{\rm quot}$ defined in Eq.~\eqref{eq:regular-BF-quotient-Haar}.
\end{proposition}

\begin{proof}
A commuting pair whose first member is regular is simultaneously conjugate
to a pair in $T_{\rm reg}\times T$, with residual ambiguity $W$.  This gives
the Borel identification
\(
\{(g,h):gh=hg,\ g\in G_{\rm reg}\}/G
\simeq(T_{\rm reg}\times T)/W
\), as reviewed in \cite[Sec.~3]{Jeffrey:2005}.
Because the measure in Eq.~\eqref{eq:regular-BF-quotient-Haar} is a
pushforward, pullback along $q$ is an isometry onto the invariant subspace.
Root hypertori are Haar-null, so $T_{\rm reg}$ may be replaced by $T$ in
$L^2$.  Explicitly, if $\Psi$ is an invariant holonomy wavefunction, take
its second-variable Fourier coefficients
\begin{equation}
 (\mathscr P_2\Psi)(x,\eta)
 =\int_T\Psi(x,t)e^{-2\pi i(\eta,t)}\,dm_T(t).
 \label{eq:partial-holonomy-Fourier}
\end{equation}
The characters $e^{2\pi i(\eta,t)}$, $\eta\in P$, form an orthonormal basis
of $L^2(T,dm_T)$.  Parseval's identity, followed by integration in $x$,
gives
\[
 \int_{T\times T}|\Psi(x,t)|^2\,dm_T(x)dm_T(t)
 =\sum_{\eta\in P}\int_T|(\mathscr P_2\Psi)(x,\eta)|^2\,dm_T(x).
\]
The formula first holds for finite Fourier sums and extends to all of $L^2$
by completion.  Changing $t$ to $wt$ proves
$(\mathscr P_2\Psi)(wx,w\eta)=(\mathscr P_2\Psi)(x,\eta)$.
The inverse Fourier series gives surjectivity onto the invariant subspace,
which proves the middle identification.  This is the compact-torus case of
Pontryagin--Plancherel theory
\cite[Ch.~4]{Folland:1995}.
Modulo the root hypertori, the regular torus is the disjoint union of the
Weyl translates of the open alcove,
\[
 T_{\mathrm{reg}}=\bigsqcup_{w\in W}wY
 \qquad\text{up to a Haar-null set}.
\]
In each alcove chart, Eq.~\eqref{eq:Haar-integral-normalization} gives
$dm_T=\gamma_G\,dx$.  Since
$(\operatorname{Ext}f)(wx,w\eta)=f(x,\eta)$,
\[
 \|\operatorname{Ext}f\|_{L^2(T\times P)}^2
 =\gamma_G\sum_{w\in W}\sum_{\eta\in P}\int_Y
 |f(x,\eta)|^2\,dx
 =|W|\gamma_G\|f\|_{L^2(Y\times P)}^2.
\]
Thus
\[
 U:=(|W|\gamma_G)^{-1/2}\operatorname{Ext}:
 \mathcal H_0\longrightarrow L^2(T\times P)^W
\]
is the last unitary in Eq.~\eqref{eq:regular-BF-Hilbert-chain-early}.  Its
normalization is determined by $m_{\rm quot}$, normalized Haar measure on
$T$, and counting measure on $P$.
This is the regular Fourier--Haar chart appearing in
\cite{Jia:2026bbo}; the cotangent interpretation is described in
\cite{Baez:1999sr}.
\end{proof}

We define the endpoint generator $\mathcal T_0$ on all of
$\mathcal H_0$ and initially define $\mathcal S_0$ on the dense
test space $C_c^\infty(Y\times P)$ by
\begin{align}
 (\mathcal S_0f)(x,\eta)
 &=\gamma_G\sum_{\omega\in P}\int_Y
   \sum_{w\in W}
   e^{-2\pi i((wx,\omega)+(w\eta,y))}
   f(y,\omega)\,dy,
 \label{eq:BF-S-operator}\\
 (\mathcal T_0f)(x,\eta)
 &=e^{2\pi i(\eta,x)}f(x,\eta).
 \label{eq:BF-T-operator}
\end{align}
The sum over $\omega$ in the first line is finite on this test
space. The operator $\mathcal T_0$ is unitary on $\mathcal H_0$:
its multiplier has unit modulus, and its inverse is multiplication
by $e^{-2\pi i(\eta,x)}$.
The operator $\mathcal S_0$ extends uniquely to a unitary on
$\mathcal H_0$.  Together with $\mathcal T_0$, it satisfies
\[
 \mathcal S_0^4=1,\qquad
 (\mathcal S_0\mathcal T_0)^3=\mathcal S_0^2,
\]
and hence defines a unitary representation
$\rho_0:SL(2,\mathbb Z)\longrightarrow U(\mathcal H_0)$ with
$\rho_0(S)=\mathcal S_0$ and $\rho_0(T)=\mathcal T_0$.
These endpoint properties are established independently of the
finite-level limit in Proposition~\ref{prop:regular-Fourier-unitarity},
whose Fourier construction and proof are given in
Appendix~\ref{app:endpoint-Fourier-unitarity}.  In particular,
$\|\mathcal S_0f\|_{\mathcal H_0}=\|f\|_{\mathcal H_0}$ for every
$f\in\mathcal H_0$.

In the holonomy representation of
Proposition~\ref{prop:regular-BF-Hilbert-model}, the same modular action
is induced by integral linear transformations of the two holonomies.
The first-holonomy-regular chart $T_{\rm reg}\times T$ need not be
preserved pointwise, but it contains a common modular-invariant conull
subset.  Removing its null complement does not change the
quotient--Haar $L^2$ space.  The relation between the two representations
and the precise conull-domain statement are explained in
Appendix~\ref{app:regular-holonomy-action}, especially
Remark~\ref{rem:regular-chart-ae-modular}.

The $S$-generator of the doubled modular data is $\mathcal S_k$, and its
$T$-generator is the categorical twist
\begin{equation}
 \mathcal T_k\delta_{(x_{a,b},\eta_{a,b})}
 =\theta_a\theta_b^{-1}
  \delta_{(x_{a,b},\eta_{a,b})}.
 \label{eq:finite-doubled-twist}
\end{equation}
Indeed, in the reversed factor
$S^{\rm rev}_{\bar b,\bar d}=\overline{S_{bd}}$ and
$\theta^{\rm rev}_{\bar b}=\theta_{\bar b}^{-1}=\theta_b^{-1}$; this is the
matrix convention fixed in Eqs.~\eqref{eq:dual-label-modular-identities} and
\eqref{eq:finite-doubled-S-operator}.  The scalar vacuum phases of the two
chiral $T$-matrices cancel because the chiralities are opposite.
More explicitly, the single-chiral RCFT $T$-matrix is
\begin{equation}
 (T_k^{\mathrm{RCFT}})_{aa}
 =e^{-2\pi i c_k/24}\theta_a,
 \label{eq:RCFT-versus-categorical-twist}
\end{equation}
where
\[
 c_k=\frac{k\,\dim G}{k+h^\vee}
\]
is the level-$k$ WZW central charge.  The reversed chiral factor carries the
complex-conjugate scalar.
Their product has diagonal entry $\theta_a\theta_b^{-1}$, with no residual
central-charge phase.
For the reverse braided and ribbon conventions, see
\cite[Secs.~8.1 and~8.10]{Etingof:2015}.

The doubled action is genuine.  Indeed, for a modular category $\mathcal D$,
put the global dimension and the Gauss sums
\begin{equation}
 \operatorname{Dim}\mathcal D
 :=\sum_{X\in\operatorname{Irr}(\mathcal D)}d_X^2,
 \qquad
 p_\pm(\mathcal D)
 :=\sum_{X\in\operatorname{Irr}(\mathcal D)}
       \theta_X^{\,\pm1}d_X^2,
 \label{eq:modular-Gauss-sums}
\end{equation}
where $\operatorname{Irr}(\mathcal D)$ is the set of simple objects and $d_X$ and
$\theta_X$ are the quantum dimension and the twist of $X$.
M\"uger's Gauss-sum formula
\cite[Theorem~1.2 and Remark~5.19]{Mueger:2001crc}
gives
\[
 \Delta_\pm(Z(\mathcal C_k))=\operatorname{Dim}\mathcal C_k,
 \qquad
 \operatorname{Dim}Z(\mathcal C_k)
 =\bigl(\operatorname{Dim}\mathcal C_k\bigr)^2.
\]
Thus $Z(\mathcal C_k)$ has multiplicative central charge one; the normalized
doubled matrices obey
$\mathcal S_k^4=1$ and
$(\mathcal S_k\mathcal T_k)^3=\mathcal S_k^2$.  They therefore define a
unitary representation
$\rho_k:SL(2,\mathbb Z)\to U(\mathcal H_k)$.

Let $I_k:\mathcal H_k\to\mathcal H_0$ be the cell isometry defined in
Eq.~\eqref{eq:lattice-cell-isometry}, and let
$\mathsf P_k=I_kI_k^*$ be the orthogonal projection onto its image.  We
transport $\rho_k$ to $\mathsf P_k\mathcal H_0$ and let it act as the identity
on the orthogonal complement:
\begin{equation}
 \rho_k^\#(\gamma)
 :=I_k\rho_k(\gamma)I_k^*+(1-\mathsf P_k).
 \label{eq:foreground-embedded-representation}
\end{equation}
Relative to $\mathcal H_0=I_k\mathcal H_k\oplus(I_k\mathcal H_k)^\perp$,
this operator is unitarily equivalent to $\rho_k(\gamma)\oplus1$.
Moreover, $I_k^*I_k=1$ and
$I_k^*(1-\mathsf P_k)=(1-\mathsf P_k)I_k=0$, so
\begin{equation}
 \rho_k^\#(\gamma_1)\rho_k^\#(\gamma_2)
 =I_k\rho_k(\gamma_1\gamma_2)I_k^*+(1-\mathsf P_k)
 =\rho_k^\#(\gamma_1\gamma_2).
 \label{eq:embedded-rho-representation}
\end{equation}
Since $\rho_k^\#(1)=1$, this gives a unitary representation on
$\mathcal H_0$.

In physical terms, the following theorem compares the same prepared
wavepacket at different levels.  For normalized $\psi,\varphi\in\mathcal H_0$
it implies convergence of the amplitude
$\langle\varphi,\rho_k^\#(\gamma)\psi\rangle$ to the $BF$ amplitude, and
therefore of the corresponding transition probability.  Strong convergence
means $\|\rho_k^\#(\gamma)\psi-\rho_0(\gamma)\psi\|\to0$ for each fixed
$\psi$; it does not assert a bound uniform over every normalized state.
In particular, states whose support approaches a singular flux or whose
charge grows with $k$ require a separate scaling analysis.  The extension
by the identity on $(I_k\mathcal H_k)^\perp$ makes each finite-level action
unitary on $\mathcal H_0$; the cell approximation will show that this
complement has vanishing projection on any fixed state as $k\rightarrow \infty$.

\begin{theorem}[Regular genus-one modular convergence]
\label{thm:strong-modular}
Let $\rho_0$ be the endpoint unitary representation on $\mathcal H_0$
fixed above by $\rho_0(S)=\mathcal S_0$ and
$\rho_0(T)=\mathcal T_0$, with generators given by
Eqs.~\eqref{eq:BF-S-operator}--\eqref{eq:BF-T-operator}.
Then every
$\rho_k^\#$ in Eq.~\eqref{eq:foreground-embedded-representation} is a
unitary representation on the same Hilbert space. For each
$\gamma\in SL(2,\mathbb Z)$,
\begin{equation}
 \rho_k^\#(\gamma)
 \xrightarrow[k\to\infty]{\mathrm{SOT}}\rho_0(\gamma),
 \qquad
 \rho_k^\#(\gamma)^*
 \xrightarrow[k\to\infty]{\mathrm{SOT}}\rho_0(\gamma)^*.
 \label{eq:strong-modular}
\end{equation}
Thus the embedded genus-one operators converge pointwise in the strong-$*$
topology on the regular label space $\mathcal H_0$.

More generally, let $(\widetilde I_k)$ be any admissible family of
sampling isometries in the sense of
Eq.~\eqref{eq:admissible-embedding-hypotheses}.
The tilde notation distinguishes this general admissible family
from the specific family $(I_k)$ fixed above. Put
\begin{equation}
 \widetilde{\mathsf P}_k=\widetilde I_k\widetilde I_k^*,
 \qquad
 \widetilde\rho_k^\#(\gamma)
 =\widetilde I_k\rho_k(\gamma)\widetilde I_k^*
  +(1-\widetilde{\mathsf P}_k).
 \label{eq:alternative-embedded-representation}
\end{equation}
For every $k$, this is a unitary representation of
$SL(2,\mathbb Z)$ on $\mathcal H_0$. For each fixed
$\gamma\in SL(2,\mathbb Z)$,
\begin{equation}
 \widetilde\rho_k^\#(\gamma)
 \xrightarrow[k\to\infty]{\mathrm{SOT}}\rho_0(\gamma),
 \qquad
 \widetilde\rho_k^\#(\gamma)^*
 \xrightarrow[k\to\infty]{\mathrm{SOT}}\rho_0(\gamma)^*.
 \label{eq:alternative-strong-modular-limit}
\end{equation}
Lemma~\ref{lem:admissible-embedding-robustness} proves these two limits.
\end{theorem}

\subsection{The Kac--Peterson kernel and Poisson decay}
\label{subsec:kp-poisson}

With long roots normalized to have squared length $2$, the Kac--Peterson matrix
\cite{Kac:1984mq,Gannon:1995uv} is
\begin{equation}
 S^{(k)}_{\lambda\mu}
 =\alpha_k^{\mathrm{KP}}\sum_{w\in W}\epsilon(w)
 e^{-\frac{2\pi i}{r_k}(w(\lambda+\rho),\mu+\rho)},
 \qquad
 |\alpha_k^{\mathrm{KP}}|^2
 =\frac{1}{r_k^{\ell}|P/Q^\vee|}.
 \label{eq:KP-normalization}
\end{equation}
This is Eq.~(3.3d) of Ref.~\cite{Gannon:1995uv} in our conventions.
The phase of $\alpha_k^{\mathrm{KP}}$ cancels against its complex
conjugate in the doubled product
$S^{(k)}_{ac}\overline{S^{(k)}_{bd}}$.

To express this product in flux--charge coordinates, we substitute
\begin{equation}
 \begin{aligned}
  a+\rho&=r_kx+\eta/2,&\qquad b+\rho&=r_kx-\eta/2,\\
  c+\rho&=r_ky+\omega/2,& d+\rho&=r_ky-\omega/2.
 \end{aligned}
 \label{eq:four-flux-charge-substitutions}
\end{equation}
For a fixed Weyl pair $(u,v)$, the combined exponent from the first
chiral factor and the complex conjugate of the second is
\begin{align*}
 &-\frac{2\pi i}{r_k}
   \bigl(u(r_kx+\eta/2),r_ky+\omega/2\bigr)\\
 &\hspace{35mm}
 +\frac{2\pi i}{r_k}
   \bigl(v(r_kx-\eta/2),r_ky-\omega/2\bigr).
\end{align*}
Expanding these products gives
\begin{align}
 S^{(k,-k)}_{(x,\eta),(y,\omega)}
 =|\alpha_k^{\mathrm{KP}}|^2\sum_{u,v\in W}\epsilon(u)\epsilon(v)
 \exp\{&-2\pi ir_k((u-v)x,y)
 -\pi i((u+v)x,\omega)\notag\\
 &-\pi i((u+v)\eta,y)
 -\tfrac{\pi i}{2r_k}((u-v)\eta,\omega)\}.
 \label{eq:S_diag}
\end{align}
In the summation, the terms with equal Weyl elements $u=v=w$ contribute $|\alpha_k^{\mathrm{KP}}|^2 K_{\rm diag}$, where the kernel is
\begin{equation}
 K_{\rm diag}((x,\eta),(y,\omega))
 =\sum_{w\in W}
 e^{-2\pi i((wx,\omega)+(w\eta,y))}.
 \label{eq:BF-diagonal-kernel}
\end{equation}

Before moving on, let us consider the example with $G=SU(2)$. Let $\varpi=\alpha/2$ be the fundamental weight, so that
$(\varpi,\varpi)=1/2$, and write
$x=t\varpi$, $y=u\varpi$, $\eta=n\varpi$, $\omega=m\varpi$.
Then $0<t,u<1$, $n,m\in\mathbb Z$, $dx=dt/\sqrt2$, and
$\gamma_{SU(2)}=1/\sqrt2$.  The two Weyl elements give
\begin{equation}
 (\mathcal S_0f)(t,n)
 =\sum_{m\in\mathbb Z}\int_0^1
     \cos\,\!\bigl(\pi(tm+nu)\bigr)f(u,m)\,du,
 \qquad
 (\mathcal T_0f)(t,n)=e^{\pi int}f(t,n).
 \label{eq:SU2-endpoint-kernel-check}
\end{equation}
Here $f(t,n)$ abbreviates $f(t\varpi,n\varpi)$; its norm still uses
$dt/\sqrt2$.  At level $k$, put $r=k+2$.  If the integers
$a,b,c,d\in\{0,\ldots,k\}$ denote the chiral highest weights in units of
$\varpi$, then
\[
 t=\frac{a+b+2}{2r},\quad n=a-b,\qquad
 u=\frac{c+d+2}{2r},\quad m=c-d.
\]
The familiar $SU(2)$ matrix
$S_{ac}=\sqrt{2/r}\sin(\pi(a+1)(c+1)/r)$ consequently gives
\begin{equation}
 S_{ac}S_{bd}
 =\frac1r\left[
    \cos\,\!\bigl(\pi(tm+nu)\bigr)
    -\cos\!\left(2\pi rtu+\frac{\pi nm}{2r}\right)
   \right].
 \label{eq:SU2-doubled-kernel-check}
\end{equation}
At fixed input charge the spacing of $u$ is $1/r$.  The first term is
therefore the Riemann sum for Eq.~\eqref{eq:SU2-endpoint-kernel-check};
the second oscillates rapidly for regular output flux $t$ as $r\rightarrow +\infty$.
The following argument proves the same cancellation for every $G$ under
consideration, including the reciprocal-lattice condition that replaces
the elementary $SU(2)$ oscillation.

In the doubled kernel, the explicit $r_k$-dependent terms in the exponent cancel when the two Weyl elements agree. The surviving phase pairs electric
charge with magnetic flux, as expected for $BF$ lines.  When the Weyl
elements differ, the phase oscillates on the scale $r_k^{-1}$ in the input
flux.  Smearing with a smooth wavepacket suppresses these terms.  Because
the flux is sampled on a lattice, absence of a stationary point by itself
would be insufficient: a lattice frequency can be indistinguishable from
zero modulo the reciprocal lattice.  The next proposition excludes exactly
these resonances on a compact set of regular output fluxes.  Its estimate
concerns the smeared operator; an unsmeared matrix entry still contains the
oscillating terms. 

\begin{proposition}[Rapid decay of the off-diagonal Weyl terms]
\label{prop:off-diagonal-decay}
Let $f\in C_c^\infty(Y\times P)$, $K\Subset Y$, $E\subset P$ finite, and
$N\ge0$.
Let $\mathcal S_k^{\rm off}$ be the operator obtained from
Eq.~\eqref{eq:S_diag} by retaining the terms with $u\ne v$.  Then
\begin{equation}
 \sup_{(x,\eta)\in\Lambda_k\cap(K\times E)}
 |(\mathcal S_k^{\rm off}Q_kf)(x,\eta)|
 =O_{f,K,E,N}(r_k^{-N}).
 \label{eq:off-diagonal-decay}
\end{equation}
\end{proposition}

\begin{proof}
For $u\ne v$ and $x\in Y$, suppose for contradiction that
$q:=(u-v)x\in Q^\vee$.  Then $ux=vx+q$, and hence
\[
 (v^{-1}u,-v^{-1}q)\mathbin{\cdot}x
 =v^{-1}ux-v^{-1}q=x.
\]
This is a nonidentity element of the affine Weyl group
$W\ltimes Q^\vee$, because $u\ne v$.  An interior point of the fundamental
alcove has trivial affine-Weyl stabilizer, so the assumed resonance is
impossible.  Since $K\Subset Y$ and $W$ is finite, these
frequencies remain a positive distance from $Q^\vee$, uniformly for
$x\in K$.  This is the nonresonance needed for Poisson summation.

Fix $u\ne v$ and one of the finitely many input charges $\omega$ occurring
in $f$, and put
\begin{equation}
 \xi=(u-v)x,\qquad
 b_{u,v,\eta,\omega}(y)
 =e^{-\pi i((u+v)\eta,y)}f(y,\omega).
 \label{eq:off-diagonal-test-function}
\end{equation}
The flux support of $f$ has positive distance from $\partial Y$, so extension
by zero makes $b_{u,v,\eta,\omega}$ a smooth compactly supported function on
$\mathfrak t$.  For a point in the translated input mesh we may write
$y=(\mu+\rho+\omega/2)/r_k$; the two reconstructed chiral weights are then
$\mu+\omega$ and $\mu$.  Since the flux support stays a fixed positive
distance from the alcove boundary and $\omega$ ranges over a fixed finite
set, both weights satisfy all integrability inequalities for sufficiently
large $k$.  More explicitly, choose $D_f\Subset Y$ containing every flux
support and let $E_f$ be the finite charge support.  Put
\[
 \delta_f=\min_{y\in D_f}\min\!\left\{
 \min_i\langle y,\alpha_i^\vee\rangle,
 1-\langle y,\theta^\vee\rangle\right\}>0
\]
and
\[
 M_f=\max_{\omega\in E_f}\max\!\left\{
 \max_i|\langle\omega,\alpha_i^\vee\rangle|,
 |\langle\omega,\theta^\vee\rangle|\right\}.
\]
The inequality $r_k\delta_f\ge1+M_f/2$ is a single threshold ensuring that
both reconstructed weights are integrable for every charge in $E_f$ and
every Weyl pair.  Beyond this threshold we extend the finite sum to all
$\mu\in P$; the extra terms vanish under the zero extension of $f$.

Apart from the factor
$\epsilon(u)\epsilon(v)|\alpha_k^{\mathrm{KP}}|^2$, the contribution of the
fixed quadruple $(u,v,\eta,\omega)$ is
\begin{equation}
 \exp\!\left[-2\pi i(\xi,\rho+\omega/2)
 -\pi i((u+v)x,\omega)
 -\frac{\pi i}{2r_k}((u-v)\eta,\omega)\right]
 \Sigma_{r_k}(\xi),
 \label{eq:off-diagonal-exact-fixed-phase}
\end{equation}
where
\begin{equation}
 \Sigma_{r_k}(\xi)
 =\sum_{\mu\in P}e^{-2\pi i(\xi,\mu)}
 b_{u,v,\eta,\omega}
 \left(\frac{\mu+\rho+\omega/2}{r_k}\right).
 \label{eq:off-diagonal-lattice-sum}
\end{equation}
The exponential factor in
Eq.~\eqref{eq:off-diagonal-exact-fixed-phase} has modulus one.
With
$\widehat b(q)=\int_{\mathfrak t}b(y)e^{-2\pi i(q,y)}\,dy$, our convention is
\begin{equation}
 \sum_{\mu\in P}g(\mu)
 =\operatorname{covol}(P)^{-1}
 \sum_{p\in Q^\vee}\widehat g(p).
 \label{eq:Poisson-convention-weight-lattice}
\end{equation}
For completeness, set $a_\omega=\rho+\omega/2$ and apply the formula to
$g(t)=e^{-2\pi i(\xi,t)}b((t+a_\omega)/r_k)$.  The change of variables
$t=r_ky-a_\omega$ gives
\[
 \widehat g(p)
 =r_k^\ell e^{2\pi i(\xi+p,a_\omega)}
     \widehat b\bigl(r_k(\xi+p)\bigr).
\]
Thus dilation and translation, including their phase, give
\begin{equation}
 \Sigma_{r_k}(\xi)
 =\frac{r_k^\ell}{\operatorname{covol}(P)}
 \sum_{p\in Q^\vee}
 e^{2\pi i(\xi+p,\rho+\omega/2)}
 \widehat b_{u,v,\eta,\omega}\bigl(r_k(\xi+p)\bigr).
 \label{eq:off-diagonal-Poisson}
\end{equation}
This is the standard Poisson formula in the normalization of
\cite[Ch.~5, Sec.~3, and Ch.~6]{Stein:2003}.

To estimate the dual-lattice sum uniformly, define
\[
 B_K:=\{(u-v)x:x\in\overline K,\ u,v\in W,\ u\ne v\}
\]
and let $\delta_K=\operatorname{dist}(B_K,Q^\vee)>0$.  The positivity follows
from the nonresonance above, together with compactness of $K$ and finiteness
of $W$.

For every $L>\ell$,
\begin{equation}
 \sup_{\xi\in B_K}\sum_{p\in Q^\vee}|\xi+p|^{-L}<\infty.
 \label{eq:uniform-dual-lattice-sum}
\end{equation}
Indeed, $B_K$ is bounded.  The finitely many terms with $|p|$ in a fixed
bounded range are controlled by $\delta_K^{-L}$, while for all remaining
terms $|\xi+p|\ge |p|/2$ uniformly in $\xi\in B_K$; the lattice series
$\sum_{p\in Q^\vee\setminus\{0\}}|p|^{-L}$ converges precisely when
$L>\ell$.

The functions in Eq.~\eqref{eq:off-diagonal-test-function} range over finite
$u,v,\eta,\omega$ sets.  To see their uniform Fourier decay directly,
integrate by parts with the Euclidean Laplacian:
\[
(1+4\pi^2|q|^2)^m \widehat{b}(q)
=
\widehat{(1-\Delta)^m b}(q),
\qquad
\left|\widehat{(1-\Delta)^m b}(q)\right|
\leq
\left\|(1-\Delta)^m b\right\|_{L^1}.
\]
Their common compact support and finitely many charge labels bound the last
norm uniformly.  Choosing $m$ with $2m\geq L$ gives
$|\widehat b(q)|\leq C_L(1+|q|)^{-L}$ for every prescribed $L$.
No surface terms appear because the input wavepacket vanishes near the
walls.  This is the Fourier rapid-decay estimate used below
\cite[Chs.~5--6]{Stein:2003}.  Combining that estimate with
Eq.~\eqref{eq:uniform-dual-lattice-sum} gives, for $r_k\ge1$,
\begin{align}
 \sup_{\xi\in B_K}
 \sum_{p\in Q^\vee}
 \left|\widehat b_{u,v,\eta,\omega}
       \bigl(r_k(\xi+p)\bigr)\right|
 &\le C_L\sup_{\xi\in B_K}
       \sum_{p\in Q^\vee}(1+r_k|\xi+p|)^{-L}\notag\\
 &\le C'_{L,K}r_k^{-L}.
 \label{eq:uniform-Schwartz-lattice-bound}
\end{align}
Thus Eq.~\eqref{eq:off-diagonal-Poisson} is
$O(r_k^{\ell-L})$.  Multiplication by the exact Kac--Peterson factor
\[
 |\alpha_k^{\mathrm{KP}}|^2
 =|P/Q^\vee|^{-1}r_k^{-\ell}
\]
leaves $O(r_k^{-L})$.  Given the exponent $N$ in the statement, choose
$L>\max\{\ell,N\}$ and sum over the finite Weyl-pair and input-charge
sets.  This proves Eq.~\eqref{eq:off-diagonal-decay} with all constants
uniform on $K\times E$.
\end{proof}

Since $Q^\vee\subset P$ has index $|P/Q^\vee|$,
\[
 \frac{|\alpha_k^{\mathrm{KP}}|^2}{v_k}
 =\frac1{|P/Q^\vee|\operatorname{covol}(P)}
 =\gamma_G.
\]
The diagonal Weyl terms therefore give precisely the endpoint operator
$\mathcal S_0$ already defined in Eq.~\eqref{eq:BF-S-operator};
Eq.~\eqref{eq:BF-diagonal-kernel} is its integral kernel.

Consider a measurement that resolves the output flux only within a compact
regular window $K$ and a finite set of electric charges $E$.  The following
proposition says that the finite and limiting outgoing wavefunctions agree
in norm on this window.  Its volume factor $v_k$ converts the discrete
probability sum into the continuum integral.  The proof combines the
Riemann sum from the equal-Weyl terms with the oscillatory estimate just
established.  Control of all outcomes, including those outside the window,
will require unitarity in the next subsection.

\begin{proposition}[Local weighted convergence]
\label{prop:local-S-convergence}
For $f\in C_c^\infty(Y\times P)$, $K\Subset Y$, and finite $E\subset P$,
\begin{equation}
 v_k\sum_{z\in\Lambda_k\cap(K\times E)}
 \left|(\mathcal S_kQ_kf-Q_k\mathcal S_0f)(z)\right|^2
 \longrightarrow0 .
 \label{eq:local-S-convergence}
\end{equation}
\end{proposition}

\begin{proof}
The support of $f$ meets only finitely many input-charge fibers, say
$\omega\in E_f$, and its flux support has positive distance from
$\partial Y$.  For the equal-Weyl terms, Eq.~\eqref{eq:KP-normalization}
and $|\alpha_k^{\mathrm{KP}}|^2/v_k=\gamma_G$ give, exactly,
\[
 \gamma_G v_k\sum_{\omega\in E_f}
 \sum_{y:(y,\omega)\in\Lambda_k}
 K_{\rm diag}((x,\eta),(y,\omega))f(y,\omega).
\]
For each
$\omega\in E_f$, extend $f(\mathord\cdot,\omega)$ by zero from $Y$ to
$\mathfrak t$.  Because its support is compactly contained in $Y$, the
extension belongs to $C_c^\infty(\mathfrak t)$.  Choose a compact set
$D_f\Subset Y$ whose interior contains the supports of all these fibers, and
set, for $y\in\mathfrak t$,
\[
 h_{x,\eta,\omega}(y)
 =K_{\rm diag}((x,\eta),(y,\omega))f(y,\omega).
\]
The $y$-gradients of these functions are bounded uniformly for
$(x,\eta)\in K\times E$ and $\omega\in E_f$.  Let
$D_P=\sup_{q\in F_P}|q|$.  Comparing the value at the anchor of each
translated cell with its integral over that cell gives
\begin{equation}
 \left|v_k\sum_{y\in r_k^{-1}(P+s)}h_{x,\eta,\omega}(y)
       -\int_{\mathfrak t}h_{x,\eta,\omega}(y)\,dy\right|
 \le |(D_f)_{2D_P/r_k}|\,
      \frac{D_P}{r_k}\,
      \sup_{y\in\mathfrak t}|\nabla_yh_{x,\eta,\omega}(y)|,
 \label{eq:uniform-translated-Riemann-bound}
\end{equation}
where $(D_f)_r=\{y:\operatorname{dist}(y,D_f)<r\}$.  Only cells meeting
the support of $h_{x,\eta,\omega}$ contribute; their union lies in
$(D_f)_{2D_P/r_k}$.  All volumes and integrals use the fixed geometric
Lebesgue measure.  The estimate is independent of the translate $s$.
For the mesh with charge $\omega$ one
takes $s=\rho+\omega/2$; because
$D_f\Subset Y$ and $E_f$ is finite, the admitted mesh and the full translated
lattice agree on $D_f$ for all sufficiently large $k$.  Thus
Eq.~\eqref{eq:uniform-translated-Riemann-bound} proves the required uniform
convergence to Eq.~\eqref{eq:BF-S-operator}.  Proposition
\ref{prop:off-diagonal-decay} controls the remaining Weyl pairs uniformly.
The translated-lattice count
$v_k\#(\Lambda_k\cap(K\times E))=O_{K,E}(1)$ then converts this uniform
pointwise convergence into Eq.~\eqref{eq:local-S-convergence}.
\end{proof}

\subsection[No escape of mass]{No escape of mass\protect\footnote{%
Here ``mass'' refers to the squared $L^2$ norm of a wavefunction
(probability for a normalized state), not to physical mass.
For each fixed input wave packet, no nonzero amount of probability
can accumulate at the boundary $\partial Y$ by concentrating
in an increasingly thin boundary layer as $k\to\infty$.
Likewise, no nonzero amount of probability can escape
to arbitrarily large charges.}}
\label{subsec:no-escape}

To pass from Proposition~\ref{prop:local-S-convergence} to a global $L^2$
estimate, we must control samples near $\partial Y$ and in the large-charge
region.  We first check that $\mathcal S_0f$ belongs to the
charge-Schwartz class whenever $f\in C_c^\infty(Y\times P)$.
Only finitely many input charges $\omega$
occur.  Differentiating Eq.~\eqref{eq:BF-S-operator} in $x$ introduces powers
of $w^{-1}\omega$, and these are uniformly bounded on that finite set.  To
estimate the output charge, let $\Delta_y$ be the Euclidean Laplacian for
the geometric inner product.  Its Weyl invariance gives
\[
 (1-\Delta_y)^N e^{-2\pi i(w\eta,y)}
 =(1+4\pi^2|\eta|^2)^N e^{-2\pi i(w\eta,y)}.
\]
Integrating by parts moves $(1-\Delta_y)^N$ onto
$f(\mathord\cdot,\omega)$.  There are no boundary terms because every input
fiber is compactly supported in $Y$.  The resulting integral is bounded
uniformly in $x\in\overline Y$, so $(\mathcal S_0f)(x,\eta)=
 O_N((1+|\eta|)^{-2N})$ uniformly in $x$; the same argument applies after any number of
$x$-derivatives.  This is the usual Fourier-transform rapid-decay argument
\cite[Chs.~5--6]{Stein:2003}.  Hence
\begin{equation}
 f\in C_c^\infty(Y\times P)\quad\Longrightarrow\quad
 \mathcal S_0f\in\mathcal A(Y\times P).
 \label{eq:S0-Schwartz}
\end{equation}

Weighted sampling must preserve the norm of a physical wavepacket as the
flux mesh is refined.  There are two potential difficulties: integrability
removes a charge-dependent strip of fluxes near the walls, and the number of
allowed charges increases with $k$.  The lemma shows that, at fixed charge,
the removed strip has vanishing volume.  A lattice count uniform in the
charge then controls the sum of the rapidly decreasing charge tails.  This
is why the charge-Schwartz class is useful after applying $\mathcal S_0$,
which generally turns a state with finite charge support into one with
infinitely many charges.
\begin{lemma}[Global sampling and charge-tail control]
\label{lem:global-sampling}
For $\eta\in P$ write
\[
 L_{k,\eta}:=\{x\in Y:(x,\eta)\in\Lambda_k\}.
\]
There is a constant $C_{\mathrm{lat}}$, independent of $k$ and $\eta$, such that
\begin{equation}
 v_k\,\#L_{k,\eta}\le C_{\mathrm{lat}} .
 \label{eq:uniform-lattice-count}
\end{equation}
For every fixed $\eta\in P$ and every $\phi\in C(\overline Y)$,
\begin{equation}
 v_k\sum_{x\in L_{k,\eta}}\phi(x)
 \longrightarrow\int_Y\phi(x)\,dx.
 \label{eq:fixed-charge-riemann}
\end{equation}
Moreover, if $g\in\mathcal A(Y\times P)$ and
$E_M:=\{\eta\in P:|\eta|\le M\}$, then
\begin{equation}
 \sup_k v_k\sum_{\eta\notin E_M}
        \sum_{x\in L_{k,\eta}}|g(x,\eta)|^2
 \le C_{\mathrm{lat}}\sum_{\eta\notin E_M}
        \sup_{x\in Y}|g(x,\eta)|^2
 \xrightarrow[M\to\infty]{}0.
 \label{eq:uniform-charge-tail}
\end{equation}
Consequently
\begin{equation}
 \|J_kg\|_{\mathcal H_k}^2
 =v_k\sum_{z\in\Lambda_k}|g(z)|^2
 \longrightarrow \|g\|_{\mathcal H_0}^2.
 \label{eq:global-sampling}
\end{equation}
The analogous Riemann-sum convergence holds on $K\times E$ for finite
$E\subset P$ whenever $K\Subset Y$ has Lebesgue-null boundary.
\end{lemma}

\begin{proof}

Solving Eq.~\eqref{eq:regular-lattice} for $a,b$ shows that, at fixed charge,
the full translated lattice and its admitted subset are, respectively,

\begin{equation}
 \widetilde L_{k,\eta}=r_k^{-1}(P+\rho+\eta/2),
 \qquad
 L_{k,\eta}=\widetilde L_{k,\eta}\cap Y_{k,\eta},
 \label{eq:fixed-charge-lattice}
\end{equation}
where the simultaneous dominance and level conditions on $a$ and $b$ are
equivalent to

\begin{align}
 r_k\langle x,\alpha_j^\vee\rangle
 &\ge 1+\tfrac12|\langle\eta,\alpha_j^\vee\rangle|,
 &&1\le j\le \ell,\notag\\
 r_k(1-\langle x,\theta^\vee\rangle)
 &\ge 1+\tfrac12|\langle\eta,\theta^\vee\rangle| .
 \label{eq:charge-dependent-admissibility}
\end{align}
Here we used $\langle\rho,\alpha_j^\vee\rangle=1$ and
$\langle\rho,\theta^\vee\rangle=h^\vee-1$ in the
long-root-length-two convention.
Since the alcove has finitely many affine walls, there is a group-dependent
constant $C$ such that
\begin{equation}
 Y\setminus Y_{k,\eta}\subset
 \left\{x\in Y:\operatorname{dist}(x,\partial Y)
 \le C\frac{1+|\eta|}{r_k}\right\}.
 \label{eq:charge-boundary-strip}
\end{equation}
In particular, for fixed $\eta$ the deleted points lie within an
$O_\eta(r_k^{-1})$ neighborhood of $\partial Y$.

Using the fundamental parallelotope $F_P$ fixed before Theorem~\ref{thm:strong-modular}, set
\[
  D:=\sup_{q\in F_P}|q|.
\]
The cells $x+r_k^{-1}F_P$ based at the full translated lattice in
Eq.~\eqref{eq:fixed-charge-lattice} are disjoint up to null boundaries and
have volume $v_k$.  Those based at points of $L_{k,\eta}$ lie in
$Y+r_k^{-1}F_P\subset\{y:\operatorname{dist}(y,Y)\leq D\}$, a bounded set independent of $k$ and $\eta$.  Hence their total volume is bounded
uniformly in $k$ and in the translate, proving eq.~\eqref{eq:uniform-lattice-count}.

For fixed $\eta$, the width in Eq.~\eqref{eq:charge-boundary-strip} tends to
zero.  The translated-lattice Riemann sum over $Y$ converges uniformly in its
translate.  One can prove this without a choice of nested meshes.  On a
full cell the difference between the sampled value and $\phi$ is bounded
by the modulus of continuity
\[
 \omega_\phi(D/r_k)
 :=\sup_{\substack{x,y\in\overline Y\\|x-y|\leq D/r_k}}
       |\phi(x)-\phi(y)|,
\]
which tends to zero.  Cells meeting the boundary occupy a two-sided strip
of thickness $O(r_k^{-1})$.  Its volume is $O(r_k^{-1})$ because the alcove
has finitely many flat faces.  Their contribution is bounded by
$O(r_k^{-1})\|\phi\|_\infty$.  Summing the full-cell and boundary-cell
bounds proves the asserted ordinary Riemann-sum convergence.
Every deleted anchor lies in the strip of
Eq.~\eqref{eq:charge-boundary-strip}.  Its full cell is therefore contained
in the two-sided $\bigl(C(1+|\eta|)+D\bigr)/r_k$ neighborhood of
$\partial Y$.  These full cells are disjoint up to null boundaries and have
volume $v_k$.  The finite-facet slab estimate for the enlarged strip gives
\[
 v_k\#\bigl(\widetilde L_{k,\eta}
       \cap(Y\setminus Y_{k,\eta})\bigr)
 =O_\eta(r_k^{-1}).
\]
Deleting these points therefore changes the Riemann sum by $o(1)$.  This proves
Eq.~\eqref{eq:fixed-charge-riemann}.  Equation
\eqref{eq:uniform-lattice-count} then gives the first inequality in
Eq.~\eqref{eq:uniform-charge-tail}.  The charge-Schwartz condition makes its
right-hand side summable on the lattice $P$: explicitly, choose
$N>\ell/2$ in Eq.~\eqref{eq:charge-Schwartz}.  Then, for a suitable $C_N$,
\[
 \sup_{x\in Y}|g(x,\eta)|^2
 \le C_N(1+|\eta|)^{-2N},
\]
and the right-hand side is summable because $P$ has rank $\ell$.  Its tail
therefore tends to zero.
The continuous charge tail satisfies the analogous estimate
\[
 \sum_{\eta\notin E_M}\int_Y|g(x,\eta)|^2\,dx
 \le |Y|\sum_{\eta\notin E_M}
       \sup_{x\in\overline Y}|g(x,\eta)|^2
 \xrightarrow[M\to\infty]{}0.
\]
Apply Eq.~\eqref{eq:fixed-charge-riemann} to
$\phi(x)=|g(x,\eta)|^2$ for $\eta\in E_M$, and then let $M\to\infty$ using
Eq.~\eqref{eq:uniform-charge-tail} and the preceding continuous-tail
estimate.  This proves
Eq.~\eqref{eq:global-sampling}.  The proof on $K\times E$ is the same
finite-charge argument; the null-boundary hypothesis ensures convergence of
the Riemann sums for $\mathbf 1_K|g(\mathord\cdot,\eta)|^2$.
\end{proof}

Let
\[
 r_Y=\max_{x\in\overline Y}\operatorname{dist}(x,\partial Y)>0.
\]
For $0<\delta<r_Y$ we use the compact polytope
\begin{equation}
 K_\delta=\{x\in Y:\operatorname{dist}(x,\partial Y)\ge\delta\}.
 \label{eq:interior-alcove-window}
\end{equation}
It is nonempty and its boundary has Lebesgue measure zero.  Moreover
$K_\delta\uparrow Y$ as $\delta\downarrow0$.
For $M>0$ put $E_M=\{\eta\in P:|\eta|\le M\}$, and let
$\Pi_{k,\delta,M}$ be the coordinate projection in $\mathcal H_k$ onto
$\Lambda_k\cap(K_\delta\times E_M)$.

Local convergence alone would allow a finite amount of probability to
accumulate near a singular conjugacy class or at charges tending to
infinity.  Exact finite-level unitarity removes this possibility for the
fixed wavepackets considered here.  The limiting Fourier transform has the
same total norm, and a sufficiently large compact window already captures
almost all of that norm.  Once the finite-level amplitudes agree inside the
window, probability conservation leaves only a small norm outside it.
This argument does not assign states to the singular strata; it proves that
they do not absorb a finite norm from a fixed regular wavepacket in this
limit.
\begin{theorem}[No escape of mass]
\label{thm:no-escape}
For every $f\in C_c^\infty(Y\times P)$,
\begin{equation}
 \|\mathcal S_kJ_kf-J_k\mathcal S_0f\|_{\mathcal H_k}
 \longrightarrow0 .
 \label{eq:global-sampled-S}
\end{equation}
More explicitly, for every $\varepsilon>0$ there are $\delta>0$, $M<\infty$,
and $k_0$ such that, for $k\ge k_0$,
\begin{equation}
 \|(1-\Pi_{k,\delta,M})\mathcal S_kJ_kf\|
 +\|(1-\Pi_{k,\delta,M})J_k\mathcal S_0f\|<\varepsilon.
 \label{eq:explicit-no-escape}
\end{equation}
Here $f$ is fixed before the window $(\delta,M)$ is chosen.
\end{theorem}

\begin{proof}

Set $a_k=\mathcal S_kJ_kf$ and $b_k=J_k\mathcal S_0f$.  Exact unitarity of
the finite doubled matrix, \eqref{eq:global-sampling}, and
\eqref{eq:S0-Schwartz}, together with norm preservation by $\mathcal S_0$
from Proposition~\ref{prop:regular-Fourier-unitarity}, give
\[
 \|a_k\|\to\|f\|,\qquad
 \|b_k\|\to\|\mathcal S_0f\|=\|f\|.
\]
Given $\varepsilon>0$, choose $\delta>0$ and $M<\infty$ so that
the $\mathcal H_0$-norm of $\mathcal S_0f$ outside
$\Omega_{\delta,M}:=K_\delta\times E_M$ is $<\varepsilon/4$.  Put
$\Pi_k=\Pi_{k,\delta,M}$.
Proposition~\ref{prop:local-S-convergence}
gives $\|\Pi_k(a_k-b_k)\|\to0$, while the localized part of
Lemma~\ref{lem:global-sampling} gives
$\|\Pi_kb_k\|^2\to\|1_{\Omega_{\delta,M}}\mathcal S_0f\|^2$.  Therefore
$\|\Pi_ka_k\|\to\|1_{\Omega_{\delta,M}}\mathcal S_0f\|$ as well.  Since
$\|a_k\|,\|b_k\|\to\|\mathcal S_0f\|$, Pythagoras gives the explicit limits
\begin{align}
 \|(1-\Pi_k)a_k\|^2
 &=\|a_k\|^2-\|\Pi_ka_k\|^2
   \longrightarrow
   \|1_{\Omega_{\delta,M}^c}\mathcal S_0f\|^2,\notag\\
 \|(1-\Pi_k)b_k\|^2
 &=\|b_k\|^2-\|\Pi_kb_k\|^2
   \longrightarrow
   \|1_{\Omega_{\delta,M}^c}\mathcal S_0f\|^2.
 \label{eq:no-escape-pythagoras}
\end{align}
The limiting norm is $<\varepsilon/4$.  For all sufficiently
large $k$ their sum is therefore $<\varepsilon$, which proves
\eqref{eq:explicit-no-escape}.  Together with
$\|\Pi_k(a_k-b_k)\|\to0$, the decomposition into $\Pi_k$ and
$1-\Pi_k$ proves \eqref{eq:global-sampled-S}.
\end{proof}

\subsection{A common Hilbert space and the modular group}
\label{subsec:common-hilbert-modular}

The cell isometry in Eq.~\eqref{eq:lattice-cell-isometry} replaces a finite
basis vector by a normalized wavefunction localized in its flux cell.
Its adjoint averages a continuum wavefunction over that cell.  Sampling at
the cell anchor and averaging should agree when the wavefunction varies
slowly on the scale $r_k^{-1}$.  The next lemma makes this statement precise,
including the cells intersected by the walls.  It also proves that every
fixed normalizable $BF$ state can be approximated by these finite-level
states: $\mathsf P_k\to1$ strongly means
$\|\mathsf P_k\psi-\psi\|\to0$ for every $\psi\in\mathcal H_0$.

\begin{lemma}[Global cell approximation]
\label{lem:global-cell-approximation}
With $I_k$ defined by Eq.~\eqref{eq:lattice-cell-isometry},
\begin{equation}
 \mathsf P_k\xrightarrow{\rm SOT}1,\qquad
 \|I_k^*g-J_kg\|\to0,\qquad
 \|I_kJ_kg-g\|\to0
 \quad(g\in\mathcal A(Y\times P)).
 \label{eq:cell-approximation}
\end{equation}
\end{lemma}

\begin{proof}
We first prove the two geometric estimates used below.  Since $Y$ is a
bounded polytope with finitely many facets, there is a constant $C_Y$ such
that, for every $s>0$,
\begin{equation}
 \left|\{y\in Y:\operatorname{dist}(y,\partial Y)<s\}\right|
 \le C_Y\min\{1,s\}.
 \label{eq:polytope-boundary-strip-volume}
\end{equation}
Indeed, after increasing the constant for $s\ge1$, the set on the left is
contained in the union of finitely many slabs of thickness $s$ adjacent to
the facets of $Y$.  Orthogonal projection to a facet and Fubini's theorem
bound the volume of each slab by $s$ times a fixed bounded
$(\ell-1)$-dimensional cross-sectional volume.

Let $D=\sup\{|q|:q\in F_P\}$.  For fixed $\eta$, the cells of the full
translated lattice $\widetilde L_{k,\eta}$
tile $\mathfrak t$ up to null boundaries.  If $y\in Y$ does not belong to a
cell anchored at a point of $L_{k,\eta}$, let $q\in\widetilde L_{k,\eta}$ be
the anchor of the full cell containing $y$.  Then $|y-q|\le D/r_k$.  If
$q\notin Y$, the point $y$ lies within $D/r_k$ of $\partial Y$.  If
$q\in Y\setminus Y_{k,\eta}$, Eq.~\eqref{eq:charge-boundary-strip} places
$q$, and hence $y$, in a boundary strip of width
$C'(1+|\eta|)/r_k$.  Applying
Eq.~\eqref{eq:polytope-boundary-strip-volume} proves
\begin{equation}
 \left|Y\setminus\bigcup_{x\in L_{k,\eta}}
   \bigl(Y\cap(x+r_k^{-1}F_P)\bigr)\right|
 \le C\min\!\left\{1,\frac{1+|\eta|}{r_k}\right\}.
 \label{eq:uncovered-cell-strip}
\end{equation}

For the second estimate, suppose that an admitted cell anchored at $x$ is
truncated, so that $|C_{k,(x,\eta)}|<v_k$.  Its full cell meets the complement
of $Y$.  Since $x\in Y$ and that cell has diameter at most $2D/r_k$, both its
anchor and every point of the cell lie in the $2D/r_k$-neighborhood of
$\partial Y$.  The full cells are pairwise disjoint up to null sets and each
has volume $v_k$.  Their union has volume $O(r_k^{-1})$ by the same
facet-slab argument.  Comparing volumes gives
\begin{equation}
 v_k\#\{x\in L_{k,\eta}:|C_{k,(x,\eta)}|<v_k\}
 \le \frac{C}{r_k},
 \label{eq:truncated-cell-count}
\end{equation}
with $C$ independent of the translated lattice and hence of $\eta$.

Now fix $\eta\in P$ and identify $h\in C^1(\overline Y)$ with
$h\delta_\eta\in\mathcal H_0$.
On a full cell in the $\eta$-summand, $I_k^*h$ is $\sqrt{v_k}$ times the cell average, whereas $J_kh$ is
$\sqrt{v_k}h(x)$.  If $C=x+r_k^{-1}F_P$ is full, Jensen's inequality and
$|y-x|\le D/r_k$ give
\[
 \left|v_k^{-1/2}\int_C h(y)\,dy-\sqrt{v_k}h(x)\right|^2
 \le \int_C|h(y)-h(x)|^2\,dy
 \le \frac{D^2}{r_k^2}\|\nabla h\|_\infty^2|C|.
\]
Summing this inequality and using
Eq.~\eqref{eq:uniform-lattice-count} gives a total full-cell error
$O(r_k^{-2})\|\nabla h\|_\infty^2$.  On truncated cells, Cauchy--Schwarz
gives
\[
 \left||C|^{-1/2}\int_C h-\sqrt{v_k}h(x)\right|^2
 \le 4v_k\|h\|_\infty^2.
\]
After summing over the truncated cells,
Eq.~\eqref{eq:truncated-cell-count} gives
\[
 \sum_{\substack{C\ {\rm truncated}}}
 \left||C|^{-1/2}\int_C h-\sqrt{v_k}h(x_C)\right|^2
 \le \frac{C}{r_k}\|h\|_\infty^2.
\]
Here $x_C$ denotes the anchor of $C$.  Together with the full-cell estimate,
this proves
\[
 \|I_k^*h-J_kh\|_{\ell^2(L_{k,\eta})}\longrightarrow0.
\]
The same decomposition, now also using
Eq.~\eqref{eq:uncovered-cell-strip}, gives
\begin{equation}
 \|I_kJ_kh-h\|_{L^2(Y)}\longrightarrow0.
 \label{eq:fixed-charge-cell-sampling}
\end{equation}
For the truncated-cell part of the second statement, note that the
squared norm of $I_kJ_kh$ on such a cell is exactly
$v_k|h(x)|^2$, independently of its volume.  Its total contribution is again
$O(r_k^{-1})\|h\|_\infty^2$, while the $L^2$ mass of $h$ on the union of
the truncated and uncovered regions tends to zero by the two strip estimates.

For $g\in\mathcal A(Y\times P)$, let $g_M=1_{E_M}g$.  The preceding
fixed-charge estimates, summed over the finite set $E_M$, apply to $g_M$.
For the remaining tail $g^{>M}=g-g_M$, contractivity of $I_k^*$ and isometry
of $I_k$ give
\begin{align*}
 \|I_k^*g^{>M}-J_kg^{>M}\|
 &\leq \|g^{>M}\|+\|J_kg^{>M}\|,\\
 \|I_kJ_kg^{>M}-g^{>M}\|
 &\leq \|J_kg^{>M}\|+\|g^{>M}\|.
\end{align*}
Here $\|g^{>M}\|\to0$ as $M\to\infty$, while
Eq.~\eqref{eq:uniform-charge-tail} gives
$\sup_k\|J_kg^{>M}\|\to0$.  This proves the last two limits in
Eq.~\eqref{eq:cell-approximation}.  On the same dense class,
\[
 \|\mathsf P_kg-g\|
 \leq \|I_k^*g-J_kg\|+\|I_kJ_kg-g\|\longrightarrow0.
\]
To extend this statement beyond smooth functions, fix
$\psi\in\mathcal H_0$ and choose $g\in\mathcal A(Y\times P)$ with
$\|\psi-g\|<\varepsilon$.  Such approximants exist by first truncating the
charge sum, then approximating each of the finitely many $L^2(Y)$ functions
by a compactly supported smooth function.  Since every $\mathsf P_k$ is a
contraction,
\[
 \|\mathsf P_k\psi-\psi\|
 \leq 2\varepsilon+\|\mathsf P_kg-g\|.
\]
Taking $k\to\infty$ and then $\varepsilon\downarrow0$ proves the strong
limit on all of $\mathcal H_0$.
\end{proof}

\begin{proof}[Proof of Theorem~\ref{thm:strong-modular}]
The endpoint representation is supplied by
Proposition~\ref{prop:regular-Fourier-unitarity}, proved independently in
Appendix~\ref{app:endpoint-Fourier-unitarity}.
For the $S$-generator, put
\begin{equation}
 \mathcal S_k^\#=I_k\mathcal S_kI_k^*+(1-\mathsf P_k).
 \label{eq:embedded-S}
\end{equation}
Thus $\mathcal S_k^\#=\rho_k^\#(S)$.
For $f\in C_c^\infty(Y\times P)$, insert $J_kf$ and
$J_k\mathcal S_0f$.  The function $\mathcal S_0f$ belongs to
$\mathcal A(Y\times P)$ by Eq.~\eqref{eq:S0-Schwartz}, so the cell lemma
applies to both functions:
\[
 \|\mathcal S_k^\#f-\mathcal S_0f\|
 \le\|I_k^*f-J_kf\|
 +\|\mathcal S_kJ_kf-J_k\mathcal S_0f\|
 +\|I_kJ_k\mathcal S_0f-\mathcal S_0f\|
 +\|(1-\mathsf P_k)f\|.
\]
Equations~\eqref{eq:cell-approximation} and
\eqref{eq:global-sampled-S} make all four terms vanish, and density then gives
\begin{equation}
 \mathcal S_k^\#\xrightarrow{\rm SOT}\mathcal S_0.
 \label{eq:strong-S}
\end{equation}
Here the extension from the test space uses only uniform boundedness:
for an arbitrary $\psi$ and a test function $f$,
\[
 \|(\mathcal S_k^\#-\mathcal S_0)\psi\|
 \leq 2\|\psi-f\|+\|(\mathcal S_k^\#-\mathcal S_0)f\|.
\]
First choose $f$ close to $\psi$, then take $k$ large.  The factor two comes
from the unit norm of each operator and is independent of $k$.
The adjoints converge strongly as well.  Indeed, if $U_k$ and $U$ are
unitary and $U_k\to U$ strongly, then, for every $\psi$,
\begin{equation}
 \|(U_k^*-U^*)\psi\|
 =\|\psi-U_kU^*\psi\|\longrightarrow0.
 \label{eq:unitary-strong-adjoint-identity}
\end{equation}
Applying this identity to $\mathcal S_k^\#$ and $\mathcal S_0$ gives
strong-$*$ convergence.

For the twist, the WZW formula
\cite[Eq.~(3.3b)]{Gannon:1995uv} and cancellation of the two
vacuum phases give an exact identity.  Substitution of
Eq.~\eqref{eq:classical-regular-lattice} gives
$ (a,a+2\rho)-(b,b+2\rho)=2r_k(\eta,x)$, and hence the doubled twist
$\theta_{(a,b)}:=\theta_a\theta_b^{-1}$ of eq.~\eqref{eq:finite-doubled-twist} is
\begin{align}
 \theta_{(a,b)}
 &=\exp\!\left[
 \frac{2\pi i}{2r_k}
 \bigl((a,a+2\rho)-(b,b+2\rho)\bigr)\right]\notag\\
 &=e^{2\pi i(\eta,x)}=\chi_\eta(g_x).
 \label{eq:twist-flux-charge-pairing}
\end{align}
Thus the endpoint twist is the operator already defined in
Eq.~\eqref{eq:BF-T-operator}, and the exact sampling identity is
$\mathcal T_kJ_k=J_k\mathcal T_0$.
Both $\mathcal T_0$ and $\mathcal T_0^*$ preserve
$C_c^\infty(Y\times P)$.  Put
$\mathcal T_k^\#=I_k\mathcal T_kI_k^*+(1-\mathsf P_k)$; then
$\mathcal T_k^\#=\rho_k^\#(T)$.  For
$f\in C_c^\infty(Y\times P)$, exact intertwining on sampled vectors gives
\begin{align}
 \|\mathcal T_k^\#f-\mathcal T_0f\|
 \le{}&\|I_k^*f-J_kf\|
       +\|I_kJ_k\mathcal T_0f-\mathcal T_0f\|
       +\|(1-\mathsf P_k)f\|.
 \label{eq:strong-twist-core-estimate}
\end{align}
All three terms tend to zero by
Lemma~\ref{lem:global-cell-approximation}.  The same calculation with
$\mathcal T_k^*$ and $\mathcal T_0^*$ is valid because
$\mathcal T_k^*J_k=J_k\mathcal T_0^*$.  Since the sampling core is dense and
all four operators are unitary, these two estimates extend by density and
prove $\mathcal T_k^\#\to\mathcal T_0$ in the strong-$*$ topology.

The generators $\mathcal S_k^\#$ and $\mathcal T_k^\#$, together with their
inverses, converge strongly to the corresponding operators for $\rho_0$.
Fix $\gamma\in SL(2,\mathbb Z)$ and choose a word in
$S^{\pm1},T^{\pm1}$ representing it.  For two factors the elementary bound
\[
 \|(A_kB_k-AB)\psi\|
 \leq\|A_k\|\,\|(B_k-B)\psi\|
     +\|(A_k-A)B\psi\|
\]
shows why uniformly bounded strongly convergent operators can be
multiplied.  Induction applies this estimate to the chosen finite word,
giving $\rho_k^\#(\gamma)\to\rho_0(\gamma)$.  The word is held fixed as
$k\to\infty$; the argument asserts no uniformity for word lengths growing
with the level.  Applying the same argument to
$\gamma^{-1}$ gives convergence of the adjoints.
This proves the pointwise strong-$*$ convergence in
Eq.~\eqref{eq:strong-modular}.  The assertion for admissible embeddings
follows from Lemma~\ref{lem:admissible-embedding-robustness}, proved below.
\end{proof}

The norm limit also controls any bounded measurement on the torus Hilbert
space.  For a bounded operator $B$ and $\psi\in\mathcal H_0$, set
$\psi_k=\rho_k^\#(\gamma)\psi$ and $\psi_0=\rho_0(\gamma)\psi$.  Then
\begin{equation}
 \bigl|\langle\psi_k,B\psi_k\rangle
       -\langle\psi_0,B\psi_0\rangle\bigr|
 \leq 2\|B\|\,\|\psi\|\,\|\psi_k-\psi_0\|\longrightarrow0.
 \label{eq:bounded-observable-convergence}
\end{equation}
Taking $B$ to be the projection onto a fixed set of fluxes and charges gives
convergence of its measurement probability.  Traces of modular operators,
unsmeared line insertions, and unbounded observables require additional
estimates; they do not follow from strong convergence alone.  In particular,
the finite vacuum $(a,b)=(0,0)$ has $x=\rho/r_k$ and approaches an alcove
wall.  Its embedded unit vector is not a fixed normalizable state of the
limiting space.  The theorem therefore does not by itself establish the
large-level limit of vacuum-to-vacuum surgery amplitudes or closed-manifold
partition functions.
\subsection{Independence of the sampling embeddings}
\label{subsec:embedding-independence}

The preceding proof uses two approximation properties of $I_k$: its
adjoint agrees asymptotically with weighted sampling, and sampled smooth
functions reconstruct in $L^2$.  We now show that these properties suffice
to determine the operator limit.

A different placement or shape of the small flux cells describes the same
continuum state if its discrete coefficients agree with weighted sampling
and its reconstructed wavefunction converges in norm.  We use these two
physical requirements as the definition of an admissible comparison.
\begin{definition}[Admissible sampling embeddings]
A family of isometries
$\widetilde I_k:\mathcal H_k\to\mathcal H_0$ is called admissible if, for
every $g\in\mathcal A(Y\times P)$,
\begin{equation}
 \|\widetilde I_k^*g-J_kg\|\longrightarrow0,
 \qquad
 \|\widetilde I_kJ_kg-g\|\longrightarrow0.
 \label{eq:admissible-embedding-hypotheses}
\end{equation}
\end{definition}

The statement below compares two continuum descriptions of the same
finite-level vector.  The transition operator between them approaches the
identity on every fixed state.  It follows that a limiting modular
amplitude cannot depend on how an admissible family of shrinking cells was
chosen.  This is a statement about the normalization and localization of
wavepackets; arbitrary level-dependent changes of basis need not satisfy
its hypotheses.
\begin{lemma}[Robustness under admissible sampling embeddings]
\label{lem:admissible-embedding-robustness}
For an admissible family $(\widetilde I_k)$, put
$\widetilde{\mathsf P}_k=\widetilde I_k\widetilde I_k^*$.
Then
\begin{equation}
 W_k:=\widetilde I_kI_k^*\longrightarrow1,
 \qquad W_k^*=I_k\widetilde I_k^*\longrightarrow1
 \quad\text{strongly on }\mathcal H_0.
 \label{eq:admissible-embedding-transition}
\end{equation}
Consequently, if $U_k\in U(\mathcal H_k)$ and
\begin{equation}
 I_kU_kI_k^*+(1-\mathsf P_k)\longrightarrow U_\infty
 \quad\text{strongly-$*$},
 \label{eq:first-embedding-operator-limit}
\end{equation}
then the same limit is obtained with the alternative embeddings:
\begin{equation}
 \widetilde I_kU_k\widetilde I_k^*+
 (1-\widetilde{\mathsf P}_k)\longrightarrow U_\infty
 \quad\text{strongly-$*$}.
 \label{eq:alternative-embedding-operator-limit}
\end{equation}
\end{lemma}

\begin{proof}
Fix $g\in\mathcal A(Y\times P)$.  Insert the common sampled vector $J_kg$:
\begin{align}
 \|W_kg-g\|
 &\leq
 \|I_k^*g-J_kg\|+\|\widetilde I_kJ_kg-g\|,\notag\\
 \|W_k^*g-g\|
 &\leq
 \|\widetilde I_k^*g-J_kg\|+\|I_kJ_kg-g\|.
 \label{eq:embedding-transition-dense-estimate}
\end{align}
Both right-hand sides tend to zero by
Lemma~\ref{lem:global-cell-approximation} and
\eqref{eq:admissible-embedding-hypotheses}.  Since $W_k$ and $W_k^*$ are
contractions, density extends the two limits to all of $\mathcal H_0$.
The same two assumptions also give
\[
 \|\widetilde{\mathsf P}_kg-g\|
 \leq \|\widetilde I_k^*g-J_kg\|
      +\|\widetilde I_kJ_kg-g\|\longrightarrow0
\]
on $\mathcal A(Y\times P)$, and hence
$\widetilde{\mathsf P}_k\to1$ strongly by density.

Let $U_k^\#=I_kU_kI_k^*+(1-\mathsf P_k)$.  The identity
\begin{equation}
 \widetilde I_kU_k\widetilde I_k^*
 =W_kU_k^\#W_k^*
 \label{eq:embedding-conjugacy-on-ranges}
\end{equation}
holds because $W_k(1-\mathsf P_k)=0$ and
$(1-\mathsf P_k)W_k^*=0$.  Products of strongly convergent uniformly bounded
operators converge strongly.  Equations
\eqref{eq:admissible-embedding-transition} and
\eqref{eq:first-embedding-operator-limit} therefore show that the right-hand
side of~\eqref{eq:embedding-conjugacy-on-ranges} converges strongly to $U_\infty$.
The complementary projection tends strongly to zero.  Repeating the same
argument with $U_k^*$ proves the adjoint convergence.
\end{proof}

The abstract sampling requirements can be checked by measuring the cells.
Their diameters must shrink, their anchors must approach them, and the
normalization factors must approach $\sqrt{v_k}$ in the summed square
error.  Finally, the uncovered volume in each fixed charge fiber must tend
to zero.  These conditions permit irregular cell shapes and cells cut by
the alcove boundary.  They ensure that neither the position nor the
probability assigned to a smooth wavepacket changes in the continuum
limit.
\begin{corollary}[Cell criterion for admissible embeddings]
\label{cor:cell-criterion-admissible-embeddings}
For each $z=(x_z,\eta_z)\in\Lambda_k$, let
$\widetilde C_{k,z}\subset Y$ be measurable with positive measure, with cells
in the same charge fiber pairwise disjoint modulo null sets, and define
\[
 \widetilde I_k\delta_z
 =|\widetilde C_{k,z}|^{-1/2}
 1_{\widetilde C_{k,z}}\delta_{\eta_z}.
\]
Assume that for every finite $E\subset P$,
\begin{align}
 \sup_{\substack{z\in\Lambda_k\\\eta_z\in E}}
 \bigl(\operatorname{diam}\widetilde C_{k,z}
 +\operatorname{dist}(x_z,\widetilde C_{k,z})\bigr)&\longrightarrow0,
 \label{eq:admissible-cell-diameter}\\
 \sum_{\substack{z\in\Lambda_k\\\eta_z\in E}}
 \bigl(|\widetilde C_{k,z}|^{1/2}-v_k^{1/2}\bigr)^2&\longrightarrow0,
 \label{eq:admissible-cell-volume}\\
 \left|Y\setminus
 \bigcup_{\substack{z\in\Lambda_k\\\eta_z=\eta}}
 \widetilde C_{k,z}\right|&\longrightarrow0
 \quad(\eta\in E).
 \label{eq:admissible-cell-cover}
\end{align}
Then $(\widetilde I_k)$ is admissible, and hence gives the same strong-$*$
limit as $I_k$.
\end{corollary}

\begin{proof}
Fix a finite set $E\subset P$ and write
\[
 \delta_{k,E}:=
 \sup_{\substack{z\in\Lambda_k\\ \eta_z\in E}}
 \bigl(\operatorname{diam}\widetilde C_{k,z}
 +\operatorname{dist}(x_z,\widetilde C_{k,z})\bigr).
\]
For $g\in\mathcal A(Y\times P)$, let
\[
 \omega_{g,E}(r):=
 \max_{\eta\in E}
 \sup_{\substack{x,y\in\overline Y\\ |x-y|\leq r}}
 |g(x,\eta)-g(y,\eta)|.
\]
Then $\delta_{k,E}\to0$ by
Eq.~\eqref{eq:admissible-cell-diameter}, while
$\omega_{g,E}(r)\to0$ as $r\downarrow0$.

Put $\widetilde C_z=\widetilde C_{k,z}$.  Directly from the definitions,
\begin{align*}
 (\widetilde I_k^*g)_z-(J_kg)_z
 &=\frac{1}{|\widetilde C_z|^{1/2}}
   \int_{\widetilde C_z}
   \bigl(g(y,\eta_z)-g(x_z,\eta_z)\bigr)\,dy\\
 &\quad+\bigl(|\widetilde C_z|^{1/2}-v_k^{1/2}\bigr)
   g(x_z,\eta_z).
\end{align*}
The cells in each charge fiber are disjoint.  Hence, with
$\|g\|_{\infty,E}:=\max_{\eta\in E}\sup_{\overline Y}|g(\mathord\cdot,\eta)|$,
\begin{align}
 \|\widetilde I_k^*(1_Eg)-J_k(1_Eg)\|^2
 &\leq 2|E|\,|Y|\,\omega_{g,E}(\delta_{k,E})^2\notag\\
 &\quad+2\|g\|_{\infty,E}^2
 \sum_{\substack{z\in\Lambda_k\\ \eta_z\in E}}
 \bigl(|\widetilde C_z|^{1/2}-v_k^{1/2}\bigr)^2
 \longrightarrow0.
 \label{eq:admissible-cell-first-sampling-estimate}
\end{align}

On $\widetilde C_z$ the second sampling operator is
\[
 (\widetilde I_kJ_kg)(y,\eta_z)
 =\left(\frac{v_k}{|\widetilde C_z|}\right)^{1/2}
  g(x_z,\eta_z).
\]
Integrating the squared modulus of the difference from $g(y,\eta_z)$ over the cells gives
\begin{align}
 &\sum_{\substack{z\in\Lambda_k\\ \eta_z\in E}}
 \int_{\widetilde C_z}
 \bigl|(\widetilde I_kJ_kg)(y,\eta_z)-g(y,\eta_z)\bigr|^2dy
 \notag\\
 &\qquad\leq 2|E|\,|Y|\,\omega_{g,E}(\delta_{k,E})^2
 +2\|g\|_{\infty,E}^2
 \sum_{\substack{z\in\Lambda_k\\ \eta_z\in E}}
 \bigl(v_k^{1/2}-|\widetilde C_z|^{1/2}\bigr)^2.
 \label{eq:admissible-cell-second-sampling-estimate}
\end{align}
The right-hand side tends to zero.  On the uncovered part of a charge fiber,
$\widetilde I_kJ_kg$ vanishes, and
Eq.~\eqref{eq:admissible-cell-cover} shows that the remaining integral of
$|g|^2$ also tends to zero.  This proves both sampling limits for $1_Eg$.

Finally, take $E_M=\{\eta\in P:|\eta|\leq M\}$ and put
$g^{>M}=(1-1_{E_M})g$.  Since $\widetilde I_k$ is an isometry,
\begin{align*}
 \|\widetilde I_k^*g^{>M}-J_kg^{>M}\|
 &\leq\|g^{>M}\|+\|J_kg^{>M}\|,\\
 \|\widetilde I_kJ_kg^{>M}-g^{>M}\|
 &\leq\|J_kg^{>M}\|+\|g^{>M}\|.
\end{align*}
The charge-Schwartz condition gives $\|g^{>M}\|\to0$, and
Lemma~\ref{lem:global-sampling} gives
$\sup_k\|J_kg^{>M}\|\to0$.  First let $k\to\infty$ on the
finite-charge part and then let $M\to\infty$.  The two conditions in
Eq.~\eqref{eq:admissible-embedding-hypotheses} follow, and
Lemma~\ref{lem:admissible-embedding-robustness} applies.
\end{proof}

If the cells satisfy the preceding criterion, their embeddings give
$\rho_0$ itself.  Multiplying the embeddings by phases with a strong-$*$
limit conjugates this realization.  Suppose that
$\widetilde I_k=M_{\phi_k}I_k$, where $|\phi_k|=1$ almost everywhere, and
that $M_{\phi_k}\to M_\phi$ strongly-$*$.  The corresponding extensions
satisfy the exact identity
\[
 \widetilde I_k\rho_k(\gamma)\widetilde I_k^*
 +(1-\widetilde I_k\widetilde I_k^*)
 =M_{\phi_k}\rho_k^\#(\gamma)M_{\phi_k}^*.
\]
Consequently, for every fixed $\gamma\in SL(2,\mathbb Z)$,
\begin{equation}
 \widetilde I_k\rho_k(\gamma)\widetilde I_k^*
 +(1-\widetilde I_k\widetilde I_k^*)
 \xrightarrow{\mathrm{SOT}}
 M_\phi\rho_0(\gamma)M_\phi^*,
 \label{eq:phase-conjugate-modular-limit}
\end{equation}
and the same convergence holds for the adjoints.

% section_06.tex
\section[The large-k limit of the boundary theories]
{The large-$k$ limit of the boundary theories}
\label{sec:boundary-theories}

This section records boundary evidence compatible with the regular-sector bulk
limit, but it is not used in the proofs of the bulk modular-convergence theorem.  Viewing the
three-dimensional TQFT as the SymTFT of a boundary suggests a corresponding
sequence of boundary theories $\mathcal B_k[G]\to\mathcal B[G]$, associated
respectively with boundary conditions $\mathcal B_k$ of $\DCS{k}{G}$ and
$\mathcal B$ of $BF[G]$; neither the boundary condition nor the comparison across
levels is canonical.  We focus on the CS--WZW
correspondence~\cite{Witten:1988hf,Elitzur:1989nr} and on a pair of WZW factors.
For $G=U(1)$, subsection~\ref{sec:U1_WZW} gives an exact comparison of the
wavefunctions and momentum lattices in the stated polarizations, including the
decompactifying radius.For compact simple $G$ under the hypotheses stated
below,
subsection~\ref{sec:G_WZW} gives four compatible checks for a sigma-model
candidate with target $\mathfrak g$. Subsection~\ref{sec:nonabelian-BF-wavefunctions} then proposes a regular-sector ansatz
for the nonabelian $BF$ boundary wavefunctions and compares its classical part with the large-$k$ limit of doubled Chern--Simons theory.
Finally, subsection~\ref{sec:flat-gauging} studies the formal flat gauging of the
free scalar candidate through its twisted partition functions, including the
infrared divergences and singular sectors that must be addressed to define the
gauged theory. We do not prove convergence of the full
boundary CFT, its local fields, OPE coefficients, sewing data, or correlation
functions.

\subsection[The U(1) case]{Boundary wavefunction for the $U(1)$ case}\label{sec:U1_WZW}

We begin with $G=U(1)$, where the wavefunction calculation can be carried out in
closed form and the relevant limit is visible in the momentum lattice of a
compact boson.  The abelian Lie algebra is of course not semisimple, so this case
falls outside the hypotheses used elsewhere in this paper.  Here and below,
``exact'' refers to the wavefunctions, lattice identities, and radii obtained
after the stated choices of polarization and symmetry boundary.

We take $k$ to be positive and even in this abelian example, so that $U(1)_k$ is a bosonic Chern--Simons theory.  Odd levels require spin structures and their modular transformations, which are not included here.

Throughout this subsection, we use the Euclidean action and the weight
$e^{-S}$.  This is the Wick-rotated form of the topological normalization in
section~\ref{sec:Limit_Regular_Simples}; the diagonal and rescaled antidiagonal fields
themselves will be normalized exactly as in~\eqref{eq:classical-field-recombination}.
For radii we use the convention
\begin{equation}\label{eq:compact-boson-radius-convention}
 p_{L,R}=\frac1{\sqrt2}\left(\frac{n}{R}\mathbin{\pm}wR\right),
 \qquad (h,\bar h)=\left(\frac{p_L^2}{2},\frac{p_R^2}{2}\right).
\end{equation}

\paragraph{$U(1)$ Chern-Simons theory $\text{CS}_k$[$U(1)$]}

As a warm-up, we review the wavefunction for the $U(1)$ Chern-Simons theory following \cite{Elitzur:1989nr}. The method can be generalized to noncompact abelian $BF$ theory.

Our starting point is $\CS{k}{U(1)}$ on a three-manifold $\mathcal M$ whose boundary is a torus of modular parameter $\tau$, parametrized by $z=\sigma^1+\tau\sigma^2$ with $\sigma^i\in[0,1)$, with the action
\begin{equation}\label{eq:U1CS-action}
S \;=\; \frac{ik}{4\pi}\int_{\mathcal M}A\wedge dA \,.
\end{equation}
Varying \eqref{eq:U1CS-action} leaves a boundary term
    \begin{equation}
        \delta S\big|_{dA=0}=-\frac{ik}{4\pi}\int_{\partial\mathcal{M}}A \wedge \delta A= -\frac{k}{2\pi} \int_{\partial \mathcal{M}} d^2z (A_z \delta A_{\bar{z}}-A_{\bar{z}}\delta A_z)\,,
    \end{equation}
where the holomorphic and antiholomorphic components are
    \begin{equation}
        A_z = \frac{A_2 - \bar{\tau}A_1}{2 i \tau_2}\,,\quad A_{\bar{z}}=\frac{\tau A_1-A_2}{2i \tau_2}\,,
    \end{equation}
and we fix the conventions
    \begin{equation}
        dz \wedge d \bar{z} = -2 i \tau_2 d \sigma^1 \wedge d\sigma^2\,,\quad d^2 z \equiv \frac{i}{2} dz \wedge d \bar{z}\,,\quad \int d^2z = \tau_2\,.
    \end{equation}

We add the standard
counterterm at the boundary~\cite{Elitzur:1989nr} and consider the action
\begin{equation}
 S'=S-\frac{k}{2\pi}\int_{\partial\mathcal M}d^2z\,A_zA_{\bar z}.
 \label{eq:U1-CS-boundary-counterterm}
\end{equation}
The boundary variation of $S'$ vanishes when $A_{\bar z}$ is fixed, while $A_z$ may vary.  This specifies the polarization used in the path integral. Introduce the wavefunction $\Psi(A_{\bar{z}})$ as
    \begin{equation}
        \Psi(A_{\bar{z}}) = \int_{A_{\bar{z}}|_{\partial \mathcal{M}}\text{ fixed}} \mathcal{D}A e^{-S'}\,,
    \end{equation}
which depends on the boundary value of $A_{\bar{z}}$. A large gauge transformation $A\rightarrow A+\omega$ at the boundary multiplies the wavefunction by a factor assembled
from the bulk action
    \begin{equation}
        \Psi(A_{\bar{z}}) \rightarrow \Psi(A_{\bar{z}}) e^{\frac{ik}{4\pi} \int_{\partial \mathcal{M}} \omega \wedge A}e^{\frac{k \tau_2}{2\pi}(A_z \omega_{\bar{z}}+A_{\bar{z}}\omega_z+|\omega_z|^2)}e^{i \varphi(\omega)}\,,
    \end{equation}
where the first term is from $S$, the second term is from the boundary term $S_{\partial}$, and the last term is an additional cocycle to assure associativity of large gauge
transformations.

These are transition functions for the prequantum line bundle over the
boundary flat-connection space.  We use the holomorphic trivialization of
\cite{Elitzur:1989nr}.  A large gauge transformation of the boundary need not
extend over a fixed filling, so its multiplier is specified by this line bundle
and its cocycle, rather than by treating every such transformation as a bulk
gauge redundancy.

Parametrize the value of $A_{\bar{z}}$ at the boundary as
    \begin{equation}
        A_{\bar{z}} = \frac{i \pi}{\tau_2} \xi\,,
    \end{equation}
and the large gauge transformation shifts $\xi$ according to
    \begin{equation}
        \xi \rightarrow \xi +n+m\tau\,,\quad n,m\in \mathbb{Z}\,.
    \end{equation}
Then we can write
    \begin{equation}
        \Psi(\xi +n+m\tau) = \Psi(\xi) \exp \left[ \frac{k \pi}{\tau_2}(n+m\bar{\tau})\xi +\frac{k\pi}{2\tau_2}|n+m\tau|^2+i\pi k nm \right]\,,
    \end{equation}
where we choose the cocycle as $\varphi = \pi k n m$. It is useful to introduce
$f=\Psi\,e^{-k\pi\xi^2/2\tau_2}$, which reduces the multiplier to
\begin{equation}\label{eq:U1-theta-quasiperiodicity}
f(\xi+1)=f(\xi)\,,\qquad
f(\xi+\tau)=f(\xi)\,e^{-2\pi ik\xi-i\pi k\tau}\,,
\end{equation}
which is the quasi-periodicity of a level-$k$ theta function. Since $f(\xi)$ is periodic along the real axis, we can do Fourier expansion
    \begin{equation}
        f=\sum_{\lambda \in \mathbb{Z}}c_{\lambda} e^{2\pi i \lambda \xi}\,,
    \end{equation}
and it turns the second relation in \eqref{eq:U1-theta-quasiperiodicity} into the recursion relation
    \begin{equation}
        c_{\lambda+k} = c_{\lambda} e^{2\pi i \lambda \tau + \pi i k \tau}\,,
    \end{equation}
which connects $c_{\lambda}$ to $c_{\lambda+k}$, therefore the solution space is $k$-dimensional. Solving the recursion relation using the ansatz $c_{\lambda} = e^{i\pi \tau\lambda^2/k}$, and setting $\lambda=\sqrt k\,p$ with $p=\sqrt k\,(n+j/k)$ and $u=\sqrt k\,\xi$, the $k$ independent wavefunctions read
\begin{equation}\label{eq:U1-CS-wavefunctions}
\Psi_j(\xi)\;=\;\frac{1}{\eta(\tau)}\sum_{n\in\mathbb Z}
\exp\!\Big[i\pi\tau p^2+2\pi ipu+\tfrac{\pi}{2\tau_2}u^2\Big]\,,
\qquad j=0,\dots,k-1\,,
\end{equation}
with $1/\eta(\tau)$ the one-loop determinant of the topologically trivial fluctuations. We recover the standard count of $k$ states, and $j$ is the anyon label of $U(1)_k$.

\paragraph{Non-compact abelian BF theory}

The calculation for non-compact abelian $BF$ theory follows the same steps, but only the compact field admits large gauge transformations. The action is
    \begin{equation}\label{eq:U1-BF-action}
        S = \frac{i}{2\pi} \int_{\mathcal{M}} b \wedge da \;\simeq\; \frac{i}{4\pi}\int_{\mathcal{M}} b\wedge da+\frac{i}{4\pi}\int_{\mathcal{M}} a\wedge db\,,
    \end{equation}
where $\simeq$ denotes equality after a local boundary counterterm.  Here $a$ remains a compact $U(1)$ gauge field while $b$ is valued in $\mathbb R$, so that $b$ admits no large gauge transformations. We use the symmetric representative in the calculation that follows.

Introduce the combinations
    \begin{equation}\label{eq:U1-BF-CS}
        A_{\pm} = \frac{1}{\sqrt{2}}\left(\frac{a}{r}\pm rb \right)\,,
    \end{equation}
so that the $BF$ action \eqref{eq:U1-BF-action} can be written as a pair of Chern--Simons theories at opposite unit level, 
    \begin{equation}
        S=\frac{i}{4\pi}\int (A_+\wedge dA_+-A_-\wedge dA_-)\,,
    \end{equation}
where $r$ is a real and positive parameter. Consider the two \emph{independent} holomorphic coordinates parametrized by
    \begin{equation}
        \xi = \frac{\tau_2}{i \pi} (A_+)_{\bar{z}}=\frac{1}{\sqrt{2}}\left(\frac{u}{r}+vr \right)\,,\quad \bar{\xi}=-\frac{\tau_2}{i\pi}(A_-)_{z}=\frac{1}{\sqrt{2}}\left(\frac{\bar{u}}{r}-\bar{v}r \right)\,,
    \end{equation}
with
    \begin{equation}
        u=\frac{\tau_2}{i\pi}a_{\bar{z}}\,,\quad v=\frac{\tau_2}{i\pi} b_{\bar{z}}\,.
    \end{equation}
Deform the action with the proper boundary term
    \begin{equation}
        S'=S-\frac{1}{2\pi} \int_{\partial \mathcal{M}}d^2 z\left(|A_+|^2+|A_-|^2 \right)\,,
    \end{equation}
and we can repeat the computation for the $U(1)$ Chern-Simons case with the two independent holomorphic coordinates $\xi$ and $\bar\xi$.

The large gauge transformation is generated purely by $a$, and we have
    \begin{equation}
        \xi \rightarrow \xi + \frac{1}{\sqrt{2}r}(n+m\tau)\,,\quad \bar{\xi}\rightarrow \bar{\xi}+\frac{1}{\sqrt{2}r}(n+m\bar{\tau})\,.
    \end{equation}
Denote $ \delta=\frac{1}{\sqrt{2}r}(n+m\tau)$, the wavefunction transforms as follows
    \begin{equation}
        \Psi(\xi+\delta,\bar{\xi}+\delta^*)=\Psi(\xi,\bar{\xi}) \exp \left[\frac{\pi}{\tau_2}(\xi \delta^*+\bar{\xi}\delta+|\delta|^2)\right]\,,
    \end{equation}
which is associative under the large gauge transformation. Similarly, let us introduce $f=\Psi e^{-\frac{\pi}{2\tau_2}(\xi^2 + \bar{\xi}^2)}$, and one can check the large gauge transformation is generated by
    \begin{equation}
    \begin{gathered}
        f\left(\xi +\frac{1}{\sqrt{2}r},\bar{\xi} +\frac{1}{\sqrt{2}r}\right)=f(\xi,\bar{\xi})\,,\\
        f\left(\xi+ \frac{\tau}{\sqrt{2}r},\bar{\xi}+ \frac{\bar{\tau}}{\sqrt{2}r}\right)=f\exp \left[ -\frac{2\pi i}{\sqrt{2}r} (\xi-\bar{\xi})+\frac{\pi \tau_2}{r^2}\right]\,.
    \end{gathered}
    \end{equation}
Then let us consider Fourier expansion and assume
    \begin{equation}
        f=\sum_{p_L,p_R} c_{p_L,p_R} e^{2\pi i p_L \xi -2\pi i p_R\bar{\xi}}\,,
    \end{equation}
where the summation runs over all possible values of $p_L,p_R$. The first relation implies
    \begin{equation}
        e^{\frac{2\pi i (p_L-p_R)}{\sqrt{2}r}} = 1\quad \Rightarrow \quad p_L - p_R \in \sqrt{2}r \mathbb{Z}\,,
    \end{equation}
and the other relation gives the recursion relation
    \begin{equation}
        c_{p_L+\frac{1}{\sqrt{2}r},p_R+\frac{1}{\sqrt{2}r}}=c_{p_L,p_R} e^{\frac{2\pi i (p_L \tau - p_R \bar{\tau})}{\sqrt{2}r}- \frac{\pi \tau_2}{r^2}}\,,
    \end{equation}
and we take the ansatz
    \begin{equation}
        c_{p_L,p_R}=\exp \left[i\pi \tau p_L^2-i \pi \bar{\tau}p_R^2 \right]\,.
    \end{equation}

Based on the above discussion, we can parametrize
    \begin{equation}
        p_L = \frac{1}{\sqrt{2}}\left(\frac{\ell +\varphi}{r}+rN \right)\,,\quad p_R = \frac{1}{\sqrt{2}}\left(\frac{\ell +\varphi}{r}-rN \right)\,,
    \end{equation}
with $\varphi\in[0,1)$ and $N\in \mathbb{Z}$, and they satisfy $p_L-p_R=\sqrt{2}rN$ as required. Then the independent wavefunctions are
    \begin{equation}\label{eq:U1-BF-wavefunctions}
        \Psi_{\varphi,N}(\xi,\bar{\xi}) = \frac{1}{|\eta|^2}\sum_{\ell \in \mathbb{Z}} \exp \left[i\pi \tau p_L^2 - i \pi \bar{\tau}p_R^2 +2\pi i (p_L \xi -p_R\bar{\xi})+\frac{\pi}{2\tau_2}(\xi^2 + \bar{\xi}^2)\right]\,.
    \end{equation}
Notice that the labels $(\varphi,N)$ of the wavefunction are exactly the same as the labels of the bulk line $W_{\varphi,N}$. In the bulk $BF$ theory, $\varphi$ parametrizes the continuous family of the Wilson line of $b$, and $N$ is the discrete electric charge associated to the Wilson line of $a$. While in the boundary momentum lattice, the same $N$ occupies the winding slot and
$\ell+\varphi$ occupies the momentum slot.  

In the SymTFT picture, the wavefunctions $\Psi_{\varphi,N}(\xi,\bar{\xi})$ actually give the physical boundary amplitudes. To retrieve the partition function, one needs to include the topological boundary determined by the Lagrangian algebra $\mathcal{L}$. Formally, suppose the wavefunction is labeled by $\Psi_i$, and the Lagrangian algebra is $ \mathcal{L} = \bigoplus_i n_i W_i$, where $i$ labels the anyons in the bulk. Then the partition function associated with $\mathcal{L}$ is written as
    \begin{equation}
        Z_{\mathcal{L}} = \sum_i n_i \Psi_i\,.
    \end{equation}
    
In the present case, the bulk anyon is labeled by $W_{\varphi,N}$. Let us first consider
    \begin{equation}
        \mathcal{L}^{BF}_1 = \bigoplus_N W_{0,N}\,,
    \end{equation}
so that the Wilson line of $a$ can end on the topological boundary, and we can turn on the $a$-background on it. Setting $v=0$ we have
    \begin{equation}
        \xi =\frac{1}{\sqrt{2}} \frac{u}{r}\,,\quad \bar{\xi}=\frac{1}{\sqrt{2}} \frac{\bar{u}}{r}\,,
    \end{equation}
which gives
    \begin{equation}
        Z_{\mathcal{L}^{BF}_1} = \frac{1}{|\eta|^2}\sum_{\ell,N \in \mathbb{Z}} \exp \left[i\pi \tau p_L^2 - i \pi \bar{\tau} p_R^2 +2\pi i (p_L \xi -p_R\bar{\xi})+\frac{\pi}{2\tau_2}(\xi^2 + \bar{\xi}^2)\right]\,,
    \end{equation}
with
    \begin{equation}
        p_L = \frac{1}{\sqrt{2}}\left(\frac{\ell}{r}+rN \right)\,,\quad  p_R = \frac{1}{\sqrt{2}}\left(\frac{\ell}{r}-rN \right)\,.
    \end{equation}
Set $R=r/\sqrt{2}$ and rewrite
    \begin{equation}
        e^{-\frac{\pi}{2\tau_2}(|\xi|^2+|\bar{\xi}|^2)} Z_{\mathcal{L}^{BF}_1} = \frac{1}{|\eta|^2}\sum_{\ell,N \in \mathbb{Z}} e^{2\pi i u_2 N } q^{\frac{1}{2}\left(\frac{\ell+u_1}{2R} + RN \right)^2} \bar{q}^{\frac{1}{2}\left(\frac{\ell+u_1}{2R} - RN \right)^2}\,,
    \end{equation}
where we set $u=u_2 + \tau u_1$. The RHS is exactly the partition function of a compact boson with a generic $U(1)_W$ background parametrized by the spatial and temporal holonomies $(u_1,u_2)$, up to a local Gaussian factor $e^{-\frac{\pi}{2\tau_2}(|\xi|^2+|\bar{\xi}|^2)}$ which can be canceled by adding a local counterterm on the boundary, as in~\eqref{eq:U1-CS-boundary-counterterm}.

Alternatively, we can consider another Lagrangian algebra
    \begin{equation}
        \mathcal{L}^{BF}_2 = \int_0^1 d\varphi W_{\varphi,0}\,,
    \end{equation}
so that the Wilson line of $b$ can end on the topological boundary, and we can turn on the $b$-background on it. Setting $u=0$ we have
    \begin{equation}
        \xi = \frac{r}{\sqrt{2}}v\,,\quad \bar{\xi}=-\frac{r}{\sqrt{2}}\bar{v}\,,
    \end{equation}
and also 
    \begin{equation}
        p_L = p_R = \frac{1}{\sqrt{2}r} (\ell+\varphi)\,.
    \end{equation}
The integral is evaluated as
    \begin{equation}
    \begin{split}
        Z_{\mathcal{L}_2^{\text{BF}}} =& \frac{1}{|\eta|^2} \int_{-\infty}^{+\infty} \sqrt{2}rd\lambda \exp \left( -2\pi \tau_2\lambda^2+2\pi i \lambda (\xi-\bar{\xi}) \right) \exp\left(+\frac{\pi}{2\tau_2}(\xi^2 + \bar{\xi}^2) \right)\\
        =&\frac{\sqrt{2}R}{|\eta|^2 \sqrt{\tau_2}} e^{\frac{\pi}{\tau_2}\xi \bar{\xi}}\,,
    \end{split}
    \end{equation}
so that one has
    \begin{equation}
        e^{-\frac{\pi}{2\tau_2}(|\xi|^2+|\bar{\xi}|^2)} Z_{\mathcal{L}^{BF}_2} = \frac{\sqrt{2}R}{|\eta|^2 \sqrt{\tau_2}} e^{-\frac{2\pi R^2}{\tau_2}|v_2+\tau v_1|^2}\,,
   \end{equation}
where we set $v = v_2 +\tau v_1$. This is the Gaussian partition-function density of a noncompact boson with a translation background.  The continuum-state normalization fixes the prefactor; the infinite target zero-mode volume of an ordinary noncompact trace has been divided out.  It should therefore be compared with a volume-normalized decompactification limit, rather than with the unnormalized trace of a compact theory.

\paragraph{Double $U(1)$ Chern-Simons theory in the large $k$ limit}

Then we turn to the double copy of Chern-Simons theory $U(1)_k \times U(1)_{-k}$, and consider the large $k$ limit. We will show that, with a proper choice of boundary condition, the partition function in the large $k$ limit agrees with those of $U(1)$ non-compact $BF$ theory we just discussed.

Begin with the double Chern-Simons theory
    \begin{equation}
        S=\frac{ik}{4\pi} \int a_+\wedge da_+ -\frac{ik}{4\pi} \int a_- \wedge da_-\,,
    \end{equation}
and we can introduce the $BF$ fields 
    \begin{equation}
        a=\frac{1}{2}(a_++a_-)\,,\quad b= k(a_+-a_-)\,,
    \end{equation}
where the period of $b$ goes to infinity as $k\rightarrow +\infty$. To compare with the previous result, we redefine the Chern-Simons fields according to \eqref{eq:U1-BF-CS} by setting
    \begin{equation}
        A_{\pm} =\frac{1}{\sqrt{2}}\left(\frac{a}{r}\pm rb \right) = \frac{1}{\sqrt{2}} \left( \frac{a_++a_-}{2r}\pm kr (a_+-a_-) \right)\,,
    \end{equation}
and introduce the same holomorphic coordinates as
    \begin{equation}
        \xi = \frac{\tau_2}{i \pi} (A_+)_{\bar{z}}\,,\quad \bar{\xi}=-\frac{\tau_2}{i\pi}(A_-)_{z}\,.
    \end{equation}

We have two independent large gauge transformations for $a_+$ and $a_-$. Denote 
    \begin{equation}
        \mu = \frac{1}{\sqrt{2}}\left(\frac{1}{2r}+kr \right)\,,\quad \nu=\frac{1}{\sqrt{2}}\left(\frac{1}{2r}-kr \right)\,,
    \end{equation}
which satisfies
    \begin{equation}
        \mu+\nu=\frac{1}{\sqrt{2}r}\,,\quad \mu-\nu=\sqrt{2}kr\,,\quad \mu^2-\nu^2=k\,.
    \end{equation}
The large gauge transformation reads
    \begin{equation}
        \delta \xi = \mu D_++\nu D_-\equiv E\,,\quad \delta \bar{\xi} = \nu \bar{D}_+ + \mu \bar{D}_- \equiv F^*\,,
    \end{equation}
where 
    \begin{equation}
        D_{\pm} = n_{\pm}+ m_{\pm}\tau \,,\quad \bar{D}_{\pm} = n_{\pm} + m_{\pm} \bar{\tau}\,.
    \end{equation}
The cocycle $\varphi$ satisfies the composition law
    \begin{equation}
        \varphi(\omega+\omega') = \varphi(\omega)+\varphi(\omega')-\frac{1}{4\pi}\int_{\partial \mathcal{M}} (\omega_{A_+} \wedge \omega'_{A_+}-\omega_{A_-} \wedge \omega'_{A_-})\,,
    \end{equation}
where
    \begin{equation}
        \omega_{A_+} = \mu \omega_+ + \nu \omega_-\,,\quad \omega_{A_-} =\nu \omega_+ + \mu \omega_-\,,
    \end{equation}
so that
    \begin{equation}
        \omega_{A_+} \wedge \omega'_{A_+} - \omega_{A_-} \wedge \omega'_{A_-}=(\mu^2-\nu^2)(\omega_+\wedge \omega'_{+}-\omega_-\wedge \omega'_-)\,,
    \end{equation}
and is solved by
    \begin{equation}
        \varphi = \pi k(m_+ n_+-m_- n_-) \,.
    \end{equation}
    
The wavefunction transforms under the large gauge transformations according to
\begin{equation}
    \Psi(\xi +E,\bar{\xi}+F^*) = \Psi(\xi,\bar{\xi}) \exp \left[\frac{\pi}{2\tau_2}\left(2\xi E^*+EE^* \right) +\frac{\pi}{2\tau_2} \left(2\bar{\xi} F+FF^* \right) + i\varphi\right]\,.
\end{equation}
Set $\Psi=fe^{\frac{\pi}{2\tau_2}(\xi^2 + \bar{\xi}^2)}$, one can check that $f(\xi +E,\bar{\xi}+F^*) = f(\xi,\bar{\xi}) \times \delta f$
transforms under the four generators as the following table
\begin{center}
\begin{tabular}{ccc}
\toprule
$(n_+,m_+,n_-,m_-)$ & $(\delta\xi,\delta\bar\xi)$ & $\delta f$\\
\midrule
$(1,0,0,0)$ & $(\mu,\nu)$ & $1$\\
$(0,0,1,0)$ & $(\nu,\mu)$ & $1$\\
$(0,1,0,0)$ & $(\mu\tau,\nu\bar\tau)$ &
$e^{-2\pi i\mu\xi+2\pi i\nu\bar\xi-i\pi\mu^{2}\tau+i\pi\nu^{2}\bar\tau}$\\
$(0,0,0,1)$ & $(\nu\tau,\mu\bar\tau)$ &
$e^{-2\pi i\nu\xi+2\pi i\mu\bar\xi-i\pi\nu^{2}\tau+i\pi\mu^{2}\bar\tau}$\\
\bottomrule
\end{tabular}
\end{center}
Consider the Fourier expansion
    \begin{equation}
        f(\xi,\bar{\xi})=\sum_{p_L,p_R} c_{p_L,p_R} e^{2\pi i (p_L \xi -p_R \bar{\xi})}\,,
    \end{equation}
and the first two generators give
    \begin{equation}
        p_L \mu-p_R\nu \in \mathbb{Z}\,,\quad p_L \nu - p_R \mu \in \mathbb{Z}\,.
    \end{equation}
Parametrize $p_L,p_R$ according to
    \begin{equation}
        p_L=\mu \hat{n}_+ + \nu \hat{n}_-\,,\quad p_R = \nu \hat{n}_++\mu \hat{n}_-\,,
    \end{equation}
we have
    \begin{equation}
        p_L \mu-p_R\nu = (\mu^2-\nu^2)\hat{n}_+=k \hat{n}_+\,,\quad p_L \nu - p_R \mu = (\nu^2-\mu^2)\hat{n}_-=-k \hat{n}_-\,,
    \end{equation}
so one can further parametrize
    \begin{equation}
        \hat{n}_{\pm} = \frac{j_{\pm}}{k}\,, \quad j_{\pm} \in \mathbb{Z}\,.
    \end{equation}

Then consider the other two generators. The third one shifts $(p_L,p_R)\rightarrow(p_L-\mu,p_R-\nu)$, which is effectively 
    \begin{equation}
        \hat{n}_+\rightarrow \hat{n}_+-1\,,\quad \text{or}\quad j_+\rightarrow j_+-k\,,
    \end{equation}
and it gives the recursion relation
    \begin{equation}
        c_{p_L+\mu,p_R +\nu} = c_{p_L,p_R}e^{2\pi i (\mu \tau p_L - \nu \bar{\tau}p_R)+i\pi \mu^2 \tau - i\pi \nu^2 \bar{\tau}}\,.
    \end{equation}
Similarly, the fourth one shifts $(p_L,p_R)\rightarrow(p_L-\nu,p_R-\mu)$, which is effectively
    \begin{equation}
        \hat{n}_-\rightarrow \hat{n}_--1\,,\quad \text{or}\quad j_-\rightarrow j_--k\,,
    \end{equation}
and it gives the recursion relation
    \begin{equation}
        c_{p_L+\nu,p_R +\mu} = c_{p_L,p_R}e^{2\pi i (\nu \tau p_L - \mu \bar{\tau}p_R)+i\pi \nu^2 \tau - i\pi \mu^2 \bar{\tau}}\,.
    \end{equation}
Both recursion relations are solved by the ansatz
    \begin{equation}
        c_{p_L,p_R} = \exp \left[i \pi \tau p_L^2-i\pi \bar{\tau}p_R^2 \right]\,,
    \end{equation}
and there are exactly $k^2$ solutions.

Based on the above discussion, let us write
    \begin{equation}
        j_{\pm}=k n_{\pm}+\lambda_{\pm}\,, \quad \lambda_{\pm}\in \{0,\cdots,k-1 \}\,, \quad n_{\pm}\in \mathbb{Z}\,,
    \end{equation}
and the $k^2$ independent wavefunctions are
    \begin{equation}
        \Psi_{\lambda_+,\lambda_-}(\xi,\bar{\xi},\tau) = \frac{1}{|\eta|^2} \sum_{n_+,n_-} \exp \left[i \pi \tau p_L^2 - i \pi \bar{\tau} p_R^2+2\pi i (p_L \xi -p_R \bar{\xi})+ \frac{\pi}{2\tau_2}(\xi^2+\bar{\xi}^2) \right]\,,
    \end{equation}
with
\begin{equation}
    p_L=\frac{1}{\sqrt{2}}\left[\frac{j_++j_-}{2kr}+r(j_+-j_-) \right]\,,\quad p_R=\frac{1}{\sqrt{2}}\left[\frac{j_++j_-}{2kr}-r(j_+-j_-) \right]\,.
\end{equation}

We then consider the partition function built by choosing a Lagrangian algebra. Denote the anyons as $W_{\lambda_+,\lambda_-}$, we first consider the diagonal one
    \begin{equation}
        \mathcal{L}^{CS}_1 = \bigoplus_{\lambda} W_{\lambda,\lambda}\,,
    \end{equation}
which sets $\lambda_+=\lambda_-=\lambda$. We have
    \begin{equation}
        \frac{j_++j_-}{2kr}=\frac{k(n_++n_-)+2\lambda}{2kr}\,, \quad r(j_+-j_-)=kr(n_+-n_-)\,,
    \end{equation}
and we label
    \begin{equation}
        w=n_+-n_-\,,\quad m = \frac{k}{2}(n_++n_-)+\lambda = \frac{k}{2}w+k n_-+\lambda\,.
    \end{equation}
For the Chern-Simons to be non-spin, we choose $k$ even, then $w,m$ can be chosen as independent integers. Thus we have
    \begin{equation}
        p_L = \frac{1}{\sqrt{2}}\left(\frac{m}{kr}+ kr w\right)\,,\quad p_R=\frac{1}{\sqrt{2}}\left(\frac{m}{kr}- kr w\right)\,,
    \end{equation}
which matches the momenta of a compact boson with radius $kr$.

On the other hand, if we choose the off-diagonal one
    \begin{equation}
        \mathcal{L}^{CS}_2 = \bigoplus_{\lambda\in\mathbb Z/k\mathbb Z} W_{\lambda,-\lambda}\,,
    \end{equation}
and we set $\lambda_+=\lambda,\lambda_-=k-\lambda$. Then one has
    \begin{equation}
        \frac{j_++j_-}{2kr}=\frac{n_++n_-+1}{2r}\,, \quad r(j_+-j_-)=kr(n_+-n_--1)+2\lambda r\,.
    \end{equation}
We then set
    \begin{equation}
        m=n_++n_-+1\,,
    \end{equation}
and
    \begin{equation}
        w=\frac{(n_+-n_--1)k}{2}+\lambda =  \frac{(m-2n_--2)k}{2}+\lambda = \frac{k}{2}m+\lambda -(n_-+1)k\,,
    \end{equation}
and both $m,w$ can take arbitrary integers when $k$ is even. Therefore we have
    \begin{equation}
        p_L = \frac{1}{\sqrt{2}}\left(\frac{m}{2r}+2rw \right)\,,\quad p_R = \frac{1}{\sqrt{2}} \left(\frac{m}{2r}-2rw \right)\,,
    \end{equation}
which is independent of $k$ and corresponding to the compact boson with radius $2r$. 

As $k\to\infty$, the diagonal boundary decompactifies. The boundary defined by
charge conjugation retains radius $2r$.  At zero background the diagonal partition function grows as $kr/(\sqrt{\tau_2}|\eta(\tau)|^2)$; in the continuum normalization above, $k^{-1}Z_{\mathcal L^{CS}_1}\to Z_{\mathcal L^{BF}_2}$.  Thus the decompactifying boundary comparison includes division by the growing zero-mode volume.  The factor of two for the compact boundary is fixed by
the compact gauge lattice and must be kept when comparing with the $BF$ radius
$r$ above.  In the $BF$ variables, the diagonal family approaches the
flux-condensing boundary $\mathcal L^{BF}_2$.  For the charge-conjugation family,
$j_++j_-=km$ and $j_+-j_-=2w$ at even $k$, so the limiting labels satisfy
$\varphi=m/2\pmod1$ and $N=2w$.  Its boundary sum is therefore
\begin{equation}\label{eq:abelian-charge-conjugation-limit}
 \mathcal L^{CS}_2\ \longrightarrow\
 \mathcal L^{BF}_{\mathrm{cc}}
 :=\bigoplus_{w\in\mathbb Z}
 \left(W_{0,2w}\oplus W_{1/2,2w}\right).
\end{equation}
Summing the two flux offsets converts $\ell+\varphi$ into $m/2$ and gives
$p_{L,R}=2^{-1/2}(m/(2r)\pm2rw)$, as required.  This is a different symmetry
boundary from $\mathcal L^{BF}_1$ at fixed $r$.  Their radii can be matched by
changing the boundary parameter, but the two boundary algebras are not equal.

\subsection[The limiting theory for general G]
{The limiting theory for general $G$}\label{sec:G_WZW}

For $G$ compact, connected, simple, and simply connected, we give four
complementary checks of the candidate boundary theory: a sigma model into
$\mathfrak g$, or $\dim G$ non-compact free scalars.
They test central charge, a local action expansion, global symmetry, and current
algebra, respectively.  Agreement of these checks is evidence for the candidate,
not a uniqueness statement or a proof of convergence of the full CFT.

The first is the central charge. The Sugawara construction gives the level-$k$ WZW model~\cite{Knizhnik:1984nr} central charge
\begin{equation}\label{eq:sugawara-c}
c_k \;=\; \frac{k\,\dim G}{k+h^\vee}\,,
\end{equation}
so that $\lim_{k\rightarrow \infty}c_k = \dim G$. A boundary theory for $\DCS{k}{G}$ built from a chiral and an antichiral WZW factor therefore approaches central charge $(\dim G,\dim G)$, consistent with $\dim G$ real non-compact non-chiral scalars.  Central charge alone is not diagnostic: many inequivalent CFTs can share it.

The second piece of evidence comes from the action itself. The sigma-model metric on the target is proportional to $-k\,\operatorname{tr}(\,\cdot\,\cdot\,)$, so the characteristic radius of the target grows like $\sqrt k$ while its curvature falls like $1/k$. To keep fluctuations canonically normalized we zoom into a finite region of the tangent space at the identity by writing
\begin{equation}\label{eq:g-exp-scaling}
g=\exp\!\Big(\frac{X}{\sqrt k}\Big)\,,\qquad X=X^at_a\in\mathfrak g\,,
\end{equation}
hence the Maurer--Cartan form can be expanded as
\begin{equation}\label{eq:MC-expansion}
g^{-1}dg=\frac{1}{\sqrt k}\,dX-\frac{1}{2k}\,[X,dX]
+\frac{1}{6k^{3/2}}\,[X,[X,dX]]+O(k^{-2})\,.
\end{equation}
Substituting into the kinetic term $-\frac{k}{4\pi}\int d^2z\,\operatorname{tr}\bigl((g^{-1}\partial g)(g^{-1}\bar\partial g)\bigr)$, one would expect a cubic vertex at order $k^{-1/2}$, whose coefficient is proportional to
\begin{equation}\label{eq:cubic-cancellation}
-\operatorname{tr}\bigl(\partial X[X,\bar\partial X]\bigr)-\operatorname{tr}\bigl([X,\partial X]\bar\partial X\bigr)\;=\;0\,,
\end{equation}
and this vanishes by cyclicity of the trace. The kinetic term is therefore free through order $k^{-1/2}$, and its first correction is a quartic vertex at order $1/k$ coming from the curvature of the group metric. The Wess--Zumino term does survive at the intermediate order, since $-\operatorname{tr}(\theta\wedge[\theta,\theta])$ scales as $k^{-3/2}$ against a prefactor of $k$, leaving a cubic interaction suppressed by $k^{-1/2}$. Assembling the two pieces, we have
\begin{equation}\label{eq:WZW-expansion}
S_{\mathrm{WZW},k}\big[e^{X/\sqrt k}\big]
=-\frac{1}{4\pi}\int_\Sigma d^2z\;\operatorname{tr}(t_at_b)\,\partial X^a\bar\partial X^b
\;+\;\frac{1}{\sqrt k}\,S^{(3)}_{\mathrm{WZ}}\;+\;\frac{1}{k}\,S^{(4)}+\cdots\,,
\end{equation}
so that in an orthonormal basis, $-\operatorname{tr}(t_at_b)=\delta_{ab}$, the leading local action is that of $\dim G$ real non-compact free scalars. Geometrically this means that $G$ with metric $-k\,\operatorname{tr}(\,\cdot\,\cdot\,)$ flattens onto its tangent space $\mathfrak g$ with metric $-\operatorname{tr}(\,\cdot\,\cdot\,)$, isometric to $\mathbb R^{\dim G}$.  The boundary coordinate scaling $X/\sqrt{k}$ differs from the bulk field contraction, whose antidiagonal connection scales as $1/k$. Using $X/\sqrt{k+h^\vee}$ in place of \eqref{eq:g-exp-scaling} changes nothing at leading order.  This tangent-space expansion controls perturbative fluctuations in a fixed coordinate neighborhood.

Third, the candidate theory carries the global symmetry we should expect on a boundary of $\CS{}{TG}$. Take
\begin{equation}\label{eq:free-scalar-action}
S=-\int_\Sigma\operatorname{tr}(dX\wedge *dX)\,,\qquad X:\Sigma\to\mathfrak g_\mathrm{ab}\,,
\end{equation}
with $*$ the worldsheet Hodge star and the positive $\mathrm{Ad}$-invariant pairing $-\operatorname{tr}(XY)$. It is invariant under
\begin{equation}\label{eq:TG-symmetry}
X\;\longmapsto\;\mathrm{Ad}_g X+\alpha\,,\qquad g\in G\,,\quad\alpha\in\mathfrak g_\mathrm{ab}\,,
\end{equation}
and defines an action of $TG\simeq G\ltimes_{\mathrm{Ad}}\mathfrak g_{\mathrm{ab}}$.
Its kernel is $Z(G)\times\{0\}$, so its faithful image on the scalar fields is
$(G/Z(G))\ltimes_{\mathrm{Ad}}\mathfrak g_{\mathrm{ab}}$.
The full affine isometry group of the free action is larger,
$O(\dim G)\ltimes\mathbb R^{\dim G}$.  For example, for $G=SU(n)$ the compact
image is $PSU(n)\subseteq SO(n^2-1)$.  Thus the action checks compatibility with
the tangent-group contraction while leaving the center invisible.

The sharpest of the four checks is at the level of the current algebra, where the limit is an \.In\"on\"u--Wigner contraction~\cite{Inonu:1953sp,Saletan:1961iqo} with rescaling. Let $J_m^a$ and $\bar J_m^a$ be the standard holomorphic and antiholomorphic level-$k$ modes.  After folding the antiholomorphic sector we set
$\tilde J_m^a:=\bar J_{-m}^a$; the reversal of the mode number changes the sign of
its central term.  Thus the two commuting affine algebras are written as
$\hat{\mathfrak g}_k\oplus\hat{\mathfrak g}_{-k}$ with
$-\operatorname{tr}(t^at^b)=\delta^{ab}$,
\begin{equation}\label{eq:opposite-level-currents}
[J^a_m,J^b_n]=if^{abc}J^c_{m+n}+km\,\delta^{ab}\delta_{m+n,0}\,,\qquad [\tilde J^a_m,\tilde J^b_n]=if^{abc}\tilde J^c_{m+n}-km\,\delta^{ab}\delta_{m+n,0}\,,
\end{equation}
with $[J^a_m,\tilde J^b_n]=0$. Recombining into a diagonal sum and a rescaled antidiagonal difference,
\begin{equation}\label{eq:current-recombination}
\mathcal J^a_m:=J^a_m+\tilde J^a_m\,,\qquad
\mathcal P^a_m:=\frac{1}{2k}\big(J^a_m-\tilde J^a_m\big)\,,
\end{equation}
leads to
\begin{equation}\label{eq:contracted-brackets}
\begin{aligned}
{}[\mathcal J^a_m,\mathcal J^b_n]&=if^{abc}\mathcal J^c_{m+n}\,,\\
[\mathcal J^a_m,\mathcal P^b_n]&=if^{abc}\mathcal P^c_{m+n}+m\,\delta^{ab}\delta_{m+n,0}\,,\\ 
[\mathcal P^a_m,\mathcal P^b_n]&=\frac{1}{4k^2}\,if^{abc}\mathcal J^c_{m+n}\,.
\end{aligned}
\end{equation}
There are three features of~\eqref{eq:contracted-brackets} that are worth stating explicitly. The level cancels exactly between the two factors in the first bracket, so $\mathcal J$ generates a loop algebra at level zero. The central term survives only in the mixed bracket, where it is normalized to one at every $k$. Lastly, the power $k^{-1}$ is forced if the mixed central term is to have a finite nonzero limit while the third bracket vanishes; the constant $1/2$ in~\eqref{eq:current-recombination} normalizes that mixed term to one at every $k$. At $k\to\infty$ we obtain $\widehat{\mathfrak{tg}}$, and restricting to zero modes removes the central term and leaves $\mathfrak{tg}=\mathfrak g\ltimes\mathfrak g_\mathrm{ab}=\mathrm{Lie}(TG)$.

At the classical current-algebra level, the phase space of $\dim G$ free scalars
realizes the same brackets.  With fields $X^a(\sigma)$, conjugate momenta
$\Pi_a(\sigma)$ and
$\{X^a(\sigma),\Pi_b(\sigma')\}=\delta^a_b\,\delta(\sigma-\sigma')$, set
\begin{equation}\label{eq:free-scalar-currents}
J_a(\sigma)=f_{ab}{}^{c}X^b(\sigma)\Pi_c(\sigma)-\operatorname{tr}(t_at_b)\,\partial_\sigma X^b(\sigma)\,,
\end{equation}
and a direct Poisson-bracket calculation gives
\begin{equation}\label{eq:free-scalar-brackets}
\begin{aligned}
\{J_a(\sigma),J_b(\sigma')\}&=f_{ab}{}^{c}J_c(\sigma)\,\delta(\sigma-\sigma')\,,\\
\{J_a(\sigma),\Pi_b(\sigma')\}&=f_{ab}{}^{c}\Pi_c(\sigma)\,\delta(\sigma-\sigma')
-\operatorname{tr}(t_at_b)\,\partial_\sigma\delta(\sigma-\sigma')\,,\\
\{\Pi_a(\sigma),\Pi_b(\sigma')\}&=0\,,
\end{aligned}
\end{equation}
which is~\eqref{eq:contracted-brackets} at $k=\infty$ under $\mathcal J\leftrightarrow J$ and $\mathcal P\leftrightarrow\Pi$. The only step requiring care is the cross term, where the two orderings of the orbital piece against $\partial_\sigma X$ produce derivative-of-delta contributions proportional to $X^e(\sigma')-X^e(\sigma)$, and smearing converts these into $-f_{ab}{}^{c}\operatorname{tr}(t_ct_d)\partial_\sigma X^d$. The would-be central term in $\{J,J\}$ cancels for the same reason, surviving only in $\{J,\Pi\}$, exactly as in \eqref{eq:contracted-brackets}. Thus the classical free-scalar phase space has the correct contracted Poisson algebra and the required $2\dim G$ phase-space directions.

Before turning to the bulk we note a caveat about chirality, since the four checks above all concern the non-chiral boundary theory in which one WZW factor supplies the holomorphic sector and the other the antiholomorphic one, so that each chirality carries $\hat{\mathfrak g}_{\pm k}$ separately and the limiting central charge is $(\dim G,\dim G)$. A different and equally natural reading of $\DCS{k}{G}$ treats it as a Chern--Simons theory for $G\times G$ at a level of split signature, whose boundary carries $\hat{\mathfrak g}_k\oplus\hat{\mathfrak g}_{-k}$ within a \emph{single} chirality. Contracting that theory produces the WZW model on the cotangent group $T^*G=G\ltimes_{\mathrm{Ad}^*}\mathfrak g^*$, a non-reductive model of the kind introduced by Nappi and Witten and developed in~\cite{Nappi:1993ie,Olive:1993hk,Sfetsos:1993na,Figueroa-OFarrill:1995lva}. For the nondegenerate cotangent pairing its Sugawara central charge is exactly $c=\dim\mathfrak d=2\dim G$~\cite{Tseytlin:1995yw}, and it differs from $2\dim G$ free scalars. Its invariant pairing has split signature $(\dim G,\dim G)$~\cite{Nappi:1993ie,Figueroa-OFarrill:1995lva}. We mention this because the two readings are easy to confuse. They correspond to genuinely different boundary conditions on the same bulk, and only the non-chiral one is the subject of this subsection.

As a separate geometric consistency check, the tangent group also recovers the
line labels of the limiting bulk.  This explains why the boundary candidate and
the regular bulk endpoint use compatible variables; it is not an independent
proof of boundary-CFT convergence.  Put $BF$ theory on $\Sigma\times\mathbb R$
and split $S=-\int\operatorname{tr}(B\wedge F_A)$ into temporal and spatial parts.
The time components carry no derivatives and act as Lagrange multipliers,
$B_0$ imposing $F_A=0$ and $A_0$ imposing $D_AB=0$, while the spatial
components inherit the symplectic form
\begin{equation}\label{eq:BF-symplectic}
\omega=-\int_\Sigma \epsilon^{ij}\,\operatorname{tr}(\delta B_i\wedge\delta A_j)\,.
\end{equation}
The field $B$ is therefore the momentum conjugate to $A$, and on the smooth
regular locus the reduced phase space is the cotangent bundle
$T^*\mathcal M_\mathrm{flat}(\Sigma,G)$. For a surface with boundary, this statement requires boundary conditions and
a choice of which gauge transformations are quotiented.  A useful local
building block is an oriented edge with holonomy $g\in G$ and a momentum
$p\in\mathfrak g^*$, parallel-transported to one endpoint.  Its phase space is
$T^*G$, with canonical one-form $\theta=\langle p,g^{-1}dg\rangle$.
This edge phase space is not, by itself, the reduced phase space of an annulus.

Vertical polarization gives the edge Hilbert space
\begin{equation}\label{eq:L2G}
 \mathcal H_{\mathrm{edge}}=L^2(G,dg)
 \simeq\widehat{\bigoplus_{\lambda\in P_+}}V_\lambda\otimes V_\lambda^*.
\end{equation}
The last equality is the Peter--Weyl decomposition.  Multiplication by functions
of $g$ represents position observables, invariant derivatives represent
momenta.  It explains how compact holonomy variables can coexist with discrete
representation labels.

The algebra that keeps track of a holonomy and its conjugation covariance is
\begin{equation}\label{eq:crossed-product}
 D(G)=C(G)\rtimes_{\mathrm{Ad}}G.
\end{equation}
This is a crossed product, not the algebra of conjugation-invariant functions
obtained by imposing a Gauss constraint.  Its irreducible representations are
classified by induction~\cite{Koornwinder:1996uq}:
\begin{equation}\label{eq:simple-objects}
 ([g],R),\qquad R\in\mathrm{Irr}\,C_G(g).
\end{equation}
The corresponding representation space consists of $L^2$ sections of
$G\times_{C_G(g)}V_R\to[g]$.  At regular $g$, $C_G(g)=T$, and its characters
are labeled by $P$.  The regular sectors are thus labeled by $Y\times P$,
in agreement with the bulk construction.  Combinatorial quantization of
Chern--Simons theory for $G\ltimes\mathfrak g^*$ produces a quantum observable
algebra with a natural $D(G)$ symmetry~\cite{Meusburger:2003hc}; the observable
algebra of a general punctured surface should not be identified with a single
copy of $D(G)$.  For $G=SU(2)$ this quantum double also appears in Euclidean
three-dimensional gravity at vanishing cosmological constant~\cite{Bais:1998yn}.

The polarization and the choice of symmetry boundary are related pieces of the
quantization problem, but they are not interchangeable.  Peter--Weyl changes
basis in $L^2(G)$; a condensation boundary additionally specifies which bulk
lines may end.  The distinct abelian radii above illustrate why the latter
choice carries information beyond a change of basis.

\paragraph{Scope of the boundary analysis.}
The $U(1)$ result is exact at the wavefunction and momentum-lattice level for the
specified polarizations.  For general $G$, the central-charge limit, tangent-space
action, $TG$ symmetry, affine contraction, and classical Poisson realization are
mutually compatible checks.  We have not proved convergence of the complete
boundary Hilbert spaces, quantum current-algebra representations, local fields,
OPE or sewing coefficients, or correlators, and no statement here covers the
singular bulk strata.

\subsection{Boundary wavefunction for the non-abelian case}
\label{sec:nonabelian-BF-wavefunctions}

We then turn to the boundary wavefunctions for non-abelian BF theory and also double Chern-Simons theory. We will focus on the regular gauge fields on the boundary.

\paragraph{Non-compact non-abelian BF theory}

We now construct a regular-sector candidate for the nonabelian $BF$ boundary wavefunction.  Lattice gauge covariance fixes its Cartan theta-function part; the root determinant will be proposed separately and tested against the free-scalar theory.  We do not derive the full nonabelian path integral from this ansatz.  After a boundary counterterm symmetrizes its kinetic term, the action is
    \begin{equation}
        S=\frac{i}{2\pi} \int \operatorname{tr}  b \wedge f = \frac{i}{4\pi}\int \operatorname{tr}  b \wedge da + \frac{i}{4\pi} \int \operatorname{tr} a \wedge db - \frac{1}{2\pi}\int \operatorname{tr} b\wedge a \wedge a\,,
    \end{equation}
where $f=da+i a\wedge a$\footnote{Throughout this subsection, we use Hermitian gauge fields in parallel with the real-field convention
used in the abelian calculation.} and the kinetic terms of $a$ and $b$ are treated on equal footing. The variation of the action gives the boundary term up to equation of motion
    \begin{equation}
        \delta S = -\frac{i}{4\pi} \int_{\partial \mathcal{M}} \operatorname{tr} (b\wedge \delta a) - \frac{i}{4\pi} \int_{\partial \mathcal{M}} \operatorname{tr} (a\wedge \delta b) \,.
    \end{equation}
Following the abelian case, we similarly introduce the combinations
    \begin{equation}
        A_{\pm} = \frac{1}{\sqrt{2}}\left(\frac{a}{r}\pm rb \right)\,,
    \end{equation}
and define the holomorphic and anti-holomorphic components for both $A_{\pm}$ as before,
and redefine the action by
    \begin{equation}
        S' = S -\frac{1}{2\pi} \int_{\partial \mathcal{M}} d^2 z \, \operatorname{tr} \left( |A_+|^2 +  |A_-|^2\right)\,.
    \end{equation}
We also introduce the complex variables as
    \begin{equation}
        \xi = \frac{\tau_2}{i \pi} (A_+)_{\bar{z}}\,,\quad \bar{\xi}=-\frac{\tau_2}{i\pi}(A_-)_{z}\,.
    \end{equation}

Consider a topologically trivial flat background on the boundary torus with regular commuting holonomies.  A gauge transformation gives a constant Cartan representative for $a$, and we denote
    \begin{equation}
        a_{\bar{z}}|_{\partial \mathcal{M}}= \frac{i \pi}{\tau_2}u\,,\quad u \in \mathfrak{t}_{\mathbb{C}}\,.
    \end{equation}
When the root local systems have no zero modes, the shift symmetry $b\rightarrow b+D_a\beta$, together with the linearized flatness equation, removes the non-Cartan cohomology of $b$, so that
    \begin{equation}
        b_{\bar{z}}|_{\partial \mathcal{M}}= \frac{i\pi}{\tau_2}v\,,\quad v \in \mathfrak{t}_{\mathbb{C}}\,.
    \end{equation}
Then we have the parametrization
    \begin{equation}
        \xi = \frac{1}{\sqrt{2}}\left(\frac{u}{r}+vr \right)\,, \quad \bar{\xi} = \frac{1}{\sqrt{2}}\left(\frac{\bar{u}}{r}-\bar{v}r \right)\,,
    \end{equation}
The bar on $\bar\xi$ labels the second polarization coordinate; it does not assert $\bar\xi=\xi^*$.

After we fixed both $a_{\bar{z}}$ and $b_{\bar{z}}$ along the Cartan direction, there remain two types of residual gauge transformations preserving the gauge choice. The first one is the single-valued lattice gauge transformation $g_{m,n}=\exp [2\pi i (n\sigma^2-m\sigma^1)]$ with $n,m \in Q^{\vee} \subset \mathfrak{t}$, where $Q^{\vee}$ is the coroot lattice, which induces
    \begin{equation}
        \xi \rightarrow \xi + \delta\,, \quad \bar{\xi}\rightarrow \bar{\xi}+\delta^*\,,\quad \delta = \frac{1}{\sqrt{2}r}\left(n+m\tau \right)\,,\quad n,m\in Q^{\vee}\,,
    \end{equation}
and one has $\delta \bar{\xi} = \delta \xi^*$. The second is the Weyl group $W=N(T)/T$, where $N(T)$ is the normalizer of $T$, which acts on both $a$ and $b$ and we write
    \begin{equation}
        (\xi , \bar{\xi}) \rightarrow (w \xi ,w \bar{\xi})\,,\quad w\in W\,.
    \end{equation}
The candidate wavefunction must satisfy
    \begin{equation}
        \Psi(\xi + \delta, \bar{\xi}+\delta^*) = \Psi(\xi,\bar{\xi}) \exp \left[ \frac{\pi}{\tau_2} \left(\bigl(\xi , \delta^*\bigr) +\bigl(\bar{\xi},\delta\bigr) + \bigl(\delta,\delta^*\bigr) \right) \right]\,,
    \end{equation}
where $(x,y)=\operatorname{tr}(xy)$ is the geometric Cartan inner product in Eq.~\eqref{eq:geometric-Cartan-form}; all pairings of lattice coordinates, momenta, and roots below use that form. The abelian derivation goes through verbatim if we replace the products by $\bigl(\cdot ,\cdot\bigr)$ and the integer lattice $\mathbb{Z}$ by coroot lattice $Q^{\vee}$. Introduce $f=\Psi e^{-\frac{\pi}{2\tau_2} \left( \bigl(\xi,\xi\bigr) + \bigl(\bar{\xi},\bar{\xi}\bigr) \right)}$ and using $\delta-\delta^*=i \frac{\sqrt{2}}{r} m\tau_2$, we have
    \begin{equation}
    \begin{gathered}
        f\left(\xi + \frac{n}{\sqrt{2}r},\bar{\xi}+\frac{n}{\sqrt{2}r} \right) = f(\xi ,\bar{\xi})\,,\\
        f\left(\xi + \frac{m \tau}{\sqrt{2}r},\bar{\xi}+\frac{m\bar{\tau}}{\sqrt{2}r} \right) = f(\xi ,\bar{\xi}) e^{-\frac{2\pi i}{\sqrt{2}r}\bigl(m,\xi -\bar{\xi}\bigr) + \frac{\pi \tau_2}{r^2} \bigl(m,m\bigr)}\,,
    \end{gathered}
    \end{equation}
for all $n,m \in Q^{\vee}$. 

Consider the Fourier expansion
    \begin{equation}
        f = \sum_{p_L,p_R \in \mathfrak{t}} c_{p_L,p_R} e^{2\pi i \left( \bigl(p_L,\xi\bigr) -\bigl(p_R,\bar{\xi}\bigr) \right)}\,,
    \end{equation}
the first relation requires $\frac{1}{\sqrt{2}r}\bigl(p_L-p_R,n\bigr) \in \mathbb{Z}$ for all $n\in Q^{\vee}$, and hence
    \begin{equation}
        p_L - p_R = \sqrt{2}r \eta^{\sharp}\,, \quad \eta \in P\,,
    \end{equation}
where $P$ is the weight lattice. Here we write the metric identification $\mathfrak t^*\simeq\mathfrak t$
explicitly using the superscript $\sharp$, so that $\eta^{\sharp}\in \mathfrak{t}$ is the metric dual of $\eta\in \mathfrak{t}^*$ defined by $\bigl(\eta^{\sharp},\lambda\bigr) \equiv \eta(\lambda)$ for any $\lambda \in \mathfrak{t}$. The second relation gives the recursion relation
    \begin{equation}
        c_{p_L+\frac{m}{\sqrt{2}r},p_R+\frac{m}{\sqrt{2}r}}= c_{p_L,p_R} \exp\left[\frac{2\pi i}{\sqrt{2}r}\left(\tau \bigl(p_L,m\bigr)-\bar{\tau}\bigl(p_R,m\bigr) \right) - \frac{\pi \tau_2}{r^2} \bigl(m,m\bigr) \right]\,,
    \end{equation}
which is solved by the ansatz
    \begin{equation}
        c_{p_L,p_R}=\exp\left[ i \pi \tau \bigl(p_L,p_L\bigr) -i\pi \bar{\tau} \bigl(p_R , p_R\bigr) \right]\,.
    \end{equation}
Due to the constraint $p_L - p_R = \sqrt{2}r \eta^{\sharp}$ with $\eta \in P$, the independent wavefunctions are labeled by a discrete label $\eta \in P$ and also a continuous label $\varphi \in \mathfrak{t}/Q^{\vee}$, parametrized by
    \begin{equation}\label{eq:non-abelian-BF-momenta}
        p_L = \frac{1}{\sqrt{2}} \left(\frac{\mu + \varphi}{r}+r\eta^{\sharp} \right)\,,\quad p_R = \frac{1}{\sqrt{2}} \left(\frac{\mu + \varphi}{r}-r\eta^{\sharp} \right)\,,\quad \mu \in Q^{\vee}\,,
    \end{equation}
and we define the theta functions
    \begin{equation}\label{eq:non-abelian-BF-theta-function}
        \Theta_{\varphi,\eta}(\xi,\bar{\xi}|\tau) = \sum_{\mu \in Q^{\vee}} \exp\left[ i \pi \tau \bigl(p_L,p_L\bigr) -i\pi \bar{\tau} \bigl(p_R , p_R\bigr) + 2\pi i \left( \bigl(p_L,\xi\bigr) -\bigl(p_R,\bar{\xi}\bigr) \right) \right]\,.
    \end{equation}

Then let us consider the action of Weyl group $W$. The inner product $\bigl(\cdot,\cdot\bigr)$, the coroot lattice $Q^{\vee}$ and the weight lattice $P$ are all Weyl invariant, and one has
    \begin{equation}
        \Theta_{\varphi,\eta}(w\xi,w\bar{\xi}|\tau) = \Theta_{w^{-1}\varphi,w^{-1}\eta}(\xi,\bar{\xi}|\tau)\,,\quad w \in W\,.
    \end{equation}
Physical wavefunction must be Weyl invariant, therefore the labels are reduced to the diagonal Weyl orbits
    \begin{equation}
        [\varphi,\eta] \in (\mathfrak{t}/Q^{\vee} \times P)/W\,.
    \end{equation}
We can use the Weyl rotation to choose a unique representative $\varphi$ in the open fundamental alcove $Y$, the regular part of $(\mathfrak{t}/Q^{\vee})/W$ and leaves $\eta\in P$ unconstrained:
    \begin{equation}
        [\varphi,\eta]\in Y \times P\,,
    \end{equation}
which is precisely the regular part of the simple object set of the quantum double $D(G)$.

These covariance conditions motivate the following Weyl-invariant wavefunction ansatz over the regular sector:
    \begin{equation}
    \label{eq:nonabelian-BF-wavefunctions}
        \Psi_{[\varphi,\eta]}(\xi,\bar{\xi}|\tau) = \frac{1}{C(\xi,\bar{\xi}|\tau)} \frac{e^{\frac{\pi}{2\tau_2} \left( \bigl(\xi,\xi\bigr) + \bigl(\bar{\xi},\bar{\xi}\bigr) \right)}}{|\eta(\tau)|^{2\text{rk}(G)}}\sum_{w\in W}\Theta_{w\varphi,w\eta}(\xi,\bar{\xi}|\tau)\,,
    \end{equation}
where the Cartan oscillator factor is inherited from the abelian calculation.  The factor $C$ is a proposed contribution from the root directions.  The normalization of the Weyl orbit sum is fixed by the convention for $d\varphi$ below; an overall boundary-state normalization remains free.

First impose modular covariance, taking $C$ to be modular invariant.  Poisson summation pairs $Q^\vee$ with its dual character lattice $P$.  We keep the flat Cartan measure $d\varphi$ of section~\ref{sec:Limit_Regular_Simples}, so normalized torus Haar measure is $d\varphi/\operatorname{covol}(Q^\vee)$. The theta functions in~\eqref{eq:non-abelian-BF-theta-function} transform under the $T/S$-transformation as
    \begin{equation}
    \begin{gathered}
        \Theta_{\varphi,\eta}(\xi,\bar{\xi}|\tau+1)= e^{2\pi i \bigl(\varphi,\eta^{\sharp}\bigr)} \Theta_{\varphi,\eta}(\xi,\bar{\xi}|\tau)\,,\\
        \begin{aligned}
        &\Theta_{\varphi,\eta}\left(\frac{\xi}{\tau},\frac{\bar\xi}{\bar\tau}\,\middle|\,-\frac1\tau\right)\\
        &\quad=\frac{|\tau|^{\mathrm{rk}(G)}}{\operatorname{covol}(Q^\vee)}
        e^{\frac{\pi i}{\tau}\bigl(\xi,\xi\bigr)-\frac{\pi i}{\bar\tau}\bigl(\bar\xi,\bar\xi\bigr)}
        \int_{\mathfrak t/Q^\vee}d\varphi'\sum_{\eta'\in P}
        e^{-2\pi i(\bigl(\varphi,\eta'^{\sharp}\bigr)+\bigl(\varphi',\eta^{\sharp}\bigr))}
        \Theta_{\varphi',\eta'}(\xi,\bar\xi|\tau).
        \end{aligned}
    \end{gathered}
    \end{equation}
and the Gaussian factor $G(\xi ,\bar{\xi}|\tau) =e^{\frac{\pi}{2\tau_2} \left( \bigl(\xi,\xi\bigr) + \bigl(\bar{\xi},\bar{\xi}\bigr) \right)}$ transforms as
    \begin{equation}
    \begin{gathered}
        G(\xi ,\bar{\xi}|\tau+1)=G(\xi ,\bar{\xi}|\tau)\,,\\
        G\left(\frac{\xi}{\tau} ,\frac{\bar{\xi}}{\bar{\tau}}|-\frac{1}{\tau}\right) = G(\xi,\bar{\xi}|\tau) \exp \left[-\frac{\pi i}{\tau} \bigl(\xi , \xi\bigr) +
        \frac{i \pi}{\bar{\tau}} \bigl(\bar{\xi},\bar{\xi}\bigr)\right]\,,
    \end{gathered}
    \end{equation}
and the $\eta$-function
    \begin{equation}
        \eta(\tau+1)=e^{\frac{i\pi}{12}} \eta(\tau)\,,\quad \eta\left(-\frac{1}{\tau} \right) = \sqrt{-i \tau} \,\eta(\tau)\,.
    \end{equation}
The Gaussian multipliers cancel, giving the following $T$ and $S$ transformations of the ansatz:
    \begin{equation}
    \label{eq:nonabelian-wavefunction-modular}
    \begin{gathered}
        \Psi_{[\varphi,\eta]}(\xi,\bar{\xi}|\tau+1)=e^{2\pi i \bigl(\varphi,\eta^{\sharp}\bigr)}
        \Psi_{[\varphi,\eta]}(\xi,\bar{\xi}|\tau)\,,\\
        \begin{aligned}
        &\Psi_{[\varphi,\eta]}\left(\frac\xi\tau,\frac{\bar\xi}{\bar\tau}\,\middle|\,-\frac1\tau\right)\\
        &\quad=\frac{1}{\operatorname{covol}(Q^\vee)}\int_Yd\varphi'
        \sum_{\eta'\in P}\sum_{w\in W}
        e^{-2\pi i(\bigl(w\varphi,\eta'^{\sharp}\bigr)+\bigl(\varphi',w\eta^{\sharp}\bigr))}
        \Psi_{[\varphi',\eta']}(\xi,\bar\xi|\tau).
        \end{aligned}
    \end{gathered}
    \end{equation}
These are the regular-sector Fourier kernels of Theorem~\ref{thm:strong-modular}, using its measure convention.  The agreement checks the ansatz but does not prove that covariance determines the boundary theory uniquely.

We choose a candidate for $C$ using three conditions:
    \begin{itemize}
        \item $C(\xi,\bar{\xi}|\tau)$ is the generalization of the $\eta(\tau)$ coupling to the background fields $\xi$ and $\bar{\xi}$.
        \item $C(\xi,\bar{\xi}|\tau)$ can be read from the one-loop contribution of root parts of the parent Chern-Simons theory $G_{k}\times G_{-k}$ in the large $k$ limit.
        \item $C(\xi,\bar{\xi}|\tau)$ is modular invariant.
    \end{itemize}
Recall that the action of the non-abelian $BF$ theory can be written as large $k$ limit of the parent Chern-Simons theory $G_{k}\times G_{-k}$ with field redefinition
    \begin{equation}
        a= \frac{\mathcal{A}_++\mathcal{A}_-}{2}\,,\quad b=k(\mathcal{A}_+-\mathcal{A}_-)\,,\quad \mathcal{A}_{\pm}=a\pm \frac{b}{2k}\,,
    \end{equation}
where we use $\mathcal{A}_{\pm}$ for level-$k$ Chern-Simons theory to distinguish them from $A_{\pm}$. The paired Chern-Simons gauge fields and their holonomies on the boundary have the form
    \begin{equation}
        \mathcal{A}_{\pm}|_{\partial \mathcal{M}} = a|_{\partial \mathcal{M}}\pm\frac{1}{2 k}b|_{\partial \mathcal{M}}\quad \Rightarrow \quad u_{\pm} = u\pm \frac{v}{2k}\,,
    \end{equation}
and in the large $k$ limit, we have both $u_{\pm} \rightarrow u$. One then expects $C(\xi,\bar{\xi}|\tau)$ is the generalization of the $\eta$-functions coupling to the background $u$, and the oscillator contribution for a given root $\alpha$ is therefore proportional to
    \begin{equation}
        \prod_{n\geq 1} \frac{1}{|1-e^{2\pi i \alpha(u)}q^n|^2}\,,
    \end{equation}
and we propose the following expression of $C(\xi,\bar{\xi}|\tau)$
    \begin{equation}
    \label{eq:nonabelian-C-factor}
        C(u|\tau) = \exp\left[-\frac{2\pi h^{\vee}}{\tau_2} \bigl(\textrm{Im}(u),\textrm{Im}(u)\bigr) \right] \prod_{\alpha >0} \left| \frac{i \vartheta_1( \alpha(u)|\tau)}{\eta(\tau)}  \right|^2\,,
    \end{equation}
where the prefactor $\exp\left[-\frac{2\pi h^{\vee}}{\tau_2} \bigl(\textrm{Im}(u),\textrm{Im}(u)\bigr) \right]$ is inserted so that $C(\xi,\bar{\xi}|\tau)$ is a modular invariant.

Before moving on, we will comment on the meaning of $u$ and $v$. Use the same anti-Hermitian coordinate $X\in\mathfrak g$ as in Eq.~\eqref{eq:g-exp-scaling}, so that $g=\exp(X/\sqrt{k})$, and consider the left and right actions as
    \begin{equation}
        g \rightarrow h_L g h_R^{-1}\,, \quad h_L =e^{i \vartheta_L}\,,h_R = e^{i\vartheta_R}\,,\quad \vartheta_L,\vartheta_R \in \mathfrak{t}\,,
    \end{equation}
Here $\vartheta_L,\vartheta_R$ are angular parameters with periods $2\pi Q^\vee$; their normalized holonomy coordinates are $\vartheta_L/(2\pi),\vartheta_R/(2\pi)$. Then it is easy to see that, the translation $P$ along the Cartan direction corresponds to $\vartheta_L =-\vartheta_R$, and the rotation $R$ corresponds to $\vartheta_L=\vartheta_R$. And we can write
\begin{equation}
    g\rightarrow e^{\frac{i}{2} (\vartheta_L+\vartheta_R)} e^{\frac{i}{2} (\vartheta_L-\vartheta_R) } g e^{\frac{i}{2} (\vartheta_L-\vartheta_R)}e^{-\frac{i}{2} (\vartheta_L+\vartheta_R)}=R(\vartheta_L+\vartheta_R)\circ P(\vartheta_L-\vartheta_R)\circ g\,.
\end{equation}
A finite translation of $X$ requires $\vartheta_L-\vartheta_R=O(k^{-1/2})$, while the common rotation remains finite; the difference must therefore be rescaled before taking the limit. Recall that the left and right $U(1)^{\text{rk}G}$ backgrounds are given by $\xi(\tau)$ and $\bar{\xi}(\bar{\tau})$, then we can identify the background for translation and rotation separately as $(\xi - \bar{\xi}^*)/2$ and $(\xi + \bar{\xi}^*)/2$, which gives
    \begin{equation}
    \label{eq:rotation-translation-backgrounds}
        \frac{1}{2}(\xi - \bar{\xi}^*)= \frac{r}{\sqrt{2}}v\,,\quad \frac{1}{2}(\xi+\bar{\xi}^*)=\frac{1}{\sqrt{2}r}u\,.
    \end{equation}
Therefore $u,v$ can be understood as the background of the rotation and translation.

Within the ansatz, we next test a flux-condensing boundary by integrating the neutral flux sectors. Setting $\eta=0$ and $u=0$, we have
    \begin{equation}
        \lambda\equiv \lambda_L=\lambda_R = \frac{1}{\sqrt{2}} \left(\frac{\mu + \varphi}{r}\right)\,,\quad \mu \in Q^{\vee}\,,
    \end{equation}
and 
    \begin{equation}
        \xi = \frac{r}{\sqrt{2}} v\,,\quad \bar{\xi}=-\frac{r}{\sqrt{2}}\bar{v}\,,
    \end{equation}
so that
    \begin{equation}
        \Theta_{\varphi,\eta}(\xi,\bar{\xi}|\tau) = \sum_{\mu \in Q^{\vee}} \exp\left[ i \pi (\tau-\bar{\tau}) \bigl(\lambda,\lambda\bigr) + 2\pi i \bigl(\lambda,(\xi-\bar{\xi})\bigr) \right]\,,
    \end{equation}
and we have
    \begin{equation}
        \int_{\mathfrak{t}/Q^{\vee}} d\varphi \Theta_{\varphi,\eta}(\xi,\bar{\xi}|\tau) = \int_{\mathfrak{t}/Q^{\vee}} d\varphi \sum_{\mu \in Q^{\vee}} \exp\left[ i \pi (\tau-\bar{\tau}) \bigl(\lambda,\lambda\bigr) + 2\pi i \bigl(\lambda,(\xi-\bar{\xi})\bigr) \right]\,.
    \end{equation}
and we can write
    \begin{equation}
        \int_{\mathfrak{t}/Q^{\vee}} d\varphi \sum_{\mu \in Q^{\vee}} = \int_{\mathfrak{t}} d \varphi\,,
    \end{equation}
so that the integral becomes a Gauss integral
    \begin{equation}
        \int_{\mathfrak{t}} d \varphi \exp \left[ - \frac{\pi \tau_2}{r^2} \bigl(\varphi ,\varphi\bigr) + \frac{\sqrt{2}\pi i}{r} \bigl(\varphi , \xi-\bar{\xi}\bigr)  \right] = \left(\frac{2R^2}{\tau_2} \right)^{\frac{1}{2}\text{rk}G} \exp \left[ -\frac{\pi}{2\tau_2}\bigl(\xi -\bar{\xi},\xi-\bar{\xi}\bigr) \right]\,,
    \end{equation}
where $R=r/\sqrt{2}$. Combined with the prefactor $e^{\frac{\pi}{2\tau_2} \left( \bigl(\xi,\xi\bigr) + \bigl(\bar{\xi},\bar{\xi}\bigr) \right)}$, we obtain
    \begin{equation}
        e^{-\frac{\pi}{2\tau_2} (|\xi|^2+|\bar{\xi}|^2)} Z_{\mathcal{L}^{BF}_2}=\frac{1}{C_{\mathrm{reg}}(0|\tau)} \left(\frac{\sqrt{2}R}{|\eta(\tau)|^2 \sqrt{\tau_2}}\right)^{\text{rk}G}e^{-\frac{2\pi R^2}{\tau_2} \bigl(v_2+\tau v_1,v_2 + \bar{\tau}v_1\bigr)}\,,
    \end{equation}
where $v=v_2+\tau v_1$ is the background of the translation symmetry along the Cartan directions.
    
We then focus on the $u\rightarrow 0$ limit of $C(u|\tau)$. Recall that the $\vartheta_1$ function satisfies 
    \begin{equation}
        \vartheta_1(z|\tau) = 2\pi\eta(\tau)^3 z+O(z^3)\,,
    \end{equation}
 in the $z\rightarrow 0$ limit, therefore one has 
    \begin{equation}
        C(u|\tau)\sim \prod_{\alpha>0} 4|\eta(\tau)|^4 |\sin \pi \alpha(u)|^2\rightarrow 0\,.
    \end{equation}
To regularize that, notice that $u$ is the background of the rotation symmetry. Let $X_{\alpha}\in \mathbb{C}$ denote the complex coordinate associated with the pair $(\alpha,-\alpha)$, and the rotation through the angle $\theta=2\pi z$ acts by $g_{\theta} X_{\alpha}=e^{i\theta}X_{\alpha}$. Denote the momentum on the complex plane as $p^2$, and consider the two dimensional heat kernel
    \begin{equation}
        K_s(x,y) = \langle x|e^{-s p^2}|y\rangle = \frac{1}{4\pi s} \exp \left[-\frac{|x-y|^2}{4s} \right]\,,\quad s=2\pi \tau_2
    \end{equation}
and consider the partition function
    \begin{equation}
        \text{Tr}_{L^2(\mathbb{C})}\left(g_{\theta} e^{-s p^2} \right)=\int_{\mathbb{C}}d^2x K_{s}(x,g_{\theta}x)=\frac{1}{4\pi s} \int_{\mathbb{C}} d^2x \exp \left[-\frac{ |1-e^{i\theta}|^2}{4s}|x|^2 \right]\,.
    \end{equation}
For $\theta \neq 0$, it gives 
    \begin{equation}
    \label{eq:rotation-heat-kernel}
        \text{Tr}_{L^2(\mathbb{C})}\left(g_{\theta} e^{-s p^2} \right) = \frac{1}{|1-e^{i\theta}|^2} = \frac{1}{4\sin^2 \pi z}\,.
    \end{equation}
Combining all $X_{\alpha}$, they give the divergent part. On the other hand, if we set $\theta=0$, we have instead
    \begin{equation}
        \text{Tr}_{L^2(\mathbb{C})}\left(e^{-s p^2} \right) = \frac{\text{Vol}(\mathbb{C})}{4\pi s} = \frac{\text{Vol}(\mathbb{C})}{8\pi^2 \tau_2}\,.
    \end{equation}
Therefore at $u\rightarrow 0$, we replace each divergent untwisted zero-mode factor by its volume-regularized heat-kernel trace.  In this prescription
    \begin{equation}
    \label{eq:C-at-zero}
        C_{\mathrm{reg}}(0|\tau):=\prod_{\alpha>0} \frac{8\pi^2|\eta(\tau)|^4 \tau_2}{\text{Vol}(\mathbb{C}_{\alpha})}\,.
    \end{equation}
Finally, we have the partition function
    \begin{equation}
    \label{eq:nonabelian-ZL2-final}
        e^{-\frac{\pi}{2\tau_2} (|\xi|^2+|\bar{\xi}|^2)} Z_{\mathcal{L}^{BF}_2}=\prod_{\alpha>0}\left(\frac{\text{Vol}(\mathbb{C}_{\alpha})}{8\pi^2 |\eta(\tau)|^4 \tau_2}\right)\left(\frac{\sqrt{2}R}{|\eta(\tau)|^2 \sqrt{\tau_2}}\right)^{\text{rk}G}e^{-\frac{2\pi R^2}{\tau_2} \bigl(v_2+\tau v_1,v_2 + \bar{\tau}v_1\bigr)}\,.
    \end{equation}

\paragraph{Double non-abelian Chern-Simons theory}

Let us move to the non-abelian case. The action is
    \begin{equation}
        S=\frac{ik}{4\pi} \int \operatorname{tr}\left( a_+\wedge da_+ + \frac{2i}{3} a_+^3\right) -\frac{ik}{4\pi} \int \operatorname{tr}\left( a_- \wedge da_- + \frac{2i}{3}a_-^3\right)\,,
    \end{equation}
and we similarly introduce the BF variables
    \begin{equation}
        a=\frac{a_++a_-}{2}\,,\quad b= k(a_+-a_-)\,,
    \end{equation}
and the action becomes
    \begin{equation}
        S_k = \frac{i}{4\pi} \int_M \operatorname{tr}(a\wedge db + b\wedge da)-\frac{1}{2\pi} \int_M \operatorname{tr}(b\wedge a \wedge a)-\frac{1}{24\pi k^2} \int_M \operatorname{tr}b^3\,,
    \end{equation}
where $a,b$ are treated in a symmetric way, and the boundary variation is written as
    \begin{equation}
        \delta S_k = -\frac{i}{4\pi} \int_{\partial M} \operatorname{tr}(a\wedge \delta b)-\frac{i}{4\pi} \int_{\partial M} \operatorname{tr}(b\wedge \delta a)\,.
    \end{equation}
Notice that it is independent of the cubic terms and is the same as the variation term in the non-abelian BF theory. Therefore the analysis for the regular sector wavefunction is exactly parallel to the abelian case, and we will summarize the result in the following.

Put $\Psi=fe^{\frac{\pi}{2\tau_2}(\bigl(\xi,\xi\bigr) + \bigl(\bar{\xi},\bar{\xi}\bigr))}$, the independent solution for $f$ can be expanded via the $\Theta$-functions as
    \begin{equation}
        \Theta_{\lambda_+,\lambda_-}(\xi,\bar{\xi}|\tau) =  \sum_{n_+,n_-} \exp \left[i \pi \tau \bigl(p_L,p_L\bigr) - i \pi \bar{\tau} \bigl(p_R,p_R\bigr) +2\pi i (\bigl(p_L,\xi\bigr) - \bigl(p_R,\bar{\xi}\bigr))\right]\,,
    \end{equation}
where the momenta are
\begin{equation}
    p_L=\frac{1}{\sqrt{2}}\left[\frac{j_++j_-}{2kr}+r(j_+-j_-) \right]\,,\quad p_R=\frac{1}{\sqrt{2}}\left[\frac{j_++j_-}{2kr}-r(j_+-j_-) \right]\,,
\end{equation}
with
    \begin{equation}
        j_{\pm} = k m_{\pm} + \lambda^{\vee}_{\pm}\,,\quad m_{\pm} \in Q^{\vee}\,,\quad\lambda^{\vee}_{\pm} \in P^* / k Q^{\vee}\,,
    \end{equation}
where $P^*$ is the dual space of $P$ raised by the inner product $(\, ,\,)$ in Eq.~\eqref{eq:geometric-Cartan-form}, and it is not the same as the dual weight lattice $P^{\vee}$.

Introduce the BF labels
    \begin{equation}
        \varphi = \frac{\lambda^{\vee}_+ + \lambda_-^{\vee}}{2k}\,,\quad \eta =\lambda_+-\lambda_- \in P\,,
    \end{equation}
and write the momenta as
    \begin{equation}
        p_L = \frac{1}{\sqrt{2}} \left[\frac{\mu+w/2+\varphi}{r} + r(kw+\eta^{\vee}) \right]\,,\quad p_R = \frac{1}{\sqrt{2}} \left[\frac{\mu+w/2+\varphi}{r} - r(kw+\eta^{\vee}) \right]\,,
    \end{equation}
with
    \begin{equation}
        \mu=n_-\,,\quad w=n_+-n_-\,.
    \end{equation}
Since the wavefunction is exponentially suppressed by
    \begin{equation}
        \Theta \sim \sum_{n_+,n_-} \exp \left[-\pi \tau_2 \bigl(p_L,p_L\bigr) - \pi \tau_2 \bigl(p_R,p_R\bigr) \right]\,.
    \end{equation}
Keeping $\varphi$ and $\eta$ fixed and taking $k\rightarrow \infty$, the wavefunction is non-vanishing only when $w=0$. Relabel the $\Theta$-function using $\varphi,\eta$, we have in the large $k$ limit
    \begin{equation}
        \Theta_{\varphi,\eta}(\xi,\bar{\xi}|\tau) = \sum_{n \in Q^{\vee}} \exp \left[i \pi \tau \bigl(p_L,p_L\bigr) - i \pi \bar{\tau} \bigl(p_R,p_R\bigr) +2\pi i (\bigl(p_L,\xi\bigr) - \bigl(p_R,\bar{\xi}\bigr))\right]\,, 
    \end{equation}
where the momenta are
    \begin{equation}
        p_L = \frac{1}{\sqrt{2}} \left[\frac{n+\varphi}{r} + r\eta^{\vee} \right]\,,\quad p_R = \frac{1}{\sqrt{2}} \left[\frac{n+\varphi}{r} - r\eta^{\vee}\right]\,.
    \end{equation}
Compared to the non-abelian BF case in \eqref{eq:non-abelian-BF-momenta} and \eqref{eq:non-abelian-BF-theta-function}, we recover the $\Theta$-function given above as the classical contribution. Assuming that the one-loop parts are the same for the two kind of theory, we therefore expect their wavefunctions are the same.

As a comment, the comparison above concerns the classical contribution and hence uses the unshifted level $k$. The quantum formulas in the preceding sections involve $r_k=k+h^\vee$ and $\lambda_\pm+\rho$, but both prescriptions yield the same limiting flux–charge labels in the regular large-$k$ scaling regime.

\subsection{Flat gauging of the candidate boundary theory}
\label{sec:flat-gauging}

The free-scalar candidate of section~\ref{sec:G_WZW} admits the compact adjoint
symmetry $X\mapsto\mathrm{Ad}_gX$.  We examine its formal continuous orbifold by
$G$, following the flat-holonomy construction used in
\cite{Gaberdiel:2011aa}.  Flatness is an additional constraint on the gauge field,
which can be imposed by a Lagrange multiplier; simply dropping a Yang--Mills
kinetic term would not impose it.  The twisted free-field blocks can be computed
exactly.  Their integral over holonomies is infrared divergent, however, and a
complete gauged CFT requires a regulator, a prescription for the singular
sectors, and consistent correlation functions.  We therefore treat this
construction as a boundary proposal, not as an identification with the
Runkel--Watts theory.

\paragraph{Flat gauging on the torus.}
Coupling the scalars to a flat background computes partition functions
twisted by its holonomies, so the gauged partition function is the integral
of the twisted partition functions over the moduli of flat $G$-bundles on
$T^2$: commuting pairs of holonomies modulo conjugation.  For $G$ compact,
connected, simple, and simply connected every commuting pair lies in a common
maximal torus, and the residual identification is the
diagonal Weyl group, so the moduli space is
$\bigl((\mathfrak t/Q^\vee)\times(\mathfrak t/Q^\vee)\bigr)/W$.  This is not
a new arena: the pair $(u_1,u_2)$, with $e^{2\pi iu_1}$ the spatial and
$e^{2\pi iu_2}$ the temporal holonomy, is the rotation background
$u=u_2+\tau u_1$ of section~\ref{sec:nonabelian-BF-wavefunctions}, its
$Q^\vee$-periodicities being the lattice gauge transformations there and its
Weyl-reduced spatial factor being the alcove $Y$ that carries the flux label.
The formal expression is
\begin{equation}\label{eq:flat-gauging-def}
 Z_{\mathfrak g/G}(\tau)
 \;\stackrel{\mathrm{formal}}{=}\;\frac{\mathcal N}{|W|\operatorname{covol}(Q^\vee)}
 \int_{(\mathfrak t/Q^\vee)^2}du_1\,du_2\;
 Z_{\mathfrak g}[u_1,u_2](\tau)\,.
\end{equation}
Here $du_1$ and $du_2$ are flat Cartan measures; the displayed covolume makes the temporal-holonomy integral a normalized projection.  The regular-stratum measure is flat: reducing the class integral over
$G$ to $T$ costs the Weyl density
$\prod_{\alpha>0}|2\sin\pi\alpha(u_1)|^2
=\prod_{\alpha>0}|1-e^{2\pi i\alpha(u_1)}|^{2}$, which cancels against the
Jacobian $\prod_{\alpha>0}|1-e^{2\pi i\alpha(u_1)}|^{-2}$ from solving the
commuting constraint for the second holonomy, per root plane the same
determinant as the twisted heat-kernel
trace~\eqref{eq:rotation-heat-kernel}.  This regular-stratum measure is also preserved by the modular action
$(u_1,u_2)\mapsto(-u_2,u_1)$.  The constant $\mathcal N$ records a choice of
gauge-volume and continuum-state normalization.  Such a constant cannot remove
holonomy-dependent infrared singularities.

Although the scalar action factors through $G/Z(G)$, gauging the two groups is
not automatically the same operation.  The quotient group admits additional
bundle sectors and discrete theta-angle choices; gauging a trivially acting
finite center can lead to decomposition rather than a single universal
decoupled factor~\cite{Hellerman:2006zs}.  In this subsection we keep the simply
connected group $G$ and its commuting-holonomy sectors.  An effective adjoint-group
orbifold requires separate global data.

\paragraph{Twisted blocks from the boundary wavefunctions.}
Under $\mathrm{Ad}(T)$ the target splits into $\text{rk}G$ neutral Cartan
directions and one complex plane $\mathbb C_\alpha$ per positive root,
rotated by the angle $2\pi\alpha(u)$.  The Cartan scalars are untwisted and
neutral, so they contribute the factor
$\bigl(\sqrt2R/(|\eta(\tau)|^2\sqrt{\tau_2})\bigr)^{\text{rk}G}$
of~\eqref{eq:nonabelian-ZL2-final}, independently of $u$; this uses the same fixed Cartan zero-mode density normalization as in the abelian calculation.  Each root plane
contributes the twisted complex-boson block: writing $u=u_2+\tau u_1$,
$s=\alpha(u_1)\in(0,1)$ and $t=\alpha(u_2)$, the product representation of
$\vartheta_1$ gives
\begin{equation}\label{eq:twisted-block-expansion}
 e^{2\pi\tau_2s^2}
 \left|\frac{\eta(\tau)}{\vartheta_1\bigl((t+s\tau)\,|\,\tau\bigr)}\right|^{2}
 =(q\bar q)^{\frac{s(1-s)}{2}-\frac1{12}}
 \prod_{n\ge0}\bigl|1-q^{n+s}e^{2\pi it}\bigr|^{-2}
 \prod_{n\ge1}\bigl|1-q^{n-s}e^{-2\pi it}\bigr|^{-2}\,,
\end{equation}
the trace over a twisted Fock space with modings shifted by $\pm s$, the
temporal twist grading the states by their $\mathfrak t$-weight; the ground
state is spinless, $h=\bar h=\tfrac12s(1-s)$, because the twist is spatial
and shifts left and right movers symmetrically.  Multiplying over positive
roots and using $\sum_{\alpha>0}\alpha\otimes\alpha=h^\vee\bigl(\cdot,\cdot\bigr)$,
under which the prefactors assemble into that
of~\eqref{eq:nonabelian-C-factor},
\begin{equation}\label{eq:twisted-block}
 Z_{\mathfrak g}[u_1,u_2](\tau)
 \;=\;\left(\frac{\sqrt2R}{|\eta(\tau)|^{2}\sqrt{\tau_2}}\right)^{\text{rk}G}
 \frac{1}{C(u_2+\tau u_1\,|\,\tau)}\,,
\end{equation}
with $C$ exactly as proposed in~\eqref{eq:nonabelian-C-factor}.  The free-field calculation reproduces the proposed factor $C$ and its modular
completion.  It is an independent consistency check on the ansatz, not a
derivation of the full $BF$ path integral or of the parent WZW limit.  There are
two useful checks on its normalization.
First, at the untwisted point the regularization~\eqref{eq:C-at-zero}
reduces the integrand to the flux-condensed partition
function~\eqref{eq:nonabelian-ZL2-final} at vanishing background,
\begin{equation}\label{eq:untwisted-point-check}
 Z_{\mathfrak g}[0,0]
 =\left(\frac{\sqrt2R}{|\eta(\tau)|^{2}\sqrt{\tau_2}}\right)^{\text{rk}G}
 \prod_{\alpha>0}\frac{\text{Vol}(\mathbb C_\alpha)}{8\pi^2|\eta(\tau)|^4\tau_2}
 \;=\;Z_{\mathcal L^{BF}_2}\Big|_{u=v=0}\,,
\end{equation}
in the heat-kernel normalization used there.  The zero modes disappear for
nontrivial root twists, but their singularities persist in neighborhoods of the
untwisted locus.  Indeed, if $z=\alpha(u_2+\tau u_1)$ approaches zero, then
\begin{equation}\label{eq:orbifold-root-divergence}
 \left|\frac{\eta(\tau)}{\vartheta_1(z|\tau)}\right|^2
 \sim \frac{1}{4\pi^2|\eta(\tau)|^4|z|^2}.
\end{equation}
Since $d\alpha(u_1)\,d\alpha(u_2)$ is proportional to $d^2z$, this gives a
logarithmic divergence $\int d^2z/|z|^2$ already for $SU(2)$.  Intersections of
root loci can be more singular.  Removing a measure-zero set does not make the
integral finite.

Second, each twisted block transforms covariantly under the modular group, and
the formal measure is invariant.  To turn this into a statement about
$Z_{\mathfrak g/G}$ one must use a regulator that transforms with the holonomies
and control its removal; a cutoff imposed only on the spatial alcove need not
preserve $S$.  In the SymTFT description one would expect gauging to change the
symmetry boundary, but deriving the charge-condensing boundary requires data at
identity flux, which lie outside the regular ansatz.  We do not infer that
boundary from the regular formula alone; see
\cite{Jia:2026vcr,Jia:2026bbo} for continuous symmetry
categories and their gauging.

\paragraph{The regular twisted sectors and their ground states.}
On Weyl-invariant integrands
$\frac1{|W|}\int_{(\mathfrak t/Q^\vee)^2}du_1du_2
=\int_Yd\varphi\int_{\mathfrak t/Q^\vee}du_2$ up to the measure-zero
singular twists, and we rename the spatial twist $\varphi\in Y$ to match the
bulk flux label.  The $u_2$-integral projects each twisted sector onto the
states of vanishing $\mathfrak t$-weight, the invariants of the residual
centralizer $T$ of the regular twist, and this projection is level matching:
from the charge assignments in~\eqref{eq:twisted-block-expansion}, a state
of total weight $\lambda$ has
\begin{equation}\label{eq:level-matching}
 h-\bar h\;\equiv\;\lambda(\varphi)\ \ \mathrm{mod}\ 1\,,
\end{equation}
so the retained states carry integer spin, while the continuous-spin states,
such as $a^{(\alpha)}_{-s_\alpha}|0\rangle_{\mathrm{tw}}$ with
$h-\bar h=s_\alpha$, are removed.  Carrying out the projection
in~\eqref{eq:twisted-block-expansion} gives
\begin{equation}\label{eq:flat-gauged-Z}
 Z_{\mathfrak g/G}(\tau)
 \stackrel{\mathrm{formal}}{=}\mathcal N\left(\frac{\sqrt2R}{|\eta(\tau)|^{2}\sqrt{\tau_2}}\right)^{\text{rk}G}
 \int_Yd\varphi\;
 (q\bar q)^{\,h_{\mathrm{tw}}(\varphi)-\frac{\dim G-\text{rk}G}{24}}\,
 \mathcal F_\varphi(q,\bar q)\,,
\end{equation}
with $s_\alpha:=\alpha(\varphi)\in(0,1)$ for $\varphi\in Y$ and
\begin{equation}\label{eq:projected-tower}
\begin{split}
 \mathcal F_\varphi(q,\bar q)
 &=\frac{1}{\operatorname{covol}(Q^\vee)}\int_{\mathfrak t/Q^\vee}\!\!du_2
 \prod_{\alpha>0}\prod_{n\ge0}\bigl|1-q^{n+s_\alpha}e^{2\pi i\alpha(u_2)}\bigr|^{-2}
 \prod_{n\ge1}\bigl|1-q^{n-s_\alpha}e^{-2\pi i\alpha(u_2)}\bigr|^{-2}\\
 &=1+\sum_{\alpha>0}\Bigl[(q\bar q)^{s_\alpha}+(q\bar q)^{1-s_\alpha}\Bigr]
 +\tfrac{\dim G-\text{rk}G}{2}\,(q+\bar q)+\cdots.
\end{split}
\end{equation}
The displayed $(q+\bar q)$ contribution includes, for each positive root, the holomorphic
invariant $a^{(\alpha)}_{-s_\alpha}a^{(-\alpha)}_{-(1-s_\alpha)}$ at integer
level.  In higher rank, other zero-weight root combinations can contribute at the same order; the displayed series lists selected terms.  These excitations are compatible with the invariant chiral algebra: the
$G$-invariant subalgebra of the $\mathfrak g$-valued Heisenberg algebra
retains no spin-one currents but contains the Casimir currents of spins
$e_i+1$, where $e_i$ are the exponents of $\mathfrak g$, at central charge $\dim G$.

The twisted ground-state weight has the closed form
\begin{equation}\label{eq:twisted-weight}
 h_{\mathrm{tw}}(\varphi)=\bar h_{\mathrm{tw}}(\varphi)
 =\frac12\sum_{\alpha>0}\alpha(\varphi)\bigl(1-\alpha(\varphi)\bigr)
 =\rho(\varphi)-\frac{h^\vee}{2}\,\bigl(\varphi,\varphi\bigr)\,,
\end{equation}
strictly concave on the closed alcove, positive on $Y$, and vanishing
precisely at the vertices at which $e^{2\pi i\varphi}$ is central, where the
adjoint twist trivializes.  Its unique maximum sits at the interior point
$\varphi_*$ determined by $h^\vee\bigl(\varphi_*,x\bigr)=\rho(x)$ for all
$x\in\mathfrak t$, where the Freudenthal--de Vries strange formula
$\bigl(\rho,\rho\bigr)=h^\vee\dim G/12$ gives
\begin{equation}\label{eq:twisted-weight-bound}
 0<h_{\mathrm{tw}}(\varphi)\;\le\;
 h_{\mathrm{tw}}(\varphi_*)=\frac{\bigl(\rho,\rho\bigr)}{2h^\vee}
 =\frac{\dim G}{24}\,;
\end{equation}
one checks $\theta(\varphi_*)=(h^\vee-1)/h^\vee<1$ for the highest root
$\theta$, so $\varphi_*\in Y$.  The total ground-state exponent
in~\eqref{eq:flat-gauged-Z} is $h_{\mathrm{tw}}(\varphi)-\dim G/24$,
consistent with the limiting central charge~\eqref{eq:sugawara-c}, and it
vanishes at $\varphi_*$, the nonabelian analogue of the antiperiodic point.
Restoring the Cartan momentum continuum from the factor
$\tau_2^{-\text{rk}G/2}\propto\int_{\mathfrak t^*}dp\,(q\bar q)^{\frac12\bigl(p,p\bigr)}$,
the oscillator ground-state families have weights
\begin{equation}\label{eq:flat-gauged-spectrum}
 h=\bar h=h_{\mathrm{tw}}(\varphi)+\tfrac12\bigl(p,p\bigr)\,,\qquad
 (\varphi,p)\in Y\times\mathfrak t^*\,,
\end{equation}
in the normalization fixed by~\eqref{eq:WZW-expansion}; regular alcove points have trivial Weyl stabilizer, while their centralizer in $G$ is $T$.  The latter acts trivially on the Cartan momenta $p$.  Excited twisted sectors can contain further primary representations; the ground-state formula does not classify them.

The continuous twist label resembles the continuous-orbifold description of
Runkel--Watts, but the two spectra should not be identified.  Here the central
charge is $\dim G$, and the twist energy is bounded by $\dim G/24$; an unbounded
Cartan momentum continuum is present as well.  Projection ensures integer
conformal spin by~\eqref{eq:level-matching}, not spin zero for every primary.
For $G=SU(2)$ the alcove is an
interval, $s=\alpha(\varphi)\in(0,1)$, and
$h_{\mathrm{tw}}=\tfrac12s(1-s)\in(0,\tfrac18]$ with
$\tfrac18=\tfrac3{24}=\dim SU(2)/24$ attained at $s=\tfrac12$.  Before
projection, the $\varphi$-twisted sector is graded by the
$\mathfrak t$-weight,
$\mathcal H_\varphi=\bigoplus_{\eta}\mathcal H_{\varphi,\eta}$, the
weight-$\eta$ states being candidate junction sectors of the bulk line
$W_{[\varphi,\eta]}$, in parallel with the fixed-charge sectors of
section~\ref{sec:runkel-watts}: the defect sectors of the flat-gauged
boundary are organized by the regular flux--charge label
set~\eqref{eq:introduction-flux-charge-revised} underlying
Theorem~\ref{thm:strong-modular}, realized for $\eta$ in the root lattice
$Q\subseteq P$, with $P/Q\cong\widehat{Z(G)}$ labeling center sectors
invisible to the adjoint fields (cf.\ section~\ref{sec:G_WZW}).  The gauged
theory also carries the dual noninvertible lines $\mathrm{Rep}(G)$, the
Wilson lines of the flat gauge field, acting on the twisted continuum by
$\chi_R\bigl(e^{2\pi i\varphi}\bigr)$; characters distinguish the regular holonomy sectors.  This diagonal action alone does not prove that these sectors are superselection sectors of the full local operator algebra.

\paragraph{What remains to define the orbifold.}
The regular blocks contain continuous ground-state labels $(\varphi,p)$ and
therefore suggest a nonrational theory, once a Hilbert space and a renormalized
spectral measure have been constructed.  In the free noncompact theory the
constant zero-momentum wavefunction is not normalizable.  This observation,
and the continuous spectrum of the regular sectors, do not determine how an
identity operator or additional singular sectors should enter a completed
orbifold.

The first problem is to regulate~\eqref{eq:orbifold-root-divergence} consistently
with modular transformations and factorization.  The second is to include twists
with larger centralizers, whose invariant-state projection differs from the
regular $T$ projection.  The third is to determine the operator products of the
resulting twist fields.  Continuous orbifolds and the Runkel--Watts construction
provide useful comparisons~\cite{Gaberdiel:2011aa,Runkel:2001ng}, but neither
supplies these data for adjoint-valued noncompact scalars.  The present
calculation fixes candidate regular-sector building blocks and identifies these
specific obstructions to completing the boundary theory.

% section_07.tex
\section[The large-k limit of interface theories]
{The large-$k$ limit of interface theories}
\label{sec:interfaces}

So far we hold the level fixed inside each comparison and treat the family
\begin{equation}\label{eq:doubled-family}
\mathsf T_k\;:=\;\DCS{k}{G}\,,\qquad k\in\mathbb Z_{>0}\,,
\end{equation}
as an externally parametrized sequence. A physical realization of the limit
should place different members of \eqref{eq:doubled-family} in one spacetime and
specify what joins them.  This section is an extension of, rather than an input
to, the bulk convergence theorems.  Subsection~\ref{sec:interface-chain}
introduces a finite-weight refinement criterion for this problem.
Proposition~\ref{prop:adjacent-wall-nonrefinement} then proves an exact
obstruction for $SU(2)$, an adjacent level jump, and nonnegative matrices
obeying the necessary genus-one $S$- and $\theta$-conditions.  The coset
branching rules and weight--displacement identities are exact properties of
the proposed gapless chiral algebra, whereas construction of the physical
doubled bulk-coupled interface and its Runkel--Watts limit remain conjectural.

\subsection{The interface chain and refinement}\label{sec:interface-chain}

The category of the doubled theory is the Drinfeld center
$Z(\mathcal C_k)\simeq\mathcal C_k\boxtimes\mathcal C_k^{\mathrm{rev}}$.
Divide a three-manifold $M$ into parallel slabs $R_K,\dots,R_{K+n}$ meeting
along two-dimensional surfaces
$R_{K+j}\cap R_{K+j+1}=\Sigma_{K+j+1/2}$, and place
$\mathsf T_{K+j}$ in $R_{K+j}$.  This gives the chain
\begin{equation}\label{eq:interface-chain}
 \mathsf T_K\ \big|_{\Sigma_{K+1/2}}\ \mathsf T_{K+1}\ \big|_{\Sigma_{K+3/2}}\ \cdots\ \big|_{\Sigma_{K+n-1/2}}\ \mathsf T_{K+n}\,,\qquad K\geq1\,,
\end{equation}
which is illustrated in figure~\ref{fig:interface-chain}.

\begin{figure}[htbp]
\centering
\begin{tikzpicture}[x=1cm,y=1cm,>=Stealth,line width=0.9pt,font=\small]
 \def\H{2.9}
 \fill[black!3]  (0,0)   rectangle (2.7,\H);
 \fill[black!8]  (2.7,0) rectangle (5.4,\H);
 \fill[black!13] (5.4,0) rectangle (7.7,\H);
 \fill[black!18] (7.7,0) rectangle (9.0,\H);
 \fill[black!24] (9.0,0) rectangle (11.7,\H);
 \draw[black!40] (0,0) rectangle (11.7,\H);
 \foreach \x in {2.7,5.4,9.0} \draw[line width=1.4pt] (\x,0) -- (\x,\H);
 \draw[black!40,dashed] (7.7,0) -- (7.7,\H);
 \node at (1.35,2.35) {$\mathsf T_K$};
 \node at (4.05,2.35) {$\mathsf T_{K+1}$};
 \node at (6.55,2.35) {$\mathsf T_{K+2}$};
 \node at (8.35,2.35) {$\cdots$};
 \node at (10.35,2.35) {$\mathsf T_{K+n}$};
 \node[anchor=south,font=\footnotesize] at (2.7,\H+0.04) {$\Sigma_{K+1/2}$};
 \node[anchor=south,font=\footnotesize] at (5.4,\H+0.04) {$\Sigma_{K+3/2}$};
 \node[anchor=south,font=\footnotesize] at (9.0,\H+0.04) {$\Sigma_{K+n-1/2}$};
 \node[anchor=north,font=\footnotesize] at (2.7,-0.06) {$\mathfrak J_K$};
 \node[anchor=north,font=\footnotesize] at (5.4,-0.06) {$\mathfrak J_{K+1}$};
 \node[anchor=north,font=\footnotesize] at (9.0,-0.06) {$\mathfrak J_{K+n-1}$};
 \draw[->,line width=1.1pt] (0.3,1.35) -- (11.4,1.35);
 \foreach \x in {2.7,5.4,9.0} \fill (\x,1.35) circle (2.1pt);
 \node[anchor=south,font=\footnotesize] at (1.5,1.45) {$J$};
 \node[anchor=south,font=\footnotesize] at (4.0,1.45) {$I$};
 \node[anchor=north,font=\footnotesize] at (5.85,1.02)
  {$(y_J,\omega_J)\ \longmapsto\ (x_I,\eta_I)$};
\end{tikzpicture}
\caption{The chain \eqref{eq:interface-chain}. Each slab carries one member of the family \eqref{eq:doubled-family}, and each thick vertical line carries an interface theory $\mathfrak J_k$ coupled to the two slabs it separates. The horizontal arrow follows a line of $\mathsf T_K$ across the chain. At every fixed conformal-weight cutoff, the refinement requirement \eqref{eq:refining-wall} asks that such crossings preserve the charge $\eta$ and move the flux $x$ by a bounded multiple of the mesh $(k+h^\vee)^{-1}$, with a bound independent of $k$.}
\label{fig:interface-chain}
\end{figure}
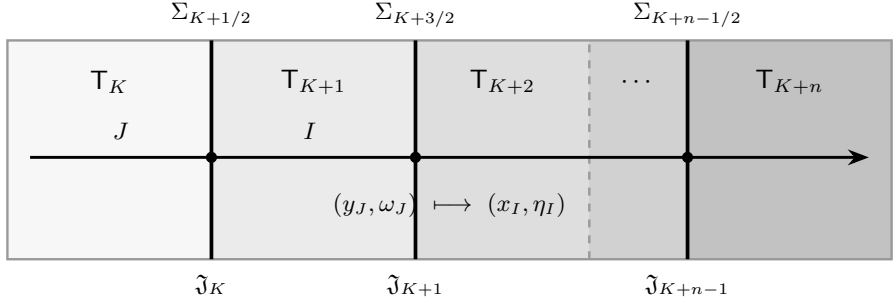

We seek a two-dimensional theory $\mathfrak J_k$ on $\Sigma_{k+1/2}$ coupled to both adjacent bulks. The coupling must specify how a line of $\mathsf T_k$ continues as a line of $\mathsf T_{k+1}$ and cancel the anomaly localized by the level jump. An intrinsic two-dimensional theory without these bulk couplings is therefore not the complete interface datum, and the detailed construction of the doubled bulk-coupled interfaces $\mathfrak J_k$ remains open. For a gapped interface, only topological transmission data survive in the infrared, described categorically by a bimodule over the two line categories~\cite{Kitaev:2011dxc,Fuchs:2012dt,Kapustin:2010if,Galviz:2026xnp}. At a gapless conformal fixed point, the interface also carries local operators with nonzero conformal weights~\cite{Bachas:2001vj}.

We use the regular coordinates in \eqref{eq:regular-lattice}, with
$a=\lambda_+$ and $b=\lambda_-$.  We write
$I\in\Irr Z(\mathcal C_{k'})$ with coordinates $(x_I,\eta_I)$ and
$J\in\Irr Z(\mathcal C_k)$ with coordinates $(y_J,\omega_J)$, keeping $I$
in the outgoing role and $J$ in the incoming role throughout.  For an
interface $\mathfrak J_k$ between neighboring levels, let
$\mathcal H^{\mathfrak J_k}_{IJ}$ denote the space of local interface
operators at which an incoming line $J$ continues as an outgoing line $I$.
When $\mathfrak J_k$ is conformal, this space is graded by the scaling
dimension $\Delta=h+\bar h$, and we write
$\mathcal H^{\mathfrak J_k}_{IJ}[\Delta\leq\Delta_0]$ for its subspace of
states of dimension at most $\Delta_0$.

\begin{definition}[Finite-weight refinement]
For the convergence problem considered here, we call the neighboring interfaces
$\mathfrak J_k$, $k\geq1$, \emph{uniformly refining at finite conformal weight}
if, for every fixed $0\leq\Delta_0<\infty$, there is a constant
$C_{\Delta_0}$, independent of $k$, such that
\begin{equation}\label{eq:refining-wall}
\mathcal H^{\mathfrak J_k}_{IJ}[\Delta\leq\Delta_0]\neq0
\quad\Longrightarrow\quad
\eta_I=\omega_J
\ \text{ and }\
|x_I-y_J|\leq\frac{C_{\Delta_0}}{k+h^\vee}
\end{equation}
holds for every $k\geq1$. Here $\mathcal H^{\mathfrak J_k}_{IJ}[\Delta\leq\Delta_0]\neq0$ means that a local interface operator of total scaling dimension at most $\Delta_0$ joins the incoming line $J$ to the outgoing line $I$. The pairs $(y_J,\omega_J)$ and $(x_I,\eta_I)$ are their respective flux--charge coordinates, $h^\vee$ is the dual Coxeter number of $G$, and $C_{\Delta_0}$ may depend on the fixed cutoff $\Delta_0$ but not on $k$. The scale $(k+h^\vee)^{-1}$ is the mesh of the flux coordinate in~\eqref{eq:regular-lattice}, since moving an integrable weight by one lattice step moves $x$ by $O((k+h^\vee)^{-1})$. Thus, at every fixed conformal-weight cutoff,~\eqref{eq:refining-wall} allows a number of lattice steps bounded independently of $k$, while the charges agree exactly.  This criterion models identity-like transport in the regular $BF[G]$ limit; it is not a universal consistency condition on interfaces.
\end{definition}

For a gapped interface, the surviving junction spaces are topological and ungraded. Regarding them as concentrated in degree zero and writing their multiplicities as a tunneling matrix $\mathsf W^{k\to k+1}$, condition~\eqref{eq:refining-wall} specializes to
\begin{equation}
 \bigl(\mathsf W^{k\to k+1}\bigr)_{IJ}\neq0
 \quad\Longrightarrow\quad
 \eta_I=\omega_J,
 \qquad
 |x_I-y_J|\leq\frac{C_0}{k+h^\vee}.
\end{equation}
In the next subsection we test this gapped specialization and relate the general finite-weight condition to the conformal spectrum of the unit-jump coset candidate.

\subsection{Gapped obstruction and the gapless coset candidate}\label{sec:gapped-walls}

We note that gapped interfaces exist in this setting. Every Drinfeld center has trivial Witt class~\cite{Davydov:2010kfz}, so $[Z(\mathcal C_k)]=[Z(\mathcal C_{k+1})]=0$ and the Witt obstruction to a topological surface defect between the two doubled phases vanishes~\cite{Fuchs:2012dt}. Placing a gapped boundary of each phase back to back produces one, with a transport map that factorizes as $\mathsf W=uv^{\mathsf T}$, so that the admissible outputs depend on $u$ alone and the admissible inputs on $v$ alone. Such an interface is reflecting, namely for such $\mathsf W$, the allowed outgoing labels are independent of the incoming label, so this interface leaves the near-diagonal transport required by~\eqref{eq:refining-wall} unrealized. More general gapped interfaces can have nonfactorized tunneling matrices. We therefore test arbitrary nonnegative tunneling matrices against the necessary genus-one consistency conditions. To be more precise, recall that $\mathsf W_{IJ}$ counts the multiplicity with which an incoming line $J$ from the level-$k$ side crosses as an outgoing line $I$ on the level-$k'$ side. Compatibility with the two generators of the torus mapping-class-group action requires
\begin{equation}\label{eq:tunneling-conditions}
 \mathcal S_{k'}\mathsf W=\mathsf W\mathcal S_k\,,
 \qquad
 \theta^{(k')}_I\mathsf W_{IJ}=\theta^{(k)}_J\mathsf W_{IJ}\,,
 \qquad
 \mathsf W_{IJ}\in\mathbb Z_{\geq0}\,,
\end{equation}
tested on $T^2\times I$ with an interface at $t_0\in I$, where $\mathcal S$ and $\theta$ are the modular $S$-matrix and the twists of the two doubled phases~\cite{Lan:2014uaa,Kitaev:2011dxc}. These conditions are necessary, so they serve as an obstruction test, while lifting a given intertwiner to a bimodule category remains a further question. For $G=SU(2)$ and neighboring levels they are restrictive in the following sense.

For this $SU(2)$ test we write the alcove coordinate as $x=t\varpi_1$
and use $d(x,y)=|t(x)-t(y)|$ in the numerical bounds below.  If instead the basic
invariant norm of section~\ref{sec:Limit_Regular_Simples} is used, every distance
bound is multiplied by $\|\varpi_1\|=1/\sqrt2$; the obstruction to uniform
refinement is unchanged.

The physical tunneling condition in~\eqref{eq:tunneling-conditions} requires nonnegative
integer entries.  The following stronger obstruction uses only nonnegative real entries
and therefore applies a fortiori to physical tunneling matrices.
\begin{proposition}[$SU(2)$ obstruction to uniformly refining gapped transport]\label{prop:adjacent-wall-nonrefinement}
Let $G=SU(2)$, let $k\geq1$, and let $\mathsf W\neq0$ be a nonnegative real matrix obeying the two linear conditions in~\eqref{eq:tunneling-conditions} with $k'=k+1$. Then $\mathsf W$ carries a nonzero entry $\mathsf W_{IJ}$ with $\eta_I\neq\omega_J$, or a nonzero charge-preserving entry with
\begin{equation}\label{eq:wall-spread}
 \eta_I=\omega_J
 \qquad\text{and}\qquad
 |x_I-y_J|\geq\frac12-\frac1{k+3}\,.
\end{equation}
Consequently, no sequence of nonzero matrices satisfying the necessary genus-one conditions can be the tunneling data of a uniformly refining gapped interface in the sense of~\eqref{eq:refining-wall}.
\end{proposition}

\begin{proof}
Label the simple objects of $Z(\mathcal C_k)$ by pairs $(a,b)$ with $0\leq a,b\leq k$, and put $p=k+2$. We identify the $SU(2)$ alcove with $[0,1]$ by writing $x=t\varpi_1$ where $\varpi_1$ is the fundamental weight of $SU(2)$ and use the scalar coordinate $t$ in the absolute value in~\eqref{eq:refining-wall}, namely we measure distance in the alcove by $d(x,y)=|t(x)-t(y)|$. In these labels the chiral and doubled $S$-matrices and the doubled twist are
\begin{align}
 S^{(k)}_{ac}
 &=\sqrt{\frac2p}\,
   \sin\!\left(\frac{\pi(a+1)(c+1)}p\right),
 &
 \mathcal S_{k,(a,b),(c,d)}
 &=S^{(k)}_{ac}S^{(k)}_{bd},
 \label{eq:su2-doubled-S-for-wall}\\
 \theta^{(k)}_{(a,b)}
 &=\exp\!\left[
   2\pi i\frac{a(a+2)-b(b+2)}{4p}
   \right].
 \label{eq:su2-doubled-twist-for-wall}
\end{align}
Suppose first that every nonzero entry of $\mathsf W$ preserves the charge. Consider the incoming level-$k$ label $J_\ast=(1,0)$. Every charge-preserving outgoing label has the form $I_q=(q+1,q)$, where $0\leq q\leq k$. For $\mathsf W_{I_q,J_*}\neq 0$, the $\theta$-equation of~\eqref{eq:tunneling-conditions} would require
\begin{equation}\label{eq:charge-one-twist-obstruction}
 \frac{2q+3}{4(p+1)}\equiv\frac{3}{4p}\pmod 1.
\end{equation}
Both representatives in~\eqref{eq:charge-one-twist-obstruction} lie strictly between $0$ and $1/2$, so the congruence is an equality. It would give $2pq=3$, which is impossible. The entire $J_\ast$-column of $\mathsf W$ therefore vanishes.

The row of the outgoing vacuum $I_{\rm vac}=(0,0)$ must be nonzero. Indeed, if it vanished, then the $I_{\rm vac}$-row of $\mathcal S_{k+1}\mathsf W=\mathsf W\mathcal S_k$ would give, for every incoming label $J$,
\begin{equation}
 0=\sum_I \mathcal S_{k+1,I_{\rm vac}I}\mathsf W_{IJ}.
\end{equation}
The vacuum row of the doubled $SU(2)$ matrix is strictly positive, and $\mathsf W$ is nonnegative, so this would force $\mathsf W=0$. Charge preservation now allows us to write the nonzero vacuum row as
\begin{equation}
 u_a:=\mathsf W_{(0,0),(a,a)}\geq0,
 \qquad 0\leq a\leq k,
\end{equation}
with not all $u_a$ zero. Since the $J_\ast$-column of $\mathsf W$ vanishes,
$(\mathcal S_{k+1}\mathsf W)_{(0,0),J_\ast}=0$. The $S$-intertwining
equation, charge preservation in the outgoing-vacuum row, and
\eqref{eq:su2-doubled-S-for-wall} therefore give
\begin{equation}
\begin{split}
	0 &=(\mathcal S_{k+1}\mathsf W)_{(0,0),J_\ast}
  =(\mathsf W\mathcal S_k)_{(0,0),J_\ast} \\ 
  &=\sum_{c,d=0}^{k}\mathsf W_{(0,0),(c,d)}
   \mathcal S_{k,(c,d),(1,0)} =\sum_{a=0}^{k}u_a\,
   \mathcal S_{k,(a,a),(1,0)}
 =\sum_{a=0}^{k}u_a\,S^{(k)}_{a1}S^{(k)}_{a0} \\ 
 &=\frac2p\sum_{a=0}^{p-2}u_a
   \sin\!\left(\frac{2\pi(a+1)}p\right)
   \sin\!\left(\frac{\pi(a+1)}p\right).
\end{split}
\label{eq:vacuum-row-sign-obstruction}
\end{equation}
The second sine is strictly positive, while the first is strictly positive for $a+1<p/2$. Since the coefficients $u_a$ are nonnegative and not all zero,~\eqref{eq:vacuum-row-sign-obstruction} forces $u_a>0$ for some $a+1\geq p/2$. For the corresponding nonzero entry the charges vanish and
\begin{equation}
 x_{I_{\rm vac}}=\frac1{p+1},
 \qquad
 y_{(a,a)}=\frac{a+1}{p}\geq\frac12,
\end{equation}
which proves~\eqref{eq:wall-spread}. If the support of $\mathsf W$ was not charge-preserving to begin with, the first alternative of the proposition already holds. Thus every nonzero solution violates the gapped specialization of~\eqref{eq:refining-wall} unless its degree-zero refinement constant obeys
\begin{equation}
 C_0\geq p\left(\frac12-\frac1{p+1}\right),
\end{equation}
whose right-hand side diverges as $p\to\infty$.
\end{proof}

\paragraph{Scope of the obstruction.}
Only special gapped interfaces preserve the charge. We impose it in~\eqref{eq:refining-wall} because the regular $BF[G]$ limit keeps $\eta$ as an unscaled discrete weight, since a flux displacement of order $k^{-1}$ vanishes in the limit, whereas $\eta_I-\omega_J\neq0$ remains a nontrivial action on the limiting charge sectors. This selection rule would follow from an appropriate symmetry preserved by an interface, and it lies beyond what gappedness or a bound on the interface central charge supplies. A charge-changing gapped interface may therefore be physically admissible, but it would realize a nontrivial $BF[G]$ defect in place of the uniformly refining transport required here. The proposition rules out gapped interfaces only as realizations of this special kind of transport, namely their genus-one data either change the charge or contain a charge-preserving transmission separated by at least $1/2-1/(k+3)$ in the flux coordinate.

Within the scope just stated, the obstruction to a gapped interface satisfying~\eqref{eq:refining-wall} motivates a gapless candidate whose local sectors correlate the line labels on the two sides. For an $SU(2)$ level-$r$ jump, a natural chiral algebra is the diagonal coset
\begin{equation}\label{eq:interface-coset}
 \begin{aligned}
  \mathcal V^{(r)}_k
  &:=\frac{SU(2)_k\times SU(2)_r}{SU(2)_{k+r}}\,,\\
  c\bigl(\mathcal V^{(r)}_k\bigr)
  &=\frac{3k}{k+2}+\frac{3r}{r+2}-\frac{3(k+r)}{k+r+2}\,,
 \end{aligned}
\end{equation}
with the denominator embedded diagonally. This coset construction is standard~\cite{Goddard:1984vk}, and a fermionic construction of interfaces related to this $SU(N)$ coset appears in~\cite{Gaiotto:2026qai}\footnote{In the resolved description of
\cite[Fig.~4 and Eq.~(6)]{Gaiotto:2026qai}, a level-$r$ jump
involves a $G_r$ Chern--Simons segment with one end at the $\mathcal V^{(r)}_k$ coset junction and the other end at an additional
two-dimensional boundary sector $\mathcal A_r/G_r$, where
$\mathcal A_r$ is a matter theory carrying the level-$r$
$G$ current algebra. Shrinking this segment brings the
additional sector into the interface. Here we only focus on the
$\mathcal V^{(r)}_k$ coset junction and its conjugate in the
doubled theory. One may seek a completion in which the
additional boundary sector $\mathcal A_r/G_r$ and its conjugate
are gapped together. We do not establish such a completion here.}.  There the level-one factor is completed by a $U(1)$ current into complex fermions, and the resulting bulk couplings are part of the construction.  Our fixed-rank, large-level doubled interface still requires its own coupling and charge-selection rule. For the unit jump the coset is the unitary Virasoro minimal model of~\cite{Friedan:1983xq,Belavin:1984vu},
\begin{equation}\label{eq:interface-minimal-model}
 \mathcal V^{(1)}_k\;\simeq\;\mathcal M(p,p+1)\,,
 \qquad
 c\bigl(\mathcal V^{(1)}_k\bigr)=1-\frac{6}{(k+2)(k+3)}\,,
\end{equation}
and, writing $L_\ell(a)$ for the integrable affine $SU(2)_\ell$ module with highest-weight label $a$, the affine branching rule reads~\cite[Eq.~(2.20)]{Goddard:1986ee}
\begin{equation}\label{eq:unit-jump-affine-branching}
 L_k(a)\otimes L_1(s)
 =\bigoplus_{\substack{0\leq b\leq k+1\\ a+s+b\equiv0\ ({\rm mod}\ 2)}}
 B^{\,b}_{a,s}\otimes L_{k+1}(b)\,,\qquad s\in\{0,1\}\,,
\end{equation}
up to the usual coset field identification, with $B^{\,b}_{a,s}$ the corresponding $\mathcal M(p,p+1)$ branching module. In this decomposition, $B^{\,b}_{a,s}$ pairs the incoming label $a$ at level $k$ and the auxiliary label $s$ at level one with the outgoing label $b$ at level $k+1$.

To test compatibility with the bulk level jump, we fold the interface by reversing the orientation of the level-$(k+r)$ bulk. The resulting boundary of $SU(2)_k\times SU(2)_{-(k+r)}$ Chern--Simons theory has relative gauge anomaly $k-(k+r)=-r$, so we introduce a boundary $\widehat{su}(2)_r$ current algebra with anomaly $+r$. Its current combines with the level-$k$ current to form $J_{\rm diag}=J^{(k)}+J^{(r)}$ at level $k+r$, matching the gauge symmetry of the folded outgoing bulk. Gauging this diagonal symmetry leaves the coset~\eqref{eq:interface-coset} as the gauge-invariant chiral algebra. Repeating the construction for the negative-level factors and choosing the diagonal modular invariant gives a standalone nonchiral completion.  This anomaly-matching argument identifies a compatible candidate; it does not construct the required interface couplings to both bulks.

In the doubled theory, we write the incoming and outgoing labels as $J=(a_+,a_-) \in\Irr(\mathcal C_k)\times\Irr(\mathcal C_k^{\mathrm{rev}}) \cong \{0,\ldots,k\}^2$ at level $k$ and $I=(b_+,b_-)\in\Irr(\mathcal C_{k+1})\times\Irr(\mathcal C_{k+1}^{\mathrm{rev}}) \cong \{0,\ldots,k+1\}^2$ at level $k+1$, and apply~\eqref{eq:unit-jump-affine-branching} to the two chiralities with auxiliary labels $s_\pm\in\{0,1\}$. To exhibit refining channels, we choose a fixed integer $\delta$, independent of $k$, and set $b_\pm=a_\pm+\delta$ and $s_\pm\equiv\delta\ ({\rm mod}\ 2)$ provided $0\le a_\pm+\delta\le k+1$. Then the branching rule $a_\pm+s_\pm+b_\pm\equiv0\pmod2$ is satisfied and the common shift preserves the charge
\begin{equation}
 \eta_I=b_+-b_-=a_+-a_-=\omega_J
 \qquad\text{in units of }\varpi_1,
\end{equation}
and $|x_I-y_J|\leq(|\delta|+1)/(k+3)$. Thus a fixed $\delta$ moves the line by only a bounded number of microscopic label steps. The same branching rule also permits shifts of order $k$, but the calculation below places them at conformal weights of order $k^2$, so they stay outside every fixed-weight cutoff. Hence the coset contains the required finite-weight near-diagonal sectors. Establishing~\eqref{eq:refining-wall} for a physical interface nevertheless requires its doubled bulk coupling to select charge-preserving left--right combinations in every bounded-weight sector, which is beyond the scope of this work.

We can further relate this displacement to conformal weight without restricting to zero charge. Consider a charge-preserving channel in~\eqref{eq:unit-jump-affine-branching} with common charge $\eta:=\eta_I=\omega_J$. Charge preservation forces the common shift $b_+-a_+=b_--a_-=\delta$ introduced above. Writing $h_\pm$ for the lowest weights of the two chiral coset modules, we have~\cite[Eq.~(1.7)]{Goddard:1986ee}
\begin{equation}\label{eq:coset-weight}
 h_\pm
 =
 \frac{\bigl((p+1)(a_\pm+1)-p(b_\pm+1)\bigr)^2-1}{4p(p+1)}
 =\frac{(a_\pm+1-p\delta)^2-1}{4p(p+1)}\,.
\end{equation}
The definitions of the flux and charge give
\begin{equation}
 a_\pm+1=p\,y_J\pm\frac{\eta}{2},
 \qquad
 x_I-y_J=\frac{\delta-y_J}{p+1}.
\end{equation}
Substituting these identities into~\eqref{eq:coset-weight} yields the exact relations
\begin{equation}\label{eq:weight-vs-displacement}
 \begin{aligned}
 \Delta:=h_++h_-
 &=
 \frac{p(p+1)}2\,(x_I-y_J)^2
 +\frac{\eta^2/4-1}{2p(p+1)},\\
 h_+-h_-
 &=-\frac{\eta}{2}\,(x_I-y_J).
 \end{aligned}
\end{equation}
For fixed $\eta$, the flux part of~\eqref{eq:refining-wall} reads $|x_I-y_J|\leq C_{\Delta_0}/p$ at a fixed cutoff $\Delta_0$, and~\eqref{eq:weight-vs-displacement} gives
\begin{equation}
 \Delta\leq\frac{p+1}{2p}\,C_{\Delta_0}^2
 +\frac{\eta^2/4-1}{2p(p+1)}
 =O(1)\,.
\end{equation}
Conversely, for fixed $\eta$, a uniform bound on $\Delta$ implies $p|x_I-y_J|=O(1)$ and hence the flux part of~\eqref{eq:refining-wall}. On a cylinder of circumference $L$, the corresponding excitation above the vacuum satisfies
\begin{equation}\label{eq:interface-cylinder-energy}
 \frac{L}{2\pi}\bigl(E-E_{\rm vac}\bigr)=\Delta\,.
\end{equation}
Thus, in every fixed-charge channel, bounded displacement in lattice units is equivalent to uniform boundedness of the dimensionless cylinder energy. Moreover,~\eqref{eq:weight-vs-displacement} gives $h_+-h_-=O(p^{-1})$ under refinement, so these sectors become spinless as $k\to\infty$. For $\eta\neq0$ the two finite-level chiral modules are paired off-diagonally and describe candidate junction operators attached to bulk lines, distinct from the intrinsic primaries of the standalone diagonal minimal model.

The charge-preserving alternative of Proposition~\ref{prop:adjacent-wall-nonrefinement} forces a transition with $\eta=0$ and macroscopic flux displacement.  If the same pair of labels is assigned the coset weights in~\eqref{eq:weight-vs-displacement}, its dimension is $\Delta\gtrsim p^2/8$.  Such a coset channel leaves every fixed-weight sector.  A gapped tunneling matrix has no conformal grading with which to make this separation; its failure of refinement is already established directly by the proposition.

Proposition~\ref{prop:adjacent-wall-nonrefinement} constrains only the gapped specialization of~\eqref{eq:refining-wall}. Its hypotheses describe the transmission by a single ungraded nonnegative matrix $\mathsf W$ satisfying the topological $S$- and $\theta$-intertwining conditions. For a gapless interface, $\mathcal H^{\mathfrak J_k}_{IJ}$ retains its conformal grading, so the ungraded matrix $\mathsf W$ alone leaves the transmission underdetermined. In the charge-preserving coset sectors above, fixed shifts $\delta=O(1)$ and fixed charges $\eta=O(1)$ occur at bounded conformal weight, whereas an $O(1)$ flux displacement has $\Delta=O(k^2)$ by~\eqref{eq:weight-vs-displacement}. Refinement~\eqref{eq:refining-wall} is therefore a condition on the uniformly bounded-weight part of this graded transmission, in place of the support obtained after forgetting the grading. The gapless candidates thus lie exactly outside the hypotheses of the proposition.

\subsection{The large-level limit of the unit-jump interface}
\label{sec:runkel-watts}

Let $\mathfrak J^{(1)}_k:\mathsf T_k\to\mathsf T_{k+1}$ denote a prospective elementary unit-jump interface with candidate chiral algebra $\mathcal V^{(1)}_k$. Holding the jump $r$ fixed in~\eqref{eq:interface-coset}, we find that the first and third terms cancel in the limit, so the exact chiral central charge obeys
\begin{equation}\label{eq:interface-central-charge-limit}
 c\bigl(\mathcal V^{(r)}_k\bigr)\;\xrightarrow[k\to\infty]{}\;\frac{3r}{r+2}\,,
 \qquad\text{so}\qquad
 c\bigl(\mathcal V^{(1)}_k\bigr)\;\longrightarrow\;1\,.
\end{equation}
The positive limiting central charge is incompatible with a purely topological limit of this fixed conformal sector.  If this candidate is realized by a physical bulk-coupled interface, conformal degrees of freedom survive on each interface of the chain~\eqref{eq:interface-chain}. For fixed $n$, its large-$K$ limit would then suggest $BF$ regions separated by nontrivial limiting interface theories. We now ask which $c=1$ theory could occupy an elementary interface.

Using~\eqref{eq:interface-minimal-model} we recognize the chiral family as exactly the unitary Virasoro minimal series $\mathcal M(p,p+1)$, taken with the diagonal modular invariant. Runkel and Watts study that same limit and propose a nonrational $c=1$ theory $\mathcal M_\infty$ whose nonchiral primary spectrum is a continuum of spinless representations with weights
\begin{equation}\label{eq:RW-spectrum}
 h=\bar h=\frac{\mathsf x^2}{4}\,,\qquad \mathsf x\in S:=\mathbb R_{>0}\setminus\mathbb Z_{>0}\,,
\end{equation}
with the removed integers labeling modules with null vectors~\cite[Sections~1 and~2.1]{Runkel:2001ng}. The proposal is neither rational nor quasi-rational.  Its later interpretation in terms of the twisted sector of a continuous orbifold provides additional structural support~\cite[Sections~1 and~4]{Gaberdiel:2011aa}; neither construction supplies the bulk couplings required of the interface considered here.

For $G=SU(2)$, we first compare neutral diagonal primaries across levels by selecting sequences in different finite-level Hilbert spaces. For every sufficiently large $k$, choose an incoming label $J=(a_k,a_k)$ at level $k$ and an outgoing label $I=(b_k,b_k)$ at level $k+1$ such that, for a fixed $\delta\in\mathbb Z$,
\begin{equation}\label{eq:RW-fixed-flux-scaling}
 \frac{a_k+1}{k+2}\longrightarrow t\in(0,1),
 \qquad
 b_k-a_k=\delta
 \quad\Longrightarrow\quad
 \frac{b_k+1}{k+3}\longrightarrow t,
\end{equation}
so the incoming and outgoing fluxes approach the same macroscopic alcove point $t$, while $\delta$ encodes their fixed microscopic label shift. Substitution in~\eqref{eq:coset-weight} gives the Runkel--Watts spectrum~\eqref{eq:RW-spectrum}:
\begin{equation}\label{eq:coset-weight-limit}
 h\;\xrightarrow[k\to\infty]{}\;\frac{(t-\delta)^2}{4}\,,
 \qquad\text{that is}\qquad
 \mathsf x=|t-\delta|\,.
\end{equation}
The limiting primary is therefore specified by both the macroscopic flux $t$ and the microscopic shift $\delta$, while fixing $t$ alone leaves $\mathsf x$ undetermined. As $t$ runs over the open alcove and $\delta$ over the integers, $|t-\delta|$ fills $S$, with the usual coset field identification relating redundant labels, and the integers are missed precisely because $t$ stays regular. Holding the weights fixed in place of the flux gives $(a_k+1)/(k+2)\to0$ and returns the removed values $\mathsf x=|\delta|$.

The same limiting weights occur at every fixed charge. Choose a charge-preserving sequence with $y_J\to t$, $\omega_J=a_+-a_-=\eta_I=b_+-b_-=\eta$ fixed, and $b_\pm-a_\pm=\delta$. Equation~\eqref{eq:coset-weight} gives
\begin{equation}\label{eq:RW-fixed-charge-limit}
 h_\pm\longrightarrow\frac{(t-\delta)^2}{4},
 \qquad
 h_+-h_-=\frac{\eta(y_J-\delta)}{2(k+3)}
 \longrightarrow0.
\end{equation}
Thus, along every refining sequence with fixed charge $\eta$, the conformal spin $h_+-h_-$ tends to zero as $k\to\infty$. As $t$ and $\delta$ vary, the charge-preserving branching sequences at every fixed $\eta$ reproduce the same set $S$ of limiting conformal weights. Comparing~\eqref{eq:weight-vs-displacement} with~\eqref{eq:RW-spectrum}, we identify the Runkel--Watts parameter with the flux displacement across the interface measured in mesh units,
\begin{equation}\label{eq:RW-parameter-is-displacement}
 \mathsf x=\lim_{k\to\infty}(k+2)\,|x_I-y_J|\,,
\end{equation}
so, for the branching sequences above, the candidate Runkel--Watts spectral label equals the rescaled flux displacement, and the integers removed from $S$ are reached at the alcove walls. This prescription selects sequences of primaries in place of a canonical identification of the finite-level Hilbert spaces. At nonzero $\eta$, the finite-level pairs in~\eqref{eq:RW-fixed-charge-limit} are candidate bulk-line-attached junction sectors, distinct from the intrinsic fields of the standalone diagonal minimal model. The calculation fixes their limiting weights, and establishing that they occur in the physical bulk-coupled interface, together with convergence of their normalized fields, OPE coefficients, and correlation functions, remains open.

\begin{conjecture}[Large-level limit of the unit-jump interface]\label{conj:runkel-watts}
For $G=SU(2)$, let $\mathfrak J^{(1)}_k$ be a unit-jump interface theory between $\mathsf T_k$ and $\mathsf T_{k+1}$ whose chiral algebra is $\mathcal V^{(1)}_k$ and whose intrinsic nonchiral completion is diagonal. Suppose its neutral diagonal primaries are compared across levels along the sequences specified in~\eqref{eq:RW-fixed-flux-scaling}. We conjecture that, after the continuum averaging and normalization required to form the limiting fields, this sector converges as $k\to\infty$ to the Runkel--Watts theory $\mathcal M_\infty$ of~\cite{Runkel:2001ng}. We further conjecture that the physical bulk-coupled interface contains the fixed-charge junction sectors of~\eqref{eq:RW-fixed-charge-limit}, whose limiting conformal-weight set is $S$ in every fixed regular $BF[SU(2)]$ charge sector. For each fixed $n$, we conjecture that the $K\to\infty$ limit of the finite chain~\eqref{eq:interface-chain} consists of $n+1$ $BF[SU(2)]$ regions separated by $n$ such bulk-enriched two-dimensional interfaces.
\end{conjecture}

\paragraph{Evidence for Conjecture~\ref{conj:runkel-watts}.}
The central charge matches by~\eqref{eq:interface-central-charge-limit}, the chiral family and the modular invariant are the ones Runkel and Watts take, the spectrum produced by~\eqref{eq:RW-fixed-flux-scaling} reproduces $S$ through~\eqref{eq:coset-weight-limit}, fixed charge leaves this limiting spectrum unchanged by~\eqref{eq:RW-fixed-charge-limit}, and by~\eqref{eq:RW-parameter-is-displacement} the continuum label is the flux displacement across the interface.  These are spectral and algebraic matches; they do not establish convergence of normalized fields, OPE coefficients, or correlators.

\paragraph{Nonuniqueness of the limiting prescription.}
A competing option is the degeneration of the same minimal series studied by Roggenkamp and Wendland, whose limit acquires a geometric interpretation on an interval together with a commutative algebra of vanishing-energy zero modes~\cite[Section~4]{Roggenkamp:2003qp}, from which both the limiting metric and the limiting brane geometry can be read off~\cite{Roggenkamp:2008jm}. This possibility is relevant here because the degenerating interval could match the alcove carrying $x$ in~\eqref{eq:regular-lattice}. Runkel and Watts already flag the ambiguity and attribute the alternative to Roggenkamp and Wendland~\cite[Section~1]{Runkel:2001ng}, and the two limits are distinct as the $\mathcal N=2$ minimal series already shows, where distinct scalings of one coset family reach either a free theory or a continuous orbifold of it~\cite{Fredenhagen:2012rb,Fredenhagen:2012bw,Restuccia:2013tba}. They arise from different prescriptions for comparing primaries across levels. Conjecture~\ref{conj:runkel-watts} adopts~\eqref{eq:RW-fixed-flux-scaling}; a construction of the doubled bulk-coupled interface must determine which sequences remain at finite weight relative to the bulk gap. We also note that $\mathcal M_\infty$ was proposed conjecturally in~\cite{Runkel:2001ng}, with analytic and numerical crossing tests, so our conjecture rests on that limiting proposal as well as on the unconstructed bulk coupling.

\paragraph{Scope of the interface analysis.}
The proved ingredients of this section are Proposition~\ref{prop:adjacent-wall-nonrefinement} under its stated $SU(2)$ and genus-one hypotheses, together with the exact coset central-charge, branching, and weight identities.  The existence of a physical doubled bulk-coupled coset interface, the selection of its bounded-weight line-attached sectors, their Runkel--Watts convergence, and the limiting chain remain open as stated in Conjecture~\ref{conj:runkel-watts}.  None of these interface claims is used in the proofs of the bulk modular-convergence theorem.

\section{Discussion and outlook}
\label{sec:discussion}

The large-level comparison established here preserves normalized regular
torus wave packets and their modular transformations.  The average of the
two chiral weights becomes a continuous magnetic coordinate, while their
difference remains an electric charge.  The lattice-cell normalization
and the no-escape estimate are both needed to turn the formal kernel
limit into strong-$*$ convergence.  The boundary calculations then
identify several compatible local limits and expose the dependence on
polarization, zero modes, and the chosen boundary algebra.  The
interface analysis gives a more selective statement: a particular
charge-preserving refinement is obstructed for gapped adjacent-level
$SU(2)$ interfaces, while coset branching supplies channels with the
required displacement and finite conformal weight.

\subsection{Singular sectors, sewing, and observables}

The singular conjugacy classes contain physically indispensable lines.
At the identity, electric charges are representations of $G$ rather
than characters of $T$; on a general wall the relevant representation
group is the enhanced centralizer.  Their absence as positive-measure
subsets of the regular chart does not remove them from the line
category.  A useful next calculation would start with sequences whose
distance from a fixed alcove face is of order $r_k^{-1}$, retain the
coordinates along that face, and resolve the transverse lattice labels
under its centralizer.  One could then ask whether these states define
additional Hilbert-space components, distributional states, or a
different scaling limit.  The compact-double classification already
provides the expected stabilizer data
\cite{Koornwinder:1996uq,Koornwinder:1998xg}.
Matching that data with the finite-level vacuum and pure-electric
sequences would be the first test of an extension beyond the regular
sector.

Higher-genus amplitudes require control of sewing.  At finite level a
cut circle introduces a sum over simple labels and fusion multiplicity
spaces; at the $BF$ endpoint the corresponding decomposition contains
direct integrals and measurable multiplicity fields
\cite{Koornwinder:1998xg}.  Numerical asymptotics of fusion
coefficients would only describe dimensions or densities.  To sew
states one must also compare normalized intertwiners, their inner
products, and the changes of channel.  A four-punctured sphere is the
first place where this demands an $F$-operator limit.  Compatibility
with the pentagon and hexagon relations is then the finite set of local
consistency tests suggested by the Moore--Seiberg construction
\cite{Moore:1988qv,Bakalov:2001}.
Passing those identities through a limit requires bounds on the
intermediate direct integrals, including regions where a sewn
holonomy becomes singular.  This is the analytic step needed before
the genus-one result can enter a higher-genus modular functor.

An operator formulation would make the meaning of categorical
convergence more precise.  One can first choose bounded, smeared line
observables supported in regular flux windows and ask for a common
representation in which their products and adjoints converge.  The
annular or tube algebra is a natural finite-level object to compare
with the conjugation-groupoid algebra at the endpoint
\cite{Neshveyev:2018wbn,Jia:2026vcr}.
Such a comparison must specify the class of observables, the topology,
and the representations used; agreement of $S$ and $T$ alone cannot
supply the missing tensor structure
\cite{Mignard:2017mms}.
If compatible boundary modules and line junctions can also be included,
the comparison could address the center construction itself, rather
than only its genus-one representation.

The present proof also suggests a quantitative question at finite
$k$.  On a fixed smooth packet supported away from the walls, Poisson
summation controls the unequal Weyl pairs to arbitrary inverse powers
of the level.  The remaining errors come from sampling, the shifted
level, and the approximation of the packet by cells.  A systematic
expansion should separate these representation-dependent errors from
corrections to an observable amplitude.  Charles's Fourier-integral
description gives a semiclassical comparison
\cite{Charles:2010ym}.  Uniform estimates in a Sobolev class
would also permit a study of modular words whose lengths grow with
$k$.  Such estimates are additional information: pointwise
strong convergence for each fixed mapping class gives no uniform
control over a level-dependent sequence of mapping classes.

Closed three-manifold amplitudes present a related issue.  A mapping
torus is formally obtained by taking a trace of its mapping-class
operator, but a unitary operator on $\mathcal H_0$ need not have a
trace.  A fixed positive trace-class operator $R$ gives a well-defined
intermediate quantity, and the theorem does imply
\begin{equation}
 \operatorname{Tr}\!\left(
 R^{1/2}\rho_k^\#(\gamma)R^{1/2}\right)
 \longrightarrow
 \operatorname{Tr}\!\left(
 R^{1/2}\rho_0(\gamma)R^{1/2}\right).
 \label{eq:outlook-regulated-trace}
\end{equation}
Indeed, expand the trace in an eigenbasis of $R$ and use strong
convergence for each eigenvector together with the summable
eigenvalues of $R$.  Removing such a regulator, choosing it compatibly
with gluing, and recovering the singular flat-connection contributions
are further problems.  The different prescriptions studied in
\cite{Blau:2022krm} provide explicit amplitudes against
which this extension could be tested.

\subsection{Boundary theories and level-changing interfaces}

For the boundary theory, the local expansion around the identity of
$G$ is only one part of a large-radius limit.  A global construction
must fix the normalization of the noncompact zero modes, the space of
allowed background holonomies, and the way boundary operators are
compared across levels.  The abelian calculation shows why the
boundary algebra must be kept in this comparison: selecting a
momentum or winding sublattice can survive the bulk limit.  For
nonabelian $G$, one could begin with torus wavefunctions in a regular
commuting background and derive their oscillator determinants and
zero-mode measure from the finite-level WZW theory.  This would test
the theta-function expressions independently of modular covariance
\cite{Elitzur:1989nr,Knizhnik:1984nr}.
It would also distinguish the local nonchiral scalar limit, with
central charge $\dim G$ in each chirality, from a chiral WZW model for
a nonsemisimple tangent or cotangent group
\cite{Nappi:1993ie,Figueroa-OFarrill:1995lva,Tseytlin:1995yw}.

Flat gauging requires a prescription for the holonomies at which
additional bosonic zero modes appear.  Their locus can have measure
zero while its neighborhoods give a divergent integral; deleting the
locus is therefore insufficient.  A constructive approach would
introduce a common zero-mode regulator before integrating over
holonomies, determine its modular transformation, and compare the
result with regulated finite-level partition functions.  The
centralizer projection must be imposed on the full oscillator and
zero-mode Hilbert space.  It enforces the appropriate invariance and
level matching, but does not make every surviving primary spinless.
The continuous-orbifold construction of
\cite{Gaberdiel:2011aa} and the candidate-gauging analysis of
\cite{Jia:2026tsl} give concrete settings in which to
study these questions.  In particular, the adjoint free-scalar
candidate of Section~\ref{sec:flat-gauging} has not been identified
with the $c=1$ Runkel--Watts theory.

The global form of the gauged group matters even when its Lie algebra
action is fixed.  An adjoint action is insensitive to the center, so
gauging $G$ and gauging its faithful quotient need not give the same
theory.  The sum over bundles, possible discrete topological weights,
and sectors associated with an ineffective finite subgroup must be
specified.  Finite orbifold operations and continuous abelian
gaugings already display the role of this data
\cite{Dijkgraaf:1989pz,Gaiotto:2020iye,
Antinucci:2024zjp}.
Comparing these choices in the $SU(2)$ and $SO(3)$ examples would test
which global information survives the local free-scalar expansion.
It would also clarify which boundary condition is represented by a
given choice of regular torus measure.

For interfaces, the next step is an action coupled to both doubled
bulks.  Anomaly matching identifies the diagonal coset chiral algebra,
but the bulk coupling must additionally specify which left--right
branching channels occur as line junctions and which conserved
symmetry, if any, enforces charge preservation.  The fermionic
Chern--Simons interfaces of Gaiotto and Zeng
\cite{Gaiotto:2026qai} offer a concrete construction to compare with
this requirement.  A successful doubled construction would determine
the graded transmission spaces
$\mathcal H^{\mathfrak J_k}_{IJ}$ and test refinement at a fixed
conformal-weight cutoff.  It would also distinguish an interface that
approaches identity transport from a charge-changing interface that
survives as a nontrivial $BF$ defect.  The gapped-interface obstruction
in Section~\ref{sec:gapped-walls} concerns the former question.

The proposed Runkel--Watts limit should next be tested on normalized
fields.  In the neutral diagonal channels, one can average the
finite-level coset primaries over a small interval of the limiting
spectral parameter, fix their two-point normalization, and compare
their three-point coefficients with those of
\cite{Runkel:2001ng}.  At fixed nonzero electric charge, the same
procedure must be carried out for fields attached to bulk lines.
Their coefficients may retain charge information even though their
limiting conformal weights agree with the neutral spectrum.
Other limits of minimal-model families show that the identification of
fields across levels is essential
\cite{Roggenkamp:2003qp,Fredenhagen:2012rb,Fredenhagen:2012bw}.
The exact relation between displacement and conformal weight gives
a physical criterion for selecting sequences, but the bulk coupling
must decide whether those sequences are present.

Finally, interfaces can be composed.  Bringing two elementary
level-changing interfaces together introduces both sums over bulk
line labels and a limit of two-dimensional defect theories.
It would be useful to determine whether this operation reproduces
the fixed-jump diagonal coset, requires an additional renormalization
group flow, or leaves further junction degrees of freedom, see e.g.~\cite{Gaiotto:2012np,Stanishkov:2016pvi,Chen:2026byx}.
The finite categorical treatment of domain walls
\cite{Kitaev:2011dxc,Fuchs:2012dt} supplies the gapped comparison,
while the graded coset transmission contains information that an
ungraded tunneling matrix forgets.  The order of interface fusion
and the large-level limit is itself an observable question.

\subsection{Global forms, anomalies, and three-dimensional gravity}

The restriction to simply connected $G$ removes bundle sectors that
would be present for $G/\Gamma$, with $\Gamma$ a subgroup of the center.
On a torus, commuting quotient-group holonomies can lift to elements
whose commutator is a nontrivial element of $\Gamma$.  Such components
are not described by the single chart $(T\times T)/W$ used here.
An extension should include these components, the allowed
Chern--Simons levels, and any discrete topological choices before
constructing its sampling maps.  This would give a direct test of
whether a finite one-form gauging commutes with the large-level
limit.  The cohomological description of Chern--Simons actions and the
higher-form-symmetry formulation provide the relevant global data
\cite{Dijkgraaf:1989pz,Gaiotto:2014kfa}.

A fixed residual anomaly is accessible already at the classical
level.  Replace the levels $(k,-k)$ by $(k+m,-k)$ with $m$ fixed and
use the same near-diagonal fields.  On a closed manifold and for
fixed smooth $A,B$, the additional term is
$m\,\operatorname{CS}(A+B/(2k))/(4\pi)$, so
\begin{equation}
 \frac{k+m}{4\pi}\operatorname{CS}(A_+)
 -\frac{k}{4\pi}\operatorname{CS}(A_-)
 \longrightarrow
 -\frac{1}{2\pi}\int_M\operatorname{tr}(B\wedge F_A)
 +\frac{m}{4\pi}\operatorname{CS}(A).
 \label{eq:outlook-residual-level}
\end{equation}
The expected endpoint is thus $BF+m\,CS$, with the relative lifts
and global level conditions understood.  Its defect Hilbert spaces
and twisted conjugation-groupoid algebra are studied in
\cite{Jia:2026vcr,Jia:2026bbo}, and a non-compact non-semisimple version has also been studied in~\cite{Balasubramanian:2025rcr, Balasubramanian:2026ymu, Balasubramanian:2026vlr}.
The quantum problem is to modify the two sampling lattices and retain
the anomaly phases in the modular action.  This would compare a
family of finite-level anomalous boundary symmetries with a fixed
continuous anomaly, rather than sending every level to zero at the
endpoint.

For $SU(2)$, a comparison with three-dimensional gravity should retain
both the gravitational observables and the regulator.  The cubic term
in Eq.~\eqref{eq:intro-action-contraction} is of order $k^{-2}$,
consistent with the scale at which a cosmological interaction enters
the quantum-group regularization
\cite{Mizoguchi:1991hk}.  A useful test would insert particles
with fixed regular holonomy and compare their normalized transition
amplitudes before attempting the partition function of a closed
manifold.  Quantum-double particle labels and gauge-fixed
Ponzano--Regge observables are available for this purpose
\cite{Bais:1998yn,Freidel:2004vi}.
Extending the comparison to a state sum requires asymptotics of the
$q$-$6j$ symbols, control of internal label sums, and treatment of
reducible flat connections.  The regular torus theorem supplies none
of these automatically.  The gauge-fixing and torsion analysis of
\cite{Barrett:2008wh} provides a way to formulate the
endpoint amplitudes before asking whether the Turaev--Viro regulator
converges to them.

A lattice realization would permit a different test of the same
limit.  String-net Hamiltonians built from $\mathcal C_k(G)$ give
finite local label spaces at each $k$
\cite{Levin:2004mi}.  One could compare smeared ground-state
amplitudes and ribbon observables as the local label cutoff grows,
using the same flux--charge normalization as on the torus.
The limiting plaquette flatness constraint is distributional, so its
normalization and its action on states require explicit control.
Convergence of the torus modular operators alone says nothing about
a uniform many-body spectral gap or convergence of the local
Hamiltonians.  Establishing even a fixed-lattice comparison would
explain which features of continuous $BF$ symmetry can be obtained
from finite local Hilbert spaces and which require an additional
limit.

\section*{Acknowledgement}

SG is supported by National Natural Science Foundation of China under Grant No. 126B1001. QJ is supported by National Research Foundation of Korea (NRF) Grant No. RS-2024-00405629 and Jang Young-Sil Fellow Program at the Korea Advanced Institute of Science and Technology. JT is supported by National Natural Science Foundation of China under Grant No. 12405085 and by the Natural Science Foundation of Shanghai (Grant No. 24ZR1419300). JT would like to thank Ying Zhang for her love and support. YNW is supported by National Natural Science Foundation of China under Grant No. 12422503.

\appendix

\section{Endpoint modular action and the regular holonomy chart}
\label{app:endpoint-modular-action}

This appendix proves the endpoint properties stated after
Eqs.~\eqref{eq:BF-S-operator}--\eqref{eq:BF-T-operator} and used in
section~\ref{sec:Modular_Limit}.  The construction is independent of
the finite-level convergence argument.  We first realize the generators
as unitary operators in the flux--charge representation, and then relate
them to the holonomy action and its common invariant conull domain.

\subsection{Fourier realization and unitary extension}
\label{app:endpoint-Fourier-unitarity}

The endpoint $S$ kernel integrates over flux and sums over arbitrarily
large charges.  Its Fourier realization makes norm preservation explicit:
the generators permute Fourier characters, so unitarity and the modular
relations can be checked without multiplying integral kernels.
In the flux--charge representation the Dehn twist is the
Aharonov--Bohm phase $\chi_\eta(g_x)=e^{2\pi i(\eta,x)}$.

\begin{proposition}[Fourier realization and endpoint modular action]
\label{prop:regular-Fourier-unitarity}
Under the last unitary in
Eq.~\eqref{eq:regular-BF-Hilbert-chain-early}, the operator $\mathcal S_0$ is
the restriction of the Pontryagin transform
\begin{equation}
 (\mathcal F F)(x,\eta)
 =\sum_{\omega\in P}\int_T
 e^{-2\pi i((x,\omega)+(\eta,y))}F(y,\omega)\,dm_T(y).
 \label{eq:torus-Pontryagin-transform}
\end{equation}
Consequently $\mathcal S_0$ extends uniquely to a unitary on $\mathcal H_0$.
Together with $\mathcal T_0$, it defines a unitary representation
$\rho_0$ of $SL(2,\mathbb Z)$.
\end{proposition}

\begin{proof}
Write $e_\lambda(x)=e^{2\pi i(\lambda,x)}$.  On the Fourier basis
$e_\lambda\delta_\nu$ of $L^2(T\times P)$,
$\mathcal F(e_\lambda\delta_\nu)=e_{-\nu}\delta_\lambda$.  Indeed, only
the summand $\omega=\nu$ survives, and character orthogonality gives
$\int_T e^{2\pi i(\lambda-\eta,y)}dm_T(y)=\delta_{\lambda,\eta}$.
Thus $\mathcal F$ permutes an orthonormal basis and extends to a unitary.
The permutation commutes with simultaneous Weyl transformations of the
labels, so the invariant subspace is preserved.

The normalization of the alcove kernel can also be checked explicitly.
For $x\in Y$, decompose $T_{\rm reg}$ into $wY$ and put
$y=wt$, $\omega=w\zeta$ in each term.  This gives
\begin{align*}
 (\mathcal F\operatorname{Ext}f)(x,\eta)
 &=\gamma_G\sum_{w\in W}\sum_{\zeta\in P}\int_Y
 e^{-2\pi i((w^{-1}x,\zeta)+(w^{-1}\eta,t))}
 f(t,\zeta)\,dt\\
 &=(\mathcal S_0f)(x,\eta).
\end{align*}
Replacing $w^{-1}$ by $w$ gives exactly
Eq.~\eqref{eq:BF-S-operator}, with no extra factor of $|W|$.
The calculation is justified first for the finite-charge test space; the
finite Weyl sum and the integral are then absolutely convergent.
Weyl invariance extends the identity to all regular $x$, so
$\mathcal F\operatorname{Ext}f=\operatorname{Ext}\mathcal S_0f$
on the test space.  With
$U=(|W|\gamma_G)^{-1/2}\operatorname{Ext}$ from
Proposition~\ref{prop:regular-BF-Hilbert-model}, this says that
$\mathcal S_0$ agrees there with
$U^{-1}\mathcal F|_{L^2(T\times P)^W}U$.
The latter operator is unitary on all of $\mathcal H_0$, and the test
space is dense.  It therefore gives the unique bounded extension of
$\mathcal S_0$, which is in particular unitary.  The intertwining identity
then holds on all of $\mathcal H_0$.

For the group action, use the generators
\[
 S=\begin{pmatrix}0&-1\\1&0\end{pmatrix},\qquad
 T=\begin{pmatrix}1&1\\0&1\end{pmatrix}
\]
of $SL(2,\mathbb Z)$.
For $A=\left(\begin{smallmatrix}a&b\\c&d\end{smallmatrix}\right)
\in SL(2,\mathbb Z)$, define
\begin{equation}
 \widetilde\rho_0(A)(e_\lambda\delta_\nu)
 =e_{a\lambda+b\nu}\delta_{c\lambda+d\nu}.
 \label{eq:natural-SL2-action}
\end{equation}
This permutes an orthonormal basis and commutes with the diagonal Weyl
action.  To fix the representation convention, regard $(\lambda,\nu)^t$ as
a column vector.  Then
\begin{equation}
 \widetilde\rho_0(A)\widetilde\rho_0(B)
   (e_\lambda\delta_\nu)
 =\widetilde\rho_0(AB)(e_\lambda\delta_\nu),
 \label{eq:endpoint-group-law-check}
\end{equation}
because the rightmost operator first replaces $(\lambda,\nu)^t$ by
$B(\lambda,\nu)^t$ and the left operator then applies $A$.
Hence the formula defines a representation.
Transporting the restriction of $\widetilde\rho_0$ through
$U=(|W|\gamma_G)^{-1/2}\operatorname{Ext}$ gives the unitary representation
\[
 \rho_0:SL(2,\mathbb Z)\longrightarrow U(\mathcal H_0),\qquad
 \rho_0(S)=\mathcal S_0,\quad \rho_0(T)=\mathcal T_0 .
\]
For the first generator, Eq.~\eqref{eq:torus-Pontryagin-transform} gives
$\widetilde\rho_0(S)=\mathcal F$.  For the second, the basis formula reads
\[
 \widetilde\rho_0(T)(e_\lambda\delta_\nu)
 =e_{\lambda+\nu}\delta_\nu
 =e^{2\pi i(\nu,\,\cdot\,)}e_\lambda\delta_\nu.
\]
Consequently
$(\widetilde\rho_0(T)Uf)(x,\eta)
=e^{2\pi i(\eta,x)}(Uf)(x,\eta)=U\mathcal T_0f(x,\eta)$, which proves the
two displayed generator identities.
In particular $\mathcal S_0^4=1$ and
$(\mathcal S_0\mathcal T_0)^3=\mathcal S_0^2$.
There is also a direct check of the central element: on $T\times P$,
$\mathcal F^2F(x,\eta)=F(-x,-\eta)$.  Since $-w_0$ preserves $Y$,
restriction of the Weyl-invariant extension gives
\begin{equation}
 (\mathcal S_0^2 f)(x,\eta)=f(-w_0x,-w_0\eta).
 \label{eq:endpoint-charge-conjugation}
\end{equation}
This is the flux--charge form of finite charge conjugation
$(a,b)\mapsto(a^*,b^*)$: the identity $-w_0\rho=\rho$ ensures that the
shifted finite coordinates transform in exactly the same way.

\end{proof}

\subsection{Holonomy representation and a common conull domain}
\label{app:regular-holonomy-action}

The representation just constructed is unitarily equivalent to the
pullback action on holonomy wavefunctions.  Put
$\mathcal H_{\rm hol}=L^2(T\times T,dm_T\otimes dm_T)^W$.
For $A=\left(\begin{smallmatrix}a&b\\c&d\end{smallmatrix}\right)
\in SL(2,\mathbb Z)$, define
\begin{equation}
 \Phi_A(x,y)=(ax+cy,bx+dy),\qquad
 (\rho_{\rm hol}(A)\Psi)(x,y)=\Psi\bigl(\Phi_A(x,y)\bigr).
 \label{eq:holonomy-modular-pullback}
\end{equation}
Each $\Phi_A$ is a torus automorphism, so it preserves normalized product
Haar measure.  It also commutes with the diagonal Weyl action.
With this convention $\Phi_A\circ\Phi_B=\Phi_{BA}$, and consequently
$\rho_{\rm hol}(A)\rho_{\rm hol}(B)=\rho_{\rm hol}(AB)$.  Thus $\rho_{\rm hol}$ is a unitary representation on
$\mathcal H_{\rm hol}$.

On the full holonomy space, before restricting to Weyl invariants,
$\rho_{\rm hol}(A)$ sends $e_\lambda(x)e_\nu(y)$ to
$e_{a\lambda+b\nu}(x)e_{c\lambda+d\nu}(y)$.
The partial Fourier transform $\mathscr P_2$ of
Eq.~\eqref{eq:partial-holonomy-Fourier} therefore intertwines this action
with Eq.~\eqref{eq:natural-SL2-action}.
In particular, with $U=(|W|\gamma_G)^{-1/2}\operatorname{Ext}$ from
Proposition~\ref{prop:regular-BF-Hilbert-model}, the unitary map
$\mathscr V=U^{-1}\mathscr P_2:\mathcal H_{\rm hol}\to\mathcal H_0$
satisfies
\begin{equation}
 \begin{gathered}
 (\mathscr V\Psi)(x,\eta)
 =\sqrt{|W|\gamma_G}\int_T
   \Psi(x,y)e^{-2\pi i(\eta,y)}\,dm_T(y),\qquad x\in Y,\\
 \rho_0(A)=\mathscr V\rho_{\rm hol}(A)\mathscr V^{-1}.
 \end{gathered}
 \label{eq:holonomy-flux-charge-intertwiner}
\end{equation}
The coefficient formula holds for almost every $x$ and defines the map
in $L^2$.  Thus the holonomy and flux--charge formulas describe the same
modular representation in different realizations.

The condition that the first holonomy is regular singles out a convenient
coordinate chart, but a modular transformation changes the chosen cycle.
A given regular pair can therefore leave that chart.  For normalizable
states this causes no loss of information: the exceptional pairs for any
fixed transformation have zero measure.  Countability of the modular
group allows one to discard all these exceptional sets at once.

\begin{remark}[A common conull domain for the modular action]
\label{rem:regular-chart-ae-modular}
Put $R=T_{\rm reg}\times T$, and use $\Phi_A$ from
Eq.~\eqref{eq:holonomy-modular-pullback}.
The first component of $\Phi_A(x,y)$ is singular precisely when
\[
 a\alpha(x)+c\alpha(y)\in\mathbb Z
\]
for some root $\alpha$.  Since $(a,c)$ is primitive, the character
$(x,y)\mapsto e^{2\pi i(a\alpha(x)+c\alpha(y))}$ is nontrivial.
Its kernel has Haar measure zero in the connected torus $T\times T$.
There are finitely many roots, so $\Phi_A^{-1}(R)$ is conull for each
fixed $A$.  Since $SL(2,\mathbb Z)$ is countable,
\begin{equation}
 D=\bigcap_{A\in SL(2,\mathbb Z)}\Phi_A^{-1}(R)
 \label{eq:common-modular-conull-domain}
\end{equation}
is also conull.  If $p\in D$, then for every $A,B\in SL(2,\mathbb Z)$,
$\Phi_A(\Phi_B(p))=\Phi_{BA}(p)\in R$; hence $\Phi_B(p)\in D$.
Applying the same argument to $B^{-1}$ gives $\Phi_B(D)=D$.
The set $D$ is also invariant under the diagonal Weyl action and
therefore descends to a common invariant conull subset of
$(T\times T)/W$.

Because removing a null set does not change an $L^2$ space, the
holonomy pullbacks define unitary operators on the quotient--Haar
space in Eq.~\eqref{eq:regular-BF-Hilbert-chain-early}; their
flux--charge realization is $\rho_0$ by
Eq.~\eqref{eq:holonomy-flux-charge-intertwiner}.
This is an almost-everywhere statement, not pointwise invariance of
the whole chart: for $x\in T_{\rm reg}$,
$\Phi_S(x,0)=(0,-x)\notin R$.
Even $T_{\rm reg}\times T_{\rm reg}$ is not invariant under the full
modular group, since $\Phi_T(x,-x)=(x,0)$.
\end{remark}

The conull-domain statement concerns the endpoint representation.
It does not supply a bound on concentration near singular holonomies
in a finite-level limit; that bound is provided by the no-escape
argument in Theorem~\ref{thm:no-escape}.

\bibliographystyle{JHEP}
\bibliography{biblio}

\end{document}